\documentclass[10pt]{amsart}

\usepackage{hyperref}

\usepackage{graphicx, enumerate, url}
\usepackage{amssymb,mathtools,amsthm}
\usepackage{algorithm}
\usepackage{algpseudocode}
\usepackage{color}
\usepackage{tikz}
\usetikzlibrary{3d,calc}
\usepackage{bm,bbm}

\numberwithin{equation}{section}
\numberwithin{algorithm}{section}

\theoremstyle{plain}
\newtheorem{theorem}{Theorem}[section]
\newtheorem{proposition}[theorem]{Proposition}
\newtheorem{lemma}[theorem]{Lemma}
\newtheorem{corollary}[theorem]{Corollary}

\theoremstyle{definition}

\theoremstyle{remark}
\newtheorem{remark}[theorem]{Remark}

\DeclareMathOperator{\tr}{Tr}
\DeclareMathOperator{\diag}{diag}
\DeclareMathOperator{\rank}{rank}

\DeclareMathOperator{\range}{\mathrm{range}}

\newcommand{\N}{\mathbb{N}}
\newcommand{\C}{\mathbb{C}}
\newcommand{\I}{\mathbb{I}}

\newcommand{\R}{\mathbb{R}}

\newcommand{\s}{\mathbf{s}}

\newcommand{\x}{\mathbf{x}}

\newcommand{\y}{\mathbf{y}}
\newcommand{\z}{\mathbf{z}}

\newcommand{\cD}{\mathcal{D}}

\newcommand{\cH}{\mathcal{H}}

\newcommand{\ba}{\mathbf{a}}
\newcommand{\bb}{\mathbf{b}}
\newcommand{\bc}{\mathbf{c}}
\newcommand{\bd}{\mathbf{d}}
\newcommand{\be}{\mathbf{e}}
\newcommand{\bbf}{\mathbf{f}}

\newcommand{\bi}{\mathbf{i}}
\newcommand{\bp}{\mathbf{p}}

\newcommand{\bs}{\mathbf{s}}
\newcommand{\bt}{\mathbf{t}}

\newcommand{\bu}{\mathbf{u}}
\newcommand{\bv}{\mathbf{v}}
\newcommand{\bw}{\mathbf{w}}
\newcommand{\bx}{\mathbf{x}}
\newcommand{\by}{\mathbf{y}}

\newcommand{\rA}{\mathrm{A}}
\newcommand{\rB}{\mathrm{B}}

\newcommand{\rF}{\mathrm{F}}
\newcommand{\rH}{\mathrm{H}}
\newcommand{\rI}{\mathrm{I}}

\newcommand{\rM}{\mathrm{M}}
\newcommand{\rP}{\mathrm{P}}

\newcommand{\rS}{\mathrm{S}}
\newcommand{\rT}{\mathrm{T}}
\newcommand{\rU}{\mathrm{U}}

\newcommand{\rW}{\mathrm{W}}

\newcommand{\0}{\mathbf{0}}
\newcommand{\1}{\mathbf{1}}

\newcounter{mnotecount}[section]
\renewcommand{\themnotecount}{\thesection.\arabic{mnotecount}}
\newcommand{\mnote}[1]
{\protect{\stepcounter{mnotecount}}$^{\mbox{\footnotesize
$
\bullet$\themnotecount}}$ \marginpar{
\raggedright\tiny\em
$\!\!\!\!\!\!\,\bullet$\themnotecount: #1} }

\definecolor{dg}{rgb}{0,.5,0} 

\newcommand{\cv}{= \vcentcolon}

\newcommand{\wind}{\mathcal{W}^Q_{C}}
\newcommand{\windd}{\mathcal{W}^Q_{C}}
\newcommand{\windQ}{\mathcal{W}^Q_{C^Q}}

\begin{document}
\title{On Quantum Optimal Transport}
\author{Sam Cole}
\address{Department of Mathematics, University of Missouri, Columbia, Missouri, 65211, USA}
\email{s.cole@missouri.edu}
\author{Micha\l $\,$ Eckstein}
\address{Institute of Theoretical Physics, Jagiellonian University, Krakow, Poland}
\email{michal.eckstein@uj.edu.pl}
\author[Shmuel Friedland]{Shmuel~Friedland}
\address{Department of Mathematics, Statistics, and Computer Science, University of Illinois at Chicago,  Chicago, Illinois, 60607-7045, USA }
\email{friedlan@uic.edu}
\author[K.~{\.Z}yczkowski]{Karol~{\.Z}yczkowski}
\address{Institute of Theoretical  Physics, Jagiellonian University, Krakow, Poland}
\address{Center for Theoretical Physics, Polish Academy of Science, Warsaw, Poland}
\email{karol.zyczkowski@uj.edu.pl}

\date{July 5, 2022}
\begin{abstract}  
We analyze 
a quantum version of the Monge--Kantorovich optimal transport problem.
The quantum transport cost related to a Hermitian cost matrix $C$ is minimized over the set of all bipartite coupling  states
 $\rho^{AB}$ with fixed reduced density matrices $\rho^A$ and $\rho^B$ of  size $m$ and $n$.  The minimum quantum optimal transport 
 cost
 $\rT^Q_{C}(\rho^A,\rho^B)$ can be efficiently computed using semidefinite programming.
 In the case $m=n$ the cost  $\rT^Q_{C}$ 
 gives a semidistance if and only if $C$ is positive semidefinite and vanishes exactly on the subspace of symmetric matrices.  
 Furthermore, if $C$ satisfies the above conditions, then $\sqrt{\rT^Q_{C}}$  induces a quantum analogue of the Wasserstein-2 distance. 
  Taking the quantum cost matrix $C^Q$  
 to be the projector  on the antisymmetric subspace, 
  we provide  a semi-analytic 
expression for $\rT^Q_{C^Q}$  for any pair of single-qubit states
 and show that
its square root 
yields a transport distance on the Bloch ball.
 Numerical simulations 
suggest that this property 
holds also in higher dimensions. 
Assuming that the cost matrix suffers decoherence and that the density matrices 
become diagonal, we
study the quantum-to-classical transition of the Earth mover's distance,
propose a continuous family of interpolating distances,
and demonstrate 
 that the quantum transport is cheaper than the classical one.
Furthermore, we introduce a related quantity --- the SWAP-fidelity --- and compare its properties with the standard Uhlmann--Jozsa fidelity.
We also discuss the quantum optimal transport for general $d$-partite systems.
\end{abstract}
\maketitle

\maketitle

\tableofcontents

\section{Introduction} \label{sec:intro}
\subsection{Classical and quantum optimal transport and related metrics}
Let us recall the discrete optimal transport problem \cite{Vil09}, 
 as stated in Hitchcock \cite{Hit41} and Kantorovich \cite{Kan60} (prepared in 1939), which is a variation of the classical transport problem initiated by Monge \cite{Mon81}.  Suppose we have $m$ factories producing an amount $G$ of the same product that has to be dispatched  to $n$ customers.  Assume that $x_{ij}^{AB}$ is the proportion of the goods sent from the factory $i$ to consumer $j$.  
Then $x_i^A$ and $x_j^B$ are the proportions of the goods produced by factory $i$ and received by consumer $j$ respectively:
\begin{equation}\label{ccoupl}
x_i^A=\sum_{j=1}^n x_{ij}^{AB},\; i\in[m],\quad x_j^B=\sum_{i=1}^m x_{ij}^{AB}, \; j\in[n],
\end{equation}
where $[m]=\{1,2,\ldots,m\}$.
It is convenient to introduce the random variables $X^A,X^B$  such that
\begin{equation*}
x_i^A=\mathbb{P}(X^A=i),\;i\in[m], \quad x_j^B=\mathbb{P}(X^B=j),\;j\in[n].
\end{equation*}
Then the nonnegative matrix $X^{AB}=[x_{ij}^{AB}]\in\R_+^{m\times n}$ satisfying the above equalities is the joint distribution of the random variable $X^{AB}$: $x_{ij}^{AB}=\mathbb{P}\big(X^{AB}=(i,j)\big)$.  The random variable $X^{AB}$, or the matrix $X^{AB}$, is called a coupling of $X^A$ and $X^B$.  Let $\bx^A=(x_1^A,\ldots,x_m^A)^\top, \bx^B=(x_1^B,\ldots,x_n^B)^\top$ be the probability vectors corresponding to $X^A$ and $X^B$ respectively.
The set of all coupling matrices $X^{AB}$  corresponding to $\bx^A, \bx^B$ is denoted by $\Gamma^{cl}(\bx^A, \bx^B)$.  Note that $X=\bx^A(\bx^B)^\top$, corresponding to the independent coupling of $X^A$ and $X^B$, is in $\Gamma^{cl}(\bx^A, \bx^B)$.  Let $C=[c_{ij}]\in \R^{m\times n}_+$ be a nonnegative matrix where $c_{ij}$ is the transport cost of a unit of goods from the factory $i$ to the consumer $j$.  The classical optimal transport problem, abbreviated as OT, is
\begin{equation}\label{COTP}
\rT_{C}^{cl}(\bx^A,\bx^B)=\min_{ X\in \Gamma^{cl}(\bx^A,\bx^B)}\tr C X ^\top,
\end{equation}
where $\tr $ denotes the trace of a square matrix, and $X^\top$ the transpose of $X$.
The optimal transport problem is a linear programming problem (LP) which can be solved in polynomial time in the size of the inputs $\bx^A,\bx^B$ and  
the matrix $C$ \cite{CCPS}. 

Assume now that $m=n$.  Let $C=[c_{ij}]\in \R^{n\times n}_+$ be a symmetric nonnegative matrix with zero diagonal and positive off-diagonal entries such that $c_{ij}$ induces a distance on $[n]$: dist$(i,j)=c_{ij}$.  That is, in addition to the above conditions one has the triangle inequality $c_{ij}\le c_{ik}+c_{kj}$ for $i,j,k\in[n]$.  
For $p>0 $ denote $(C^{\circ})^p=[c_{ij}^p]\in\R_+^{n\times n}$.
Then the quantity
\begin{equation}\label{Wasserp}
\rW_{C,p}^{cl}(\bx^A,\bx^B)=\big(\rT_{(C^{\circ})^p}^{cl}(\bx^A,\bx^b)\big)^{1/p}, \quad p\ge 1
\end{equation}
is the Wasserstein-$p$ distance on the simplex of probability vectors, $\Pi_n\subset \R_+^n$. This follows from the continuous version of the Wasserstein-$p$ distance, as in \cite{Vas69}.  See \cite{Cut13} for $p=1$. It turns out that  $\rT_{C}^{cl}(\bx^A,\bx^B)$ has many recent applications in machine learning \cite{AWR17,ACB17,LG15,MJ15,SL11}, statistics  \cite{BGKL17,FCCR18,PZ16,SR04} and computer vision \cite{BPPH11,RTG00,SGPCBNDG}.

Several attempts to generalize the notion of the Monge--Kantorovich
distance in quantum information theory are known.
An early contribution defines the distance between any
 two quantum states  by the Monge distance between the corresponding Husimi
 functions \cite{ZS98,ZS01}.   
As this approach depends on the choice of the set of coherent states,
other efforts were undertaken  \cite{AF17,CHW19,PT21,GMP16,GP18}
to introduce  the transport distance between quantum states 
by applying the Kantorovich--Wasserstein optimization
over the set of bipartite quantum states with fixed marginals.
However, none of the proposed `distances' satisfy all of the properties of a genuine distance.
 Even though the matrix transport problem
 has been often investigated in the recent literature 
\cite{BGJ19,BV01,CGGT17,Duv20,Fri20,FECZ,FGZ19},
 this aim has not been fully achieved until now \cite{Ikeda20,Rie18,YZYY19}.
 
Quantum optimal transport has found a number of applications in quantum physics: the measure of proximity of quantum states \cite{CGP20,CM20,DR20,PMTL20},  quantum metrology \cite{BC94,LYLW20,Sa17}, and quantum machine learning \cite{QML,CHW19,KdPMLL21,LW18}.
The aim of this work is to present a constructive solution for
the minimal transport cost 
 in the finite-dimensional quantum setting.
Here we provide the complete account of our mathematical study related to quantum optimal transport.   
We show that, for the case of two arbitrary single-qubit quantum states and 
a projection cost matrix $C^Q$, the root of the optimal transport cost gives the Wasserstein-2 distance.
The physical aspects of quantum optimal transport, along with its potential applications in quantum information processing, are discussed in the companion paper \cite{FECZ}.
 
  Denote by $\Omega_m$ the convex set of density matrices, i.e., the set of $m\times m$ Hermitian positive semidefinite matrices of trace one.
Let $\rho^{A}\in\Omega_m$ and $\rho^B\in \Omega_n$.  A quantum coupling of $\rho^A,\rho^B$ is a density matrix $\rho^{AB}\in\Omega_{mn}$, whose partial traces give $\rho^A, \rho^B$ respectively: $\tr_B\rho^{AB}=\rho^A$ and $\tr_A\rho^{AB}=\rho^B$.  The set of all quantum couplings of $\rho^{AB}$ is denoted by $\Gamma^{Q}(\rho^A, \rho^B)$.  Observe that we always have $\rho^A\otimes\rho^B\in \Gamma^{Q}(\rho^A, \rho^B)$.
Let $C$ be a given positive Hermitian matrix of order $mn$.  The {\em quantum optimal transport} problem,  
 abbreviated as QOT, is defined as follows:
\begin{equation}\label{defkapCAB}
\rT^Q_{C}(\rho^A,\rho^B)=\min_{\rho^{AB}\in \Gamma^{Q}(\rho^A,\rho^B)} \tr  C\rho^{AB}.
\end{equation}
The matrix $C$ can be viewed as a ``cost matrix'' in certain instances that will be explained later.  
%

The quantum optimal transport has a simple operational interpretation. Suppose that Alice and Bob represent two parties who share a bipartite state $\rho^{AB}$. Their local detection statistics are fixed by the marginals $\rho^A = \tr_B\rho^{AB}$ and $\rho^B = \tr_A\rho^{AB}$. If $C$ is an effect, i.e. $0 \leq C \leq 1$, then $\rT^Q_{C}(\rho^A,\rho^B)$ is the minimum probability of observing $C$ with fixed local states $\rho^A$, $\rho^B$. If $C$ is just positive semidefinite, then $\rT^Q_{C}(\rho^A,\rho^B)$ is the minimum expected value of the observable $C$. For more details on the physical interpretation and applications we refer the reader to the companion paper \cite{FECZ} and references therein.

%

Observe that finding the value of $\rT^Q_{C}(\rho^A,\rho^B)$ is a semidefinite programming problem (SDP).  Using standard complexity results for SDP, as in  \cite[Theorem~5.1]{VB96},
we show that the complexity of finding
the value of $\rT^Q_{C}(\rho^A,\rho^B)$ within a given precision $\varepsilon>0$ is polynomial in the size of the given data and $\log\frac{1}{\varepsilon}$.
There are quantum algorithms that offer a speedup for SDP \cite{brandao2017quantum}.

One of the goals of this work is to analyze the properties of $\rT_{C}^Q(\rho^A,\rho^B)$.  It is useful to compare $\rT_{C}^Q$ with 
its classical counterpart $\rT_{C^{cl}}^{cl}$ defined as follows.  Observe that the diagonal entries of $\rho^A$ and $\rho^B$ form two
probability vectors $\bp^A$ and $\bp^B$.  Physically, 
the suppression of the off-diagonal terms in a density matrix of a system corresponds the process of decoherence caused by the interaction of the system with some environment.
For $\bx\in\R^n,X\in\R^{n\times n}$ denote by $\diag(\bx),\diag(X)\in\R^{n\times n}$ the diagonal matrices induced by the entries  of $\bx$ 
and the diagonal entries of $X$,  respectively.
For $\bp^A\in\Pi_m,\bp^B\in\Pi_n$ denote by
$\Gamma_{de}^Q(\diag(\bp^A),\diag(\bp^B))$
the convex subset of diagonal matrices in $\Gamma^Q(\diag(\bp^A),\diag(\bp^B))$.
We show that $\Gamma_{de}^Q(\diag(\bp^A),\diag(\bp^B))$ is isomorphic to 
the set $\Gamma^{cl}(\bp^A,\bp^A)$
of classical coupling matrices.
Let now $C^{cl} \in\R^{m\times n}$ be the matrix induced by the diagonal entries of $C$ (see Section \ref{sec:diagdm}).
  Then 
\begin{equation}\label{qtrancheap}
\rT_{C}^Q(\diag(\bp^A),\diag(\bp^B))\le \rT_{C^{cl}}^{cl}(\bp^A,\bp^B) \; \textrm{ for } \;
\bp^A\in\Pi_m, \bp^B\in\Pi_n.
\end{equation}
We give examples where strict inequality holds.  Specific cases of this inequality were studied in \cite{CGP20}.  

Le us now concentrate on the most important case $m=n$. We would like to find an analog of the Wasserstein-$p$ distance on $\Omega_n$.
A symmetric function sdist$:\Omega_n\times \Omega_n\to [0,\infty)$ is called a semidistance when sdist$(\rho^A,\rho^B)=0$ if and only if $\rho^A=\rho^B$.
We show that $\rT^Q_C$ is a semidistance if and only if $C$ is zero on $\cH_S$ and positive definite on $\cH_A$, where $\cH_S$ and $\cH_A$ are the subspaces of symmetric and skew-symmetric $n\times n$ matrices viewed as subspaces of $\C^n\otimes\C^n=\C^{n\times n}=\cH_S\oplus\cH_A$. 
 If $C$ is zero on $\cH_S$ and positive definite on $\cH_A$, then so is $C^p$ for any $p>0$. Consequently, 
 $\big( \rT_{C^p}^Q \big)^{1/p}$ is also a semidistance for any $p>0$. We further show that $\sqrt{\rT_C^Q}$ is a weak distance, i.e. there exists a distance $D'$ on $\Omega_n$ such that 
$\sqrt{\rT_C^Q(\rho^A,\rho^B)} \ge D'(\rho^A,\rho^B)$ for all $\rho^A,\rho^B\in\Omega_n$  -- see Theorem \ref{Wa2metthm}).
%
%
Then, we prove that for such $C$ there exists a unique maximum distance $D'$ on $\Omega_n$, which we shall call the induced quantum Wasserstein-2 distance, 
 given by the formula:
\begin{equation}\label{defQOTmet}
\wind(\rho^A,\rho^B)=\lim_{N\to\infty} \min_{\substack{\rho^{A_1},\ldots,\rho^{A_N}\in\Omega_n,\\ \rho^{A_0}=\rho^A, \, \rho^{A_{N+1}}=\rho^B}} \; \sum_{i=1}^{N+1}\sqrt{\rT^Q_{C}(\rho^{A_{i-1}},\rho^{A_i})}. 
\end{equation}
Similar construction can be done for any $p \geq 2$. The distance \eqref{defQOTmet} does not seem to be easily computable for a general cost matrix, however for some choices of $C$ formula \eqref{defQOTmet} simplifies significantly.


A simple example of the quantum cost matrix is provided by $C^Q$ --- the orthogonal projection of $\C^{n\times n}$ on $\cH_A$, as advocated
also in \cite{CHW19,Duv20,YZYY19} and \cite{Rie18}.  It is straightforward   to show that $C^Q=\frac{1}{2}(\I -S)$, where $S$ is the SWAP operator, $S(\bx\otimes\by)= \by\otimes\bx$, while 
 $\I$ is the identity operator on $\C^n\otimes\C^n$.  {Note that $C^Q$ is a projection, hence $(C^Q)^p = C^Q$ for any $p > 0$.
We show that $\big(\rT_{C^Q}^Q\big)^{1/p}$ does not satisfy the triangle inequality for $p\in[1,2)$. 
 On the other hand, for the single-qubit case, 
$n=2$, the square root of the optimal quantum transport cost, 
$\sqrt{\rT_{C^Q}^Q}$, does form  a distance.
In fact, we show that $\windQ=\sqrt{\rT_{C^Q}^Q}$ for qubits.
Furthermore,  $\sqrt{\rT_{C^Q}^Q}$ is a distance on pure states for any $n$ 
and numerical simulations strongly suggest that $\sqrt{\rT_{C^Q}^Q}$ satisfies the triangle inequality for all density matrices in $\Omega_n$ for $n=3, \ldots, 8$ (see \cite{FECZ} for the details). 
 It is also remarkable, that the optimal quantum transport cost $\rT_{C^Q}^Q$ is monotonous under quantum channels, at least in the single-qubit case \cite{BEZ22}.

For the specific cost matrix $C^Q$ one can study a related quantity, the {\sl SWAP-fidelity} between any two states, defined as \cite{FECZ},
\begin{equation}\label{defswapf}
\rF_S(\rho^A,\rho^B)=\max_{\rho^{AB}\in \Gamma^{Q}(\rho^A,\rho^B)} \tr  S\rho^{AB} = 1 - 2 \rT^Q_{C^Q}(\rho^A,\rho^B).
\end{equation} 
It is a symmetric function with the following properties: it is continuous, jointly concave, unitarily invariant, and super-multipicative with respect to the tensor product. Moreover, $\rF_S$ is bounded from below by the standard Uhlmann--Jozsa fidelity \cite{Uh76,Jo94}, $\rF$, and from above by its square root, $\sqrt{\rF}$.  

A simple generalization of $C^Q $  is the following operator that vanishes on $\cH_S$ and is positive definite on $\cH_A$:
\begin{equation}\label{defCQE}
\begin{aligned}
C_E^Q=\sum_{1\le i<j\le n} e_{ij}\frac{1}{\sqrt{2}}\big(|i\rangle|j\rangle - |j\rangle|i\rangle\big)\big(\langle i|\langle j| -\langle j|\langle i|\big),\\ 
\text{with } e_{ij}>0 \textrm{ for } 1\le i<j \le n.
\end{aligned}
\end{equation}
Here $|1\rangle,\ldots,|n\rangle$ is any orthonormal basis in $\cH_n$,
while the entries of a fixed symmetric matrix $e_{ij}$ can be interpreted as 
classical distances between the sites $i$ and $j$.
We show that decoherence 
of the marginal states, $\rho \to \diag (\rho)$, decreases the cost of QOT for $C^Q_E$:
\begin{equation}\label{decineq}
\rT_{C^Q_E}^Q(\diag(\rho^A),\diag(\rho^B))\le \rT_{C^Q_E}^Q(\rho^A,\rho^B) \textrm{ for } \rho^A, \rho^B\in\Omega_n.
\end{equation}

As in \cite{FrV18,Fri20}, we show that quantum transport can be defined on $d$-partite states.  In particular, one can define an analog of $C^Q$ for multi-partite systems.  More precisely, $C^Q$ is the projection on the orthogonal complement of the boson subspace --- the subspace of symmetric tensors in $\otimes^d\C^n$.


%
\subsection{A brief survey of the main results}\label{subsec:survey}

To make the paper accessible to a wide mathematical audience
we use a fusion of standard mathematical notation with the
notation of Dirac common in the physics community (see Subsection \ref{sec:not}). 
In Section \ref{sec:prilres} we present some preliminary results that are used in the rest of the paper.   Proposition \ref{exrank1}  shows that the  coupling set $\Gamma^Q(\rho^A,\rho^B)$ contains a matrix  of rank one if and only if $\rho^A$ and $\rho^B$ are isospectral.
Proposition \ref{honconv} shows that the function $\rT^Q_C(\rho^A,\rho^B)$ is continuous and convex on $\Omega_n\times \Omega_n$.  
In Subsection \ref{sec:swap} we discuss QOT with respect to the SWAP operator  
$S \in \rB(\cH_n\otimes \cH_n)$, which swaps the two factors of 
$\cH_n\otimes \cH_n$.  The operator $S$ has two invariant subspaces of $\cH_n\otimes \cH_n$, which is viewed as the set of $n\times n$ complex valued matrices $\C^{n\times n}$: the subspaces of symmetric and skew-symmetric matrices, denoted as $\cH_S$ and $\cH_A$, respectively.  The subspaces $\cH_S$ and $\cH_A$ correspond to the eigenvalues $1$ and $-1$ of $S$, respectively.   
In Subsection \ref{sec:swapfid} we discuss the properties of the SWAP-fidelity function, which shares many properties with the standard fidelity and coincides with the latter if at least one of the states is pure.  

Section \ref{sec:QOT} discusses the connection between QOT and  semidefinite programming.  
Theorem \ref{QOTSDP}  states formally that the computation of $\rT_C^Q$ is a SDP problem.  In particular, the computation time of $\rT_C^Q(\rho^A,\rho^B)$  within a precession of $\varepsilon\in (0,1)$ is polynomial in the size of the data and $\log1/\varepsilon$.  
 Theorem~\ref{dualQOT} establishes the dual problem and shows that its resolution yields the value of $\rT^Q_C$.  This was also shown in \cite{CHW19}
 for the specific projection cost matrix $C^Q$. 
Furthermore, Theorem \ref{dualQOT} states  the complementary conditions in the case that the supremum in the dual problem are achieved. 
 (This condition holds if $\rho^A$ and $\rho^B$ are positive definite.)  We found these complementary conditions to be very useful. 
 
 In Section \ref{sec:diagdm} we compare the  classical and quantum optimal transport
problems for diagonal density matrices.  For a given density matrix $\rho$ the diagonal density matrix $\diag(\rho)$ can be viewed as the decoherence of $\rho$ 
(for a more detailed physical mechanism see Supplemental Material in \cite{FECZ}).  Lemma \ref{diagdecr} shows that decoherence decreases the QOT for $C=C_E^Q$, cf. Formula \eqref{decineq}.  Lemma \ref{diaglemobs} gives a map of $\Gamma^{cl}(\bp^A,\bp^B)$ to $\Gamma^Q(\diag(\bp^A),\diag(\bp^B))$.  The main result of this section, Theorem \ref{clqotdiagdm}, proves three fundamental results.  First, the classical optimal transport is more expensive than the quantum optimal transport \eqref{qtrancheap}. Second, $\rT^Q_{C^Q}(\diag(\bp^A),\diag( \bp^B))$ can be stated as the minimum of a certain convex function on $\Gamma^{cl}(\bp^A,\bp^B)$.  This shows that the computation of $\rT^Q_{C^Q}(\diag(\bp^A), \diag(\bp^B))$ is simpler than the computation of $\rT^Q_{C^Q}(\rho^A,\rho^B)$ for general  states $\rho^A$ and $\rho^B$.   Third, it gives a simple formula for $\rT^Q_{C^Q}(\rho^A,\rho^B)$ for diagonal qubits.  In Subsection \ref{subsec:QOTCOT} we give a simple example of two diagonal qubits for which the cost of the classical OT is seven times higher then the cost of the QOT.  
In Subsection \ref{sec:decoh} we discuss the decoherence of 
 the quantum cost matrix on qubits, and introduce a convex combination 
$C^{Q}_{\alpha}=\alpha C^Q+(1-\alpha)\diag(C^Q)$, where $\alpha\in[0,1]$. Thus $\alpha=1$ and $\alpha=0$ correspond to QOT and OT, respectively.
 It yields that the quantity $\rT^Q_{\alpha}(\diag(\bp^A),\diag(\bp^B))$, which interpolates between the quantum and classical transport costs, is strictly decreasing in $\alpha$ on the interval $[0,1]$, unless either of the states is pure or $\bp^A = \bp^B$.  In particular, the cost of the classical optimal transport is larger than the cost of the quantum optimal transport.

Section \ref{subsec:lbdiagdm} gives lower bounds on the cost of the QOT with the cost matrix
$C^Q$ for any pair of density matrices of any dimension.  Equality holds for qubits.  This result (Theorem \ref{lowbdT}) is one of the main results of the paper, and will allow us to
 derive a semi-analytic
  formula for $\rT_{C^Q}^Q$ for qubits, as discussed in Appendix~\ref{apb:qubits}.

In Section \ref{sec:metrics} we show that any positive semidefinite cost matrix $C\in\rS(\cH_n\otimes\cH_n)$, that vanishes exactly on the on the subspace of symmetric matrices, yields the induced Wasserstein-2 distance \eqref{defQOTmet}.  For the qubit cost matrix $C^Q$ this Wasserstein-2 distance is $\sqrt{\rT_{C^Q}^Q(\rho^A,\rho^B)}$.

Section \ref{sec:QOTdpar} discusses the quantum optimal transport for $d$-partite systems for $d\ge 3$, denoted as $\rT_C^Q(\rho^{A_1},\ldots,\rho^{A_d})$.  The classical optimal transport of $d$-partite systems is discussed in \cite{FrV18,Fri20}.  The most interesting case is when the density matrix is in $\otimes^d\cH_n$.  Then the analog of the cost matrix $C^Q$ is given by  the projection operator $C^B$
on the complement of the subspace of symmetric tensors. 
The computation of $\rT_{C^B}^Q(\rho^{A_1},\ldots,\rho^{A_d})$ is related to the permanent function on positive semidefinite matrices.
If we assume that $d=2\ell$, where $\ell>1$, then, as in \cite{Fri20}, one can define a Wasserstein-2 distance on the space of $\ell$-tuples of density matrices $\Omega_n^\ell$ and on the space of unordered $\ell$-tuples $\{\rho^{A_1},\ldots,\rho^{A_\ell}\}$.

In Appendix \ref{sec:partr} we recall briefly the basic properties of partial traces.  
In Appendix \ref{apb:qubits} we discuss additional properties of the QOT for qubits.
Subsection~\ref{subsec:Blochb}
provides (Theorem \ref{thmC3}) a closed formula for $\rT_{C^Q}^Q(\rho^A,\rho^B)$ in terms of solutions of the trigonometric equation \eqref{Phieq}.  Lemma \ref{6sollem} shows that this trigonometric equation is equivalent to a polynomial equation of degree
at most $6$.  Subsection \ref{subsec:isospectral} gives a nice closed formula for the value of QOT for two isospectral qubit density matrices. In Subsection \ref{subsec:nonexisF} we present a simple example where the supremum of the dual SDP problem to QOT is not achieved. 
Appendix \ref{subsec:diagqut} gives a closed formula for the QOT for some
pairs of diagonal states of a single-qutrit system.
\section{Preliminary results}\label{sec:prilres}
The aim of this section is fivefold. First, we discuss briefly our notation.  Second, 
Proposition \ref{exrank1} shows that the coupling set $\Gamma^Q(\rho^A,\rho^B)$ contains a rank one matrix if and only if $\rho^A$ and $\rho^B$ are isospectral.  
Third, we discuss some basic properties of of $\rT_{C}^Q(\rho^A,\rho^B)$.
Fourth,  we introduce the SWAP operator,  and the corresponding cost matrices $C^Q, C^Q_E$, which are positive semidefinite and vanish on the set of symmetric matrices, the two-qubit bosons. Fifth,  we discuss SWAP fidelity and related quantities.
\subsection{Notation}\label{sec:not}
In what follows we combine 
the standard  mathematical notation with the Dirac notation used in quantum theory.
We view $\C^n$, the vector space of column vectors over the complex field $\C$, as a Hilbert space $\cH_n$ with the inner product 
\begin{equation*}
\langle \y,\x\rangle=\y^\dagger\x=\langle \y |\x\rangle.
\end{equation*}
Then $|i\rangle\in\cH_n$ is identified with the unit vector $\be_i=(\delta_{1i},\ldots,\delta_{ni})^\top$ for $i\in[n]$.
Let $\rB(\cH_n)\supset \rS(\cH_n)\supset \rS_{+}(\cH_n)\supset \Omega_n$ be  the space of linear operators, the real subspace of selfadjoint operators, the cone of positive semidefinite operators, and the convex set of density operators, respectively.   For $\rho\in \rB(\cH_n)$ we denote $|\rho|=\sqrt{\rho\rho^\dagger}\in \rS_+(\cH_n)$.  Then $\|\rho\|_1=\tr |\rho|$.
For $\rho,\sigma\in \rS(\cH_n)$ we write $\rho\geq \sigma$ and $\rho> \sigma$ if if the eigenvalues of $\rho-\sigma$ are all nonnegative or positive respectively.
 
The space of $n\times n$ complex valued matrices, denoted as  $\C^{n\times n}$, is a representation of $\rB(\cH_n)$, where the matrix $\rho=[\rho_{ij}]\in\C^{n\times n}$ represents the operator $\rho\in\rB(\cH_n)$.  The set  of density operators in $\rB(\cH_n)$ are viewed as  $\Omega_n$: the convex set of $n\times n$ Hermitian positive semidefinite trace-one matrices.
The tensor product $\cH_m\otimes\cH_n$ is represented by $\C^{m\times n}$.
An element of $\C^{m\times n}$ is a matrix $X=[x_{ip}]=\sum_{i=p=1}^{m,n} x_{ip} |i\rangle |p\rangle$, which correspond to a bipartite state.  Observe that $\x\otimes\y=|\x\rangle|\y\rangle$ is represented by the rank-one matrix $\x\y^\top$.  
We denote by $X^\dagger=\langle X|$ the complex conjugate of the transpose of  $X\in\C^{m\times n}$.
The inner product of bipartite states $X,Y\in\C^{m\times n}$ is $\langle X,Y\rangle=\langle X|Y\rangle=\tr X^\dagger Y$.
 We identify  $\rB(\cH_m\otimes \cH_n)$ with $\C^{(mn)\times (mn)}$ as follows. An operator $\rho^{AB}\in \rB(\cH_m\otimes\cH_n)$ is represented by a matrix $R\in\C^{(mn)\times (mn)}$, whose entries are indexed with two pairs of indices $r_{(i,p)(j,q)}$ where $i,j\in [m], p,q\in[n]$.  Then the partial traces of $R$ are defined as follows:

\begin{equation}\label{ptfrom}
\tr_A R=[\sum_{i=1}^m r_{(i,p)(i,q)}]=\rho^B\in\C^{n\times n}, \quad \tr_B R=[\sum_{p=1}^n r_{(i,p)(j,p)}]=\rho^A\in\C^{m\times m}.
\end{equation}
Recall that $\tr R=\tr (\tr_A R)=\tr (\tr_B R)$.  
Some more known facts about partial traces that we use in this paper are discussed in  Appendix \ref{sec:partr}.

Let $\rM:\rB(\cH_m\otimes\cH_n)\to \rB(\cH_m)\oplus\rB(\cH_n)$ be the partial trace map: $\rho^{AB}\mapsto (\rho^A, \rho^B)$.  We identify $\rM$ with the map $\rM: \C^{(mn)\times (mn)}\to\C^{m\times m}\oplus\C^{n\times n}$.   
For $ \rho^A\in\Omega_m, \rho^B\in\Omega_n$ we denote by $\Gamma^Q(\rho^A,\rho^B)$ the set of  all {\sl quantum coupling matrices} -- bipartite density matrices
$\rho^{AB}$ whose partial traces are $\rho^A$ and $\rho^B$ respectively,
$$\Gamma^Q(\rho^A,\rho^B)=\{ \rho^{AB}\in\Omega_{mn}, \tr_B \rho^{AB}= \rho^A, \tr_A \rho^{AB}= \rho^B\}.$$
Then $\Omega_{mn}$ fibers over $\Omega_m\times \Omega_n$, that is, $\Omega_{mn}=\bigcup_{( \rho^A, \rho^B)\in \Omega_m\times\Omega_n} \Gamma^Q( \rho^A, \rho^B)$.  The Hausdorff distance between $\Gamma^Q( \rho^A, \rho^B)$ and $\Gamma^Q(\rho^C,\rho^D)$ is a complete distance on the fibers \cite{FGZ19}.

\subsection{Isospectral density matrices}\label{subsec:isospec}
We identify $\cH_n\otimes \cH_n$ as the space of $n\times n$ complex valued matrices
$\C^{n\times n}$ 
as follows.   Let $\be_i=(\delta_{i1},\ldots,\delta_{in})^\top\equiv|i\rangle, i\in[n]$ be the standard basis in $\C^n\equiv \cH_n$.  Then a state $|\psi\rangle\in \cH_n\otimes \cH_n$ is given by $|\psi\rangle=\sum_{i,j=1}^n x_{ij}|i\rangle|j\rangle$.  Thus we associate with $|\psi\rangle$ the matrix $X=[x_{ij}]\in \C^{n\times n}$.  Then $|\psi\rangle$ is a normalized state if and only if $\|X\|^2=\tr X  X^\dagger=1$.
Suppose we change the orthonormal basis $\be_1,\ldots,\be_n$ to an orthonormal basis $\bbf_1,\ldots,\bbf_n$, where $\be_i=\sum_{p=1}^n u_{pi}\bbf_p$.  Here $U=[u_{ip}]\in\C^{n\times n}$ is a unitary matrix. 
 Then $|\psi\rangle=\sum_{p,q=1}^n y_{pq}|\bbf_p\rangle|\bbf_q\rangle$, where $Y=UXU^\top$.

We now consider a pure state density operator 
\begin{equation*}
|\psi\rangle \langle \psi|=\Big(\sum_{i,j=1}^n x_{ij}|i\rangle|j\rangle\Big) \Big(\sum_{p,q=1}^n \bar x_{pq}\langle p|\langle q|\Big)=\sum_{i,j,p,q=1}^n x_{ij}\bar x_{pq} |i\rangle|j\rangle \langle p|\langle q|.
\end{equation*}
We identify the coefficient matrix with the Kronecker product $X\otimes \bar X$.
Then
\begin{align*}
\rho^A & =\tr_B |\psi\rangle \langle \psi|=\sum_{i,p=1}^n (X X^\dagger)_{ip} |i\rangle \langle p|,\\ 
\rho^B & =\tr_A |\psi\rangle \langle \psi|=\sum_{j,q=1}^n (X^\top  \bar X)_{jq} |j\rangle \langle q|.
\end{align*}
Thus in the standard basis of $\cH_n$ we can identify $\rho^A$ and $\rho^B$ with the density matrices
\begin{equation}\label{Xrhosigrel}
\rho^A = X X^\dagger, \quad \rho^B=X^\top  \bar X.
\end{equation}

If we change from the standard basis  to
 the  basis $\bbf_1,\ldots,\bbf_n$,  using a unitary matrix $U$,
 the both partial traces read,
\begin{equation}\label{tilderhosigfor}
\begin{aligned}
\tilde \rho^A &= \tilde X \tilde X^\dagger \, = \, U(X X^\dagger) U^\dagger \, = \, U\rho^A U^\dagger, \\ 
\tilde \rho^B &= \tilde X^\top  \overline{\tilde X}=U(X^\top  \bar X) U^\dagger=U\rho^B U^\dagger.
\end{aligned}
\end{equation}

Note that if $\nu_1\ge \cdots\ge\nu_n\ge 0$ are the singular values of the matrix $X$ then $\lambda_1=\nu_1^2\ge \cdots \ge \lambda_n=\nu_n^2\ge 0$ are the eigenvalues of $\rho^A$ and $\rho^B$.  That is $\rho^A$ and $\rho^B$ are isospectral.  Vice versa: 
\begin{proposition}\label{exrank1}  Let $\rho^A,\rho^B\in\Omega_n$.  Then $\Gamma^Q(\rho^A,\rho^B)$ contains a matrix $R$ of rank one if and only if $\rho^A$ and $\rho^B$ are isospectral.
\end{proposition}
\begin{proof}  Suppose first that $\rho^A$ and $\rho^B$ are isospectral, i.e., have the same eigenvalues $\lambda_1\ge\cdots\ge\lambda_n\ge 0$.  Assume that $\rho^A$ and $\rho^B$ have the following spectral decompositions:
\begin{align}\label{specdecrho} 
\begin{split}
\rho^A & =\sum_{i=1}^n \lambda_i |\x_i\rangle\langle \x_i|, \quad \langle \x_i,\x_j\rangle =\delta_{ij},\\
\rho^B & =\sum_{j=1}^n \lambda_i |\y_j\rangle\langle \y_j|, \quad\! \langle \y_i,\y_j\rangle =\delta_{ij}.
\end{split}
\end{align}
Then the set  $\Gamma^Q(\rho^A, \rho^B)$ of coupling matrices 
contains the rank-one matrix
\begin{equation}\label{simpur}
R= \Big(\sum_{i=1}^n \sqrt{\lambda_i} |\x_i\rangle |\y_i\rangle \Big) \Big(\sum_{j=1}^n \sqrt{\lambda_j}\langle \x_j|\langle \y_j| \Big).
\end{equation}
Vice versa, if $R$ represents a projector onto a pure bipartite state 
$|\psi\rangle$ in $\rS_+(\cH_n\otimes \cH_n)$,
 then its  Schmidt decomposition, 
 related to the
 Singular Value Decomposition (SVD) of the above matrix $X=[x_{ij}]$, 
takes the above form \cite{Frb16}.  
Hence $\tr_A R$ and $\tr_B R$ are isospectral density matrices.
\end{proof}


\subsection{Some properties of $\rT_{C}^Q(\rho^A,\rho^B)$}\label{subsec:propQOT}

 We first recall the definition of the Hausdorff distance between non-empty sets in a space $X$ equipped with a distance $D: X\times X \to [0,\infty)$.  Assume that $Y, Z\subset X$.  Then the Hausdorff distance between $Y$ and $Z$ is:
\begin{equation*}
D_H(Y,Z) = \max \Big( \sup_{y\in Y}\inf_{z\in Z}D(y,z), \, \sup_{z\in Z}\inf_{y\in Y}D(y,z) \Big).
\end{equation*}
We will need the Hausdorff distance in the proof of the following result.

\begin{proposition}\label{honconv} For $C\in \rS(\cH_m\otimes \cH_n)$  the function 
$\rT_{C}^Q(\cdot,\cdot)$ is a continuous convex function on $\Omega_m\times\Omega_n$: for any $0<a<1$,
\begin{align*}
\rT_{C}^Q(a\rho^A+(1-a)\sigma^A,a\rho^B+(1-a)\sigma^B))\le 
a\rT_{C}^Q(\rho^A,\rho^B)+(1-a)\rT_{C}^Q(\sigma^A,\sigma^B).
\end{align*}
Furthermore, if $C\ge  0$ then  $\rT_{C}^Q(\cdot,\cdot)$ is nonnegative.
\end{proposition}
\begin{proof} 
Assume that 
\begin{align*}
&\rT_{C}^Q(\rho^A,\rho^B)=\tr C \rho^{AB}, &&\hspace*{-2cm} \rho^{AB}\in \Gamma^Q(\rho^A,\rho^B),\\ 
&\rT_{C}^Q(\sigma^A,\sigma^B)=\tr C \sigma^{AB}, &&\hspace*{-2cm} \sigma^{AB}\in \Gamma^Q(\sigma^A,\sigma^B).
\end{align*}
Let $\tau^{AB}=a\rho^{AB}+(1-a)\sigma^{AB}$. Then $\tau^{AB}\in\Gamma^Q(a\rho^A+(1-a)\sigma^A,a\rho^B+(1-a)\sigma^B)$. 
Clearly $\tr C \tau^{AB}=a\rT_{C}^Q(\rho^A,\rho^B)+(1-a)\rT_{C}^Q(\sigma^A,\sigma^B)$.  The minimal characterization~\eqref{defkapCAB} of $\rT$ yields the first inequality of the lemma.  Clearly if $C\ge 0$ then $\rT_{C}^Q(\cdot,\cdot)$ is nonnegative.
This yields the second inequality of the lemma.

The continuity of $\rT_{C}^Q(\cdot,\cdot)$ follows from the following arguments.  
Assume that 
\begin{gather*}
\rho^A=\lim_{k\to\infty} \rho^{A_k},\qquad \rho^B=\lim_{k\to\infty} \rho^{B_k}, \\
\rT_{C}^Q(\rho^{A_k},\rho^{B_k})=\tr C\rho^{A_kB_k}, \qquad \rho^{A_kB_k}\in \Gamma^Q(\rho^{A_k},\rho^{B_k}), \; k\in\N.
\end{gather*}
As the set of density matrices of a fixed dimension is a compact set,  there exists a subsequence $\{l_k, k\in\N\}$  such that $\rho^{A_{l_k}B_{l_k}}$ converges to a density matrix $\rho$.   Hence 
\begin{align*}
\tr_A \rho & =\lim_{k\to\infty}\tr_A\rho^{A_{l_k}B_{l_k}}=\lim_{k\to\infty}\rho^{B_{l_k}}=\rho^B, \\
\tr_B \rho & =\lim_{k\to\infty}\tr_B\rho^{A_{l_k}B_{l_k}}=\lim_{k\to\infty}\rho^{A_k}=\rho^A.
\end{align*}
Hence,  $\rho\in\Gamma^Q(\rho^A,\rho^B)$.  Therefore,
\begin{equation}\label{lowbdQOTlim}
\lim_{k\to\infty} \rT_{C}^Q(\rho^{A_{l_k}},\rho^{B_{l_k}})=\lim_{k\to\infty} \tr C\rho^{A_{l_k}B_{l_k}}=\tr C \rho\ge \rT_{C}^Q(\rho^A,\rho^B)=\tr C\rho^{AB}.
\end{equation}

Observe that for each $\rho^A\in\Omega_m,\rho^B\in\Omega_n$, the set $\Gamma^Q(\rho^A,\rho^B)$, viewed as a fiber over $(\rho_A,\rho_B)$, is a compact convex set.  Hence, one can define the Hausdorff distance (distance) on the fibers.
It is shown in \cite[Theorem 5.2]{FGZ19} that the Hausdorff distance is a complete distance.  Furthermore, the sequence $\Gamma^Q(\rho^{A_k},\rho^{B_k}),k\in\N$ converges to
$\Gamma^Q(\rho^A,\rho^B)$ in the Hausdorff distance if and only if $\lim_{k\to\infty}(\rho^{A_k},\rho^{B_k})=( \rho^A, \rho^B)$.    Hence,  there exists a sequence $\omega_k\in \Gamma^Q(\rho^{A_{l_k}B_{l_k}})$ such that $\lim_{k\to\infty}\|\omega_k-\rho^{AB}\|=0$.  (Here we let $\|\eta\|=\sqrt{\tr \eta^2}$ for any Hermitian operator on $\cH_N$.) 
Clearly,
\begin{align*}
\rT_{C}^Q(\rho^{A_{l_k}},\rho^{B_{l_k}})\le \tr C\omega_{k}& =\tr C\rho^{AB}+\tr C(\omega_k-\rho^{AB}) \\
& = \rT_{C}^Q(\rho^A,\rho^B)+\tr C(\omega_k-\rho^{AB}).
\end{align*}
Let $k\to\infty$ to deduce the inequality $\lim_{k\to\infty} \rT_{C}^Q(\rho^{A_{l_k}},\rho^{B_{l_k}})\le \rT_{C}^Q(\rho^A,\rho^B)$.  Combine that with \eqref{lowbdQOTlim}
to deduce the equality $\lim_{k\to\infty} \rT_{C}^Q(\rho^{A_{l_k}},\rho^{B_{l_k}})= \rT_{C}^Q(\rho^A,\rho^B)$.   Clearly, the sequence $\{\rT_{C}^Q(\rho^{A_{k}},\rho^{B_{k}}), k\in\N\}$ is a bounded sequence.  We showed that from this sequence we can always extract a subsequence which converges to $\rT_{C}^Q(\rho^A,\rho^B)$.
Hence, this sequence converges to $\rT_{C}^Q(\rho^A,\rho^B)$.
\end{proof}
  
The following Proposition shows that to compute $\rT_{C}^Q( \rho^A, \rho^B)$ one can assume that the eigenvalues of $C$ are in the interval $[0,1]$: 
\begin{proposition}\label{propkapCC'}  Assume that $C\in \rS(\cH_m\otimes \cH_n)$ is not a scalar operator, $C\ne c \I$.  Let   
\begin{equation*}
\tilde C=\frac{1}{\lambda_{\max}(C)-\lambda_{\min}(C)} \big(C-\lambda_{\min}(C) \I \big).
\end{equation*}
Then $0\le \tilde C\le \I$.  Furthermore for $\rho^A\in\Omega_m,\rho^B\in\Omega_n$  the following equality holds:
\begin{equation}\label{kapCC'rel}
\rT_{C}^Q(\rho^A,\rho^B)=(\lambda_{\max}(C)-\lambda_{\min}(C))\rT_{\tilde C}^Q(\rho^A,\rho^B)+\lambda_{\min}(C).
\end{equation}
\end{proposition}
\begin{proof} Clearly $C=\bigl(\lambda_{\max}(C)-\lambda_{\min}(C) \bigr)\tilde C+\lambda_{\min}(C)\I$.  Furthermore
\begin{align*}
\tr C \rho^{AB}=\bigl(\lambda_{\max}(C)-\lambda_{\min}(C)\bigr)\tr \tilde C\rho^{AB}+\lambda_{\min}(C), \quad \rho^{AB}\in\Gamma^Q(\rho^A,\rho^B).
\end{align*}
As $\lambda_{\max}(C)-\lambda_{\min}(C)> 0$ we deduce \eqref{kapCC'rel}.
\end{proof}

We next observe that one can reduce the computation of $\rT_{C}^Q(\rho^A,\rho^B)$ to a smaller 
dimension problem if either $\rho^A$ or $\rho^B$ are not positive definite:
\begin{proposition}\label{redprop}    Assume that  $\rho^A\in\Omega_m, \rho^B\in\Omega_n$. Let $m'$ and $n'$ be the dimensions of\, $\range\rho^A=\cH_{m'}$ and\, $\range\rho^B=\cH_{n'}$ respectively.  Denote by $\,\rho^{A'}\in\Omega_{m'}$, and $\rho^{B'}\in\Omega_{n'}$ the restrictions of $\rho^A$ and $\rho^B$ to $\cH_{m'}$ and $\cH_{n'}$ respectively.  Assume that  $C\in\rS(\cH_m\otimes \cH_n)$, and denote by $C'\in\rS(\cH_{m'}\otimes \cH_{n'})$ the restriction of $C$ to $\cH_{m'}\otimes \cH_{n'}$.  Then
\begin{equation*}
\rT^Q_{C}(\rho^A,\rho^B)=\rT^Q_{C'}(\rho^{A'},\rho^{B'}).
\end{equation*}
\end{proposition}
\begin{proof}
Without loss of generality we can assume that we chose orthonormal bases in $\cH_m$ and $\cH_n$ to be the eigenvectors of $\rho^A$ and $\rho^B$ respectively.
Thus to prove the lemma it is enough to consider the following case: $\rho^A=\rho^C\oplus 0_{m-l}$ where $\rho^C\in \Omega_l, l<m$ and $0_l$ is an $l\times l$ zero matrix.  Let $\tilde C\in \rS(\cH_l\otimes\cH_n)$ be the restriction of $C$ to $\cH_l\otimes\cH_n$.  We claim that
\begin{equation}\label{redprop1}  
\rT^Q_{C}(\rho^A,\rho^B)=\rT^Q_{\tilde C}(\rho^{C},\rho^{B}).
\end{equation}
Let $R=[R_{(i,p)(j,q)}]\in\Gamma^Q(\rho^A,\rho^B)$.  As $R\ge 0$ it follows that the submatrix $R_{ii}=[R_{(i,p)(i,q)}], p,q\in[n]$ is positive semidefinite for each $i\in[m]$.
Since $\tr _B R=\rho^A$ we deduce that $\rho^A_{ii}=\sum_{p\in[n]} R_{(i,p)(i,p)}=\tr R_{ii}=0$ for $i>l$.   Therefore $R_{ii}=0$, that is, $R_{(i,p)(i,q)}=0$ for $p,q\in[n]$ and $i>l$.  Let $R'$ be the following submatrix of $R$: $[R_{(i,p)(j,q)}], i,j\in [l], p,q\in[n]$.  Then $R'\in\Gamma^Q(\rho^C,\rho^B)$.  Vice versa, given $R'\in\Gamma^Q(\rho^C,\rho^B)$, one can enlarge trivially $R'$ to $R$ in $\Gamma^Q(\rho^C, \rho^B) $.  Clearly $\tr CR=\tr \tilde C R'$.  Repeating the same process with $\rho^B$ establishes \eqref{redprop1}.  
\end{proof}

As we point out in the next subsection, it is natural to consider the case $m=n$.  However, if either $\rho^A$ or $\rho^B$ are singular density matrices then we can reduce the computation of $T_{C}^Q(\rho^A,\rho^B)$ to a lower-dimensional problem, and after this reduction it may happen that the dimensions are no longer equal.

\subsection{Quantum transport problem induced by the SWAP operator}\label{sec:swap}
When describing any two distinguishable physical objects one can 
introduce an operation $S$ which exchanges them.
On the composite space $\cH_n\otimes \cH_n$ 
it corresponds to
 a natural isometry induced by swapping the two factors $\x\otimes \y\mapsto \y\otimes \x$. On the space of square matrices  the SWAP operator is the map $S:X\mapsto X^\top$.  This map is of fundamental importance in quantum information theory.  It allows to observe some interesting properties of bipartite system
 and is useful 
 in the criterion for separability by Peres and Horodecki \cite{Per96, Hor96}. 
 It is shown below that  the SWAP operator $S$ 
 induces a cost matrix 
 \begin{equation}\label{defT} 
C^Q=\frac{1}{2}(\I -S),
\end{equation}
for the quantum transport problem, which enjoys several nice properties.

For $\cH_n\otimes\cH_n$ the SWAP operation $S\in\rB(\cH_n\otimes\cH_n)$ acts on the product states as follows: $S(|\x\rangle|\bu\rangle)=|\bu\rangle| \x\rangle$.  So $S$ is both unitary and an involution operator:
$S^\dagger S=I$ and $S^2=I$. Hence the eigenvalues of $S$ are $\pm 1$ and $S$ is selfadjoint, $S^\dagger=S$.  The invariant subspaces of $S$ corresponding to the eigenvalues $1$ and $-1$ are the symmetric and skew-symmetric tensors respectively,
which can be identified with the symmetric  $\cH_S=\rS^2\C^n$ and skew-symmetric $\cH_A=\rA^2\C^n$  matrices in $\C^{n\times n}$, respectively.   
Note that the decomposition of a  matrix $X$ into a sum of symmetric and skew-symmetric matrices $X=(1/2)(X+X^\top)+(1/2)(X-X^\top)$ is an orthogonal decomposition.  That is 
\begin{equation*}
\cH_n\otimes\cH_n=\cH_S\oplus \cH_A=\C^{n\times n}=\rS^2\C^n\oplus \rA^2\C^n
\end{equation*}
is an orthogonal decomposition.  Observe that $S(X)=X^\top$.   Hence the action of $S$ on a rank-one operator $|X\rangle\langle Y|$ in $\rB(\cH_n\otimes \cH_n)$ is $S(|X\rangle \langle Y|)=|X^\top\rangle \langle Y|$.  Therefore the action of $S$ on rank one product operator in $\rB(\cH_n\otimes\cH_n)$ is given by
\begin{eqnarray*}
S(|\x\rangle|\bu\rangle\langle \by|\langle\bv|)=S(|\x\rangle|\bu\rangle)\langle \by|\langle \bv|=|\bu\rangle |\x\rangle \langle \by|\langle\bv|.
\end{eqnarray*}
Hence
\begin{eqnarray*}
\tr S(|\x\rangle|\bu\rangle\langle \by|\langle\bv|)=
(\langle \by|\langle\bv|)(|\bu\rangle|\x\rangle)=\langle \y|\bu\rangle\langle \bv|\x\rangle.
\end{eqnarray*}
Similarly
\begin{eqnarray*}
S(|\x\rangle|\bu\rangle\langle \by|\langle\bv|)S^\dagger=
|\bu\rangle |\x\rangle\langle\bv|\langle\y|.
\end{eqnarray*}
Use the identity \eqref{tenprodiden} and the above results to deduce that 
\begin{eqnarray*}
&&\tr S(|\x\rangle|\bu\rangle\langle \by|\langle\bv|)=
\langle \y|\bu\rangle \langle \bv|\x\rangle=\tr \big( (|\bu\rangle \langle\y|)\otimes(|\x\rangle \langle\bv|) \big),\\
&&\tr_A S(|\x\rangle|\bu\rangle\langle \by|\langle\bv|)S^\dagger=
\langle \bv|\bu\rangle|\x\rangle\langle \by|=\tr_B |\x\rangle|\bu\rangle\langle \by|\langle\bv|,\\
&&\tr_B S(|\x\rangle|\bu\rangle\langle \by|\langle\bv|)S^\dagger=
\langle \y|\x\rangle|\bu\rangle\langle \bv|=\tr_A |\x\rangle|\bu\rangle\langle \by|\langle\bv|.
\end{eqnarray*}
Use \eqref{tenprodiden} to deduce 
\begin{equation*}
S((|\x\rangle \langle \y|)\otimes (|\bu\rangle\langle \bv|))=|\bu\rangle |\x\rangle \langle\by|\langle\bv|=(|\bu\rangle\langle\y|)\otimes (|\bx\rangle\langle \bv|).
\end{equation*}
Combine the above equalities to obtain the following identities: 
\begin{align}\label{Werid}
\begin{split}
& \tr S(\rho^A\otimes\rho^B)=\tr \rho^A\rho^B, \quad
\rho^A,\rho^B\in\rB(\cH_n),\\
& \tr_A S\rho^{AB}S^\dagger=\tr_B \rho^{AB}, \;\; \tr_B S\rho^{AB}S^\dagger=\tr_A \rho^{AB}, \quad \rho^{AB}\in\rB(\cH_n\otimes\cH_n).
\end{split}
\end{align}
The first identity is due to Werner \cite{Wer89}, see also \cite{MPHUZ}.

Denote by $\ker C$ the kernel of a linear operator $C:\cH_n\otimes\cH_n\to \cH_n\otimes\cH_n$.  An operator $C$  is said to vanish exactly on symmetric matrices if $\ker C=\cH_S$. Thus a positive semidefinite $C$  vanishes exactly on $\cH_S$ if and only if it has $n(n-1)/2$ positive eigenvalues (counting with multiplicities) with the corresponding skew-symmetric eigenvectors.

Let $|1\rangle,\ldots,|n\rangle$ be an orthonormal basis in $\cH_n$.  Define (as in \cite{FECZ}) the maximally entangled singlet states spanned on two dimensional subspaces:
\begin{equation}\label{SOBHA}
|\psi_{ij}^-\rangle=\frac{1}{\sqrt{2}}\big(|i\rangle|j\rangle -|j\rangle|i\rangle\big) \textrm{ for } 1\le i<j\le n.
\end{equation}
Given a classical distance matrix $E = [e_{ij}]_{i, j = 1}^n$ with $e_{ij} > 0$ for all $1 \le i < j \le n$, the following operator is positive semidefinite and vanishes exactly on
the symmetric subspace, $\rS^2\C^n$
\cite[(11)]{FECZ}:
\begin{equation}\label{defCEop}
C^Q_E=\sum_{1\le i<j\le n} e_{ij} |\psi_{ij}^-\rangle \langle \psi_{ij}^-|, 
\end{equation}
Consider the operator $C^Q$ given by \eqref{defT}.
Then $C^Q$ is an orthogonal projection of $\C^{n\times n}$ onto 
antisymmetric subspace, $\rA^2\C^n$. 
 Hence $C^Q$ is of the form \eqref{defCEop}, where $e_{ij}=1$ for all $i < j$,
so such a distance matrix $E$ represents 
 the simplex configuration.
Denote by $\rU(n)\subset \C^{n\times n}$ the group of unitary matrices. The following lemma shows that $\rT_{C^Q}^Q$ is invariant under conjugation by a unitary matrix: 
\begin{proposition}\label{conjunit}  Assume that $\rho^A,\rho^B\in\Omega_n$ and $\rho^{AB}\in\Gamma^Q(\rho^A,\rho^B)$.  Then for $U\in\rU(n)$ the following equalities hold:
\begin{equation}\label{conjunit1}
\begin{aligned}
& \tr_B((U\otimes U) \rho^{AB} (U^\dagger\otimes U^\dagger)) = U\rho^A U^\dagger, \\ 
& \tr_A ((U\otimes U) \rho^{AB} (U^\dagger\otimes U^\dagger)) = U\rho^B U^\dagger,\\
& (U\otimes U) \Gamma^Q(\rho^A,\rho^B) (U^\dagger\otimes U^\dagger) =\Gamma^Q(U\rho^A U^\dagger,U \rho^B U^\dagger),\\
& \rT_{C}^Q(\rho^A,\rho^B) =\rT_{(U\otimes U)C(U^\dagger\otimes U^\dagger)}^Q(U\rho^A U^\dagger,  U\rho^B U^\dagger).
\end{aligned}
\end{equation}
In particular, the optimal transport cost for $C$ given by \eqref{defT} is unitarily invariant,
\begin{equation}\label{CQinvQOT}
\rT_{C^Q}^Q(\rho^A,\rho^B)=\rT_{C^Q}^Q(U\rho^A U^\dagger,  U\rho^B U^\dagger).
\end{equation}
\end{proposition}
\begin{proof}
Assume that $R$ is a pure state $R=|\psi\rangle \langle \psi|$.  The state $|\psi\rangle$ corresponds to a matrix $X\in\C^{n\times n}$ with $\tr X X^\dagger=1$.  Then $\tr_B R=X X^\dagger$ and $\tr_A R=X^\top \bar X$.
 Recall that $(U\otimes U)|\psi\rangle $ is represented by $\tilde X=UX U^\top$
   Now use \eqref{tilderhosigfor} to deduce the first two equalities in \eqref{conjunit1} if $R\in\Gamma^Q(\rho^A,\rho^B)$.
Recall that any $\rho^{AB}\in \Gamma^Q(\rho^A,\rho^B)$ is a convex combination of pure states $R_i=|\psi_i\rangle \langle \psi_i|, i \in[k]$.  That is $R=\sum_{i=1}^k a_i R_i$, where $a_i>0$ and $\sum_{i=1}^k a_i=1$.  Then $\tr_B R_i=\rho_i^A, \tr_A R_i=\rho^B_i$.
Now use the above results for $R_i$ to deduce the first two equalities in \eqref{conjunit1},
which imply  two following equalities.
Equality \eqref{CQinvQOT} is deduced from the fact
\begin{equation}\label{invCqt}
(U\otimes U) C^{Q}(U^\dagger\otimes U^\dagger)=C^{Q}.\qedhere
\end{equation}
\end{proof}
\subsection{SWAP fidelity and related quantities}\label{sec:swapfid}
Recall the definitions of \emph{fidelity}, \emph{Bures distance} and \emph{root infidelity} \cite{BZ17}. For $\rho^A,\rho^B\in\Omega_n$ these are defined, respectively, as

 \begin{align}
\mathrm{F}( \rho^A,  \rho^B) & = \Big(\rm {Tr} \big|\sqrt{ \rho^A} \sqrt{ \rho^B}\big|  \Big)^2, \label{fidelity}\\
\rB( \rho^A,  \rho^B) & =\sqrt{2} \sqrt{1 -\sqrt{\rF( \rho^A, \rho^B)}}, \label{Bures}\\
\rI( \rho^A,  \rho^B) & = \sqrt{1-\rF( \rho^A, \rho^B)}. \label{Ifidelity}
\end{align}
   
It is known that the root infidelity \cite{GLN05} and Bures distance lead to metrics on $\Omega_n$.
As $0\le \rF(\rho^A,\rho^B)\le 1$ it follows that 
\begin{equation*}
\rB(\rho^A,\rho^B)\le \sqrt{2} \; \rI(\rho^A,\rho^B).
\end{equation*}
{
We now describe basic properties of the SWAP-fidelity $\rF_S(\rho^A,\rho^B)$ given by \eqref{defswapf}:
\begin{proposition}\label{swapfidprop}  Assume that $\rho^A,\rho^B\in\Omega_n$ and $\rho^{AB}\in\Gamma^Q(\rho^A,\rho^B)$.   Then the SWAP-fidelity function has the following properties:
\begin{enumerate}[(a)]
\item 
\begin{equation}\label{eqTFS}
\rF_S(\rho^A,\rho^B)=1-2\,\rT^Q_{C^Q}(\rho^A,\rho^B).
\end{equation}
\item The function $\rF_S$ is a symmetric concave function on $\Omega_n\times \Omega_n$ with values in the interval $[0,1]$. 
\item $\rF_S(\rho^A,\rho^B)=1$ if and only if $\rho^A=\rho^B$.
\item $\rF_S(\rho^A,\rho^B)=0$ if and only if $\tr \rho^A\rho^B=0$.
\item The following equalities and inequalities hold 
\begin{align}
& \rF_S(\rho^A,\rho^B)=\rF_S(U\rho^AU^\dagger,U\rho^B U^\dagger), && \textrm{ for } U\in\rU(n),\label{swapfideU}\\
& \rF(\rho^A,\rho^B)\le \rF_S(\rho^A,\rho^B)\le \sqrt{\rF(\rho^A,\rho^B)}, && \label{swapfideqin}\\
& \rF_S(\rho^A\otimes\sigma^C, \rho^B\otimes\sigma^D)\ge \rF_S(\rho^A,\rho^B)\rF_S(\sigma^C,\rho^D), && \textrm{ for } \sigma^C,\sigma^D\in \Omega_m.\label{swapfidM}
\end{align}
\item $\rF_S(\rho^A,\rho^B)=\rF(\rho^A,\rho^B)$ if either $\rho^A$ or $\rho^B$ is a pure state.
\end{enumerate}
\end{proposition}

\medskip

\begin{proof}
(a)  Observe that $\tr S\rho^{AB}=1 -2\tr\big((1/2)(\I-S)\rho^{AB}\big)$.  
Then the maximal characterization of $\rF_S(\rho^A,\rho^B)$ in \eqref{defswapf} and the minimum characterization of $\rT^Q_{C^Q}(\rho^A,\rho^B)$ yield equality  \eqref{eqTFS}.

\medskip

\noindent 
(b) Proposition \ref{honconv} and equality \eqref{eqTFS} yield that $\rF_S$ is a concave function on $\Omega_n\times \Omega_n$ whose value is at most $1$.  Part (d) of Theorem \ref{kapTABprop} yields that $\rT_{C^Q}^Q(\rho^A,\rho^B)\le 1/2$.  
Equality \eqref{eqTFS} yields that $\rF_S\ge 0$.

\medskip

\noindent 
(c) Clearly, $\rF_S(\rho^A,\rho^B)=1\iff \rT_{C^Q}^Q(\rho^A,\rho^B)=0$.  Part (c) of Theorem \ref{kapTABprop} yields that this happens if and only if $\rho^A=\rho^B$.

\medskip

\noindent
(e)  Equation \eqref{swapfideU} follows from the equalities \eqref{eqTFS} and \eqref{CQinvQOT}.  Inequality \eqref{swapfideqin} is stated and proved in Theorem 10 in \cite{YZYY19}.  Finally, inequality \eqref{swapfidM} follows from 
Lemma \ref{ptcombS} and the following observations. Let $S_n$ and $S_{n,m}$ be the SWAP operators on $\cH_n\otimes \cH_n$ and $\cH_{(nm)^2}:=(\cH_n\otimes \cH_m)\otimes (\cH_n\otimes \cH_m)$, respectively.  Denote by $R_{n,m}:\cH_{n}\otimes \cH_{n}\otimes\cH_m\otimes \cH_m\to \cH_{(mn)^2}$ the SWAP of the two middle factors.
 Assume that 
\begin{equation*}
\rF_S(\rho^A,\rho^B)=\tr S_n \rho^{AB}, \quad
\rF_S(\sigma^C,\sigma^D)=\tr S_m\sigma^{CD}.
\end{equation*}
Let $\tau^{ACBD}:=R_{n,m}(\rho^{AB}\otimes \sigma^{CD})R_{n,m}\in \Omega_{nm}$.   The first line of \eqref{ptcombS1} yields 
 $\tau^{ACBD}\in\Gamma^Q(\rho^A\otimes \sigma^C,\rho^B\otimes \sigma^D)$.  
 The second line of \eqref{ptcombS1} yields 
 \begin{equation*}
 \tr S_{n,m}\tau^{ACBD}=(\tr S_n\rho^{AB})(\tr S_m\sigma^{CD})=
 \rF_S(\rho^A,\rho^B)\rF_S(\sigma^A,\sigma^B).
 \end{equation*}
 The maximum characterization of $\rF_S(\rho^A\otimes\sigma^C, \rho^B\otimes\sigma^D)$ yields inequality \eqref{swapfidM}.

  \medskip
 
 \noindent
 (d) Recall that $\tr \rho^A\rho^B=0$ if and only if the eigenvectors of $\rho^A$ and $\rho^B$ corresponding to positive eigenvalues are orthogonal.  This is equivalent to $\rF(\rho^A,\rho^B)=0$.  Use inequality \eqref{swapfideqin} to deduce part (d).
 
\medskip 
 
 \noindent
 (f) Assume that $\rho^A$ is a pure state.   Then $\sqrt{\rho^A}=\rho^A$.   Hence, 
 \begin{align*}
 & \big|\sqrt{ \rho^A} \sqrt{ \rho^B}\big|^2= \sqrt{\rho^A}\rho^B\sqrt{\rho^A}=\rho^A\rho^B\rho^A=(\tr\rho^A\rho^B)\rho^A \\ 
 \Rightarrow \quad &
 \big|\sqrt{ \rho^A} \sqrt{ \rho^B}\big|=
 \sqrt{\tr\rho^A\rho^B}\rho^A \; \Rightarrow \; \rF_S(\rho^A,\rho^B)=\tr \rho^A\rho^B.
 \end{align*}
 Use \eqref{eqTFS} and \eqref{QOTrankone} to deduce $\rF_S(\rho^A,\rho^B)=\rF(\rho^A,\rho^B)$.
\end{proof}
Recall the definition of $\windQ$ \eqref{defQOTmet}. 
\begin{corollary}\label{compmetric} Let $\rho^A,\rho^B\in\Omega_n$.  Then the following inequalities hold:
\begin{gather}\label{lowupbdsQT1}
\frac{1-\sqrt{\mathrm{F}(\rho^A,\rho^B)}}{2}  \le \rT_{C^Q}^Q(\rho^A,\rho^B)\le \frac{1-\mathrm{F}(\rho^A,\rho^B)}{2},\\
\frac{1}{2}\rB(\rho^A,\rho^B) \le  \windQ(\rho^A,\rho^B) \le \sqrt{\rT_{C^Q}^Q(\rho^A,\rho^B)} \le \frac{1}{\sqrt{2}}\rI(\rho^A,\rho^B)
\label{lowupbdsQT2}
%
\end{gather}
Moreover, if either of the states is pure, then
$\rT_{C^Q}^Q(\rho^A,\rho^B) = (1-\mathrm{F}(\rho^A,\rho^B))/2$. 

\end{corollary}

\begin{proof}
The bound \eqref{lowupbdsQT1} follows directly from inequalities \eqref{swapfideqin}. The third inequality in \eqref{lowupbdsQT2} is an immediate consequence of \eqref{lowupbdsQT1}, the first and second one follow from Theorem \ref{Wa2metthm}.
 The last statement is a consequence of equality \eqref{QOTrankone}, which holds if either of the states is pure, and point $(f)$ of Proposition \ref{swapfidprop}.
\end{proof}

Proposition \ref{swapfidprop} shows that the introduced SWAP-fidelity shares many features with the standard quantum fidelity introduced by Uhlmann \cite{Uh76} and Jozsa \cite{Jo94}. Notably, $\mathrm{F_S}$ is jointly concave, as is the root fidelity, $\sqrt{\mathrm{F}}$, but not the fidelity itself \cite{BZ17}. It is also remarkable that the SWAP-fidelity is supermultiplicative with respect to the tensor product, as is the superfidelity introduced in \cite{MPHUZ}. This feature might prove relevant for applications in quantum machine learning -- see \cite{FECZ} and references therein.

}

\section{Quantum Optimal Transport as  a  Semidefinite programming  problem}\label{sec:QOT}
One of the main results of this paper is the observation that the computation of the quantum transport is carried out efficiently using semidefinite programming (SDP) \cite{VB96}.  The main results of this section are the statement of QOT as the direct and the dual semidefinite programs:
\begin{theorem}\label{QOTSDP} Assume that $C\in \rS(\cH_m\otimes \cH_n)$, $\rho^A\in\Omega_m, \rho^B\in\Omega_n$.
Then the computation of $\rT_{C}^Q(\rho^A,\rho^B)$ is a semidefinite programming problem.  The value  of $\rT_{C}^Q(\rho^A,\rho^B)$ can be approximated within precision  $\varepsilon>0$ in polynomial time in the size of the data and $\log1/\varepsilon$.
\end{theorem}
\begin{theorem}\label{dualQOT} Assume that $\rho^A\in\Omega_m,\rho^B\in\Omega_n$ and $C\in \rS(\cH_m\otimes \cH_n)$.  Then the dual problem to \eqref{defkapCAB} is 
\begin{equation}\label{dualQOT1}
\sup\{\tr \sigma^A  \rho^A +\tr \sigma^B  \rho^B, \; \sigma^A\in\rS(\cH_m), \: \sigma^B\in \rS(\cH_n), \:  C-\sigma^A\otimes \mathbb{I}_n-\mathbb{I}_m\otimes \sigma^B\ge 0\}.
\end{equation}
Furthermore, the above supremum is equal to $\rT_{C}^Q( \rho^A, \rho^B)$.
Moreover, for a coupling matrix $\rho^{AB}\in \Gamma^Q(\rho^A,\rho^B)$ and $F=C-\sigma^A\otimes \mathbb{I}_n-\mathbb{I}_m\otimes \sigma^B\ge 0$ the following complementary implication holds: 
\begin{equation}\label{orthcond}
\tr F\rho^{AB}=0 \iff \tr C \rho^{AB}=\tr \sigma^A  \rho^A +\tr \sigma^B  \rho^B=\rT_{C}^Q( \rho^A, \rho^B).
\end{equation}
In particular, if $\tr F\rho^{AB}=0$ then $\rank F\le mn -\rank \rho^{AB}$. 

Assume that $ \rho^A, \rho^B> 0$. Then the above supremum is achieved:  There exist $\sigma^A\in\rS(\cH_m), \sigma^B\in \rS(\cH_n)$ such that
\begin{equation}\label{maxsig12}
\rT_{C}^Q( \rho^A, \rho^B)=\tr (\sigma^A\rho^A+\sigma^B\rho^B),\quad C-\sigma^A\otimes \mathbb{I}_n-\mathbb{I}_m\otimes\sigma^B\ge  0.
\end{equation}
\end{theorem} 
We remark that the equality \eqref{dualQOT1} is stated in \cite[(4.2)]{CHW19}.\subsection{Proofs of Theorems \ref{QOTSDP} and \ref{dualQOT}}\label{subsec:prfthm3.12}
{\textsc{Proof of Theorem \ref{QOTSDP}.}
Assume that $\rho^A=[a_{ij}]\in\Omega_m, \rho^B=[b_{pq}]\in\Omega_n$.
Denote the entries of the Hermitian matrix $C$ by $c_{(i,p)(j,q)}$, i.e.,  $c_{(i,p)(j,q)}=\overline{c_{(j,q)(i,p)}}$. Let $\bi=\sqrt{-1}$, and
\begin{align*}
& E_{ij}^A=|i\rangle\langle j|, && G_{ij}^A=\frac{1}{2}(E_{ij}^A+E_{ji}^A), && H_{ij}^A=\frac{1}{2}\bi (E_{ij}^A-E_{ji}^A),&& i,j\in[m],\\
& E_{pq}^B=|p\rangle\langle q|, && G_{pq}^B=\frac{1}{2}(E_{pq}^B+E_{qp}^B), && H_{pq}^B=\frac{1}{2}\bi (E_{pq}^B-E_{qp}^A), && p,q\in[n].
\end{align*}
Thus $|i\rangle,i\in[m]$, $E_{ij}^A,i,j\in[m]$,  $G_{ij}^A, 1\le i\le j\le m, H_{ij}^A , 1\le i< j\le m$ are the standard bases in $\C^m$,  $\C^{m\times m}$, and in the subspace of $m\times m$ Hermitian matrices respectively.  A similar observation applies when we replace $A$ and $m$ by $B$ and $n$.
The conditions $\tr_B \rho^{AB}=\rho^A, \tr_A \rho_{AB}=\rho^B$ are stated as the following linear conditions:
\begin{equation}\label{margcond}
\begin{aligned}
&\tr \rho^{AB}(G_{ij}\otimes \mathbb{I}_n)=\Re a_{ij}, && i\le j, && \tr \rho^{AB} (H_{ij}\otimes \mathbb{I}_n)=\Im a_{ij}, && i<j,\\
&\tr \rho^{AB}(\mathbb{I}_m\otimes G_{pq})=\Re b_{pq}, && p\le q, && \tr \rho^{AB} (\mathbb{I}_m\otimes H_{pq})=\Im b_{pq}, && p<q.
\end{aligned}
\end{equation} 
Here $\Re z,\Im z$ are the real and the imaginary part of the complex number $z\in\C$.
We assume that $ \rho^{AB}\geq 0$.  
Hence $\rT_{C}^Q(\rho^A,\rho^B)$ is a semidefinite problem for $\rho^{AB}$.

Assume first that $\rho^A, \rho^B$ are positive definite.  Then $ \rho^A\otimes  \rho^B$, viewed as a Kronecker tensor product, is positive definite. Thus $\Gamma^Q(\rho^A,\rho^B)$ contains a positive definite operator $ \rho^A\otimes  \rho^B$.  The standard SDP theory \cite[Theorem 5.1]{VB96} yields that $\rT_{C}^Q(\rho^A,\rho^B)$ can be computed in polynomial time with precision $\varepsilon>0$.  

(Note that the standard SDP is stated for real symmetric positive semidefinite matrices.  It is well known that Hermitian positive semidefinite matrices can be encoded as special real symmetric matrices of double dimension.  See the proof of Theorem \ref{dualQOT} for details.)

Assume that $\rho^A,\rho^B\ge 0$.  Then the restrictions 
$\rho^{A'}=\rho^A|_{\range\rho^A}$ and $\rho^{B'}=\rho^B|_{\range\rho^B}$ are positive definite.  Use
Proposition \ref{redprop}
to deduce that $\rT_{C}^Q(\rho^A,\rho^B)$ can be computed in polynomial time in precision $\varepsilon>0$.\qed

{\textsc{Proof of Theorem \ref{dualQOT}. }}
Let us first consider the simplified case where $ \rho^A, \rho^B,  C$ are real symmetric.  Let $\rS_k\supset \rS_{k,+}\supset \rS_{k,+,1}$ be the space of $k\times k$ real symmetric matrices, the cone of positive semidefinite matrices and the convex set of real density matrices.
Define 
\begin{align*}
& \Gamma^Q( \rho^A, \rho^B,\R)=\rS_{mn,+,1}\cap \Gamma^Q( \rho^A, \rho^B),\\ 
& \rT_{C}^Q( \rho^A, \rho^B,\R)=\min_{  \rho^{AB}\in \Gamma^Q( \rho^A, \rho^B,\R)}\tr  C  \rho^{AB}.
\end{align*}
We claim that the dual problem to $\rT_{C}^Q( \rho^A, \rho^B,\R)$ is given by
\begin{equation}\label{dualQOT1R}
\sup\{\tr \sigma^A \rho^A +\tr \sigma^B \rho^B, \sigma^A\in\rS_m, \sigma^B\in \rS_n, C-\sigma^A\otimes \mathbb{I}_n-\mathbb{I}_m\otimes \sigma^B\ge 0\}.
\end{equation}
Indeed, the conditions  $\tr_B  \rho^{AB}= \rho^A, \tr_A  \rho^{AB}= \rho^B$ for $ \rho^{AB}\in \rS_{mn,+}$ are stated as the linear conditions given by the first part of \eqref{margcond}.  Assume that $ \rho^A=[a_{ij}]\in \Omega_m,  \rho^B=[b_{ij}]\in\Omega_n$.  Recall the definition of the matrices $G_{ij,m}$ introduced in the beginning of the proof of Theorem \ref{QOTSDP}.
Then the standard dual characterization of the above semidefinite problem over   $\Gamma^Q( \rho^A, \rho^B,\R)$ has the following form (see~\cite[Theorem 3.1]{VB96} or~\cite[(2.4)]{FriSDP}):
\begin{align*}
 \max \Big\{ & \sum_{1\le i \le j\le m} a_{ij} \tilde u_{ij} + \sum_{1\le p \le q\le n} b_{pq}\tilde v_{pq}, \quad \tilde u_{ij},\tilde v_{pq}\in\R, \\
& \Big(\sum_{1\le i \le j\le m} \tilde u_{ij}(G_{ij,m}\otimes \mathbb{I}_n) + \sum_{1\le p \le q\le n}\tilde v_{pq}(\mathbb{I}_m\otimes G_{pq,n})\Big) \le   C \Big\}.
\end{align*}
Let 
\begin{eqnarray*}
\sigma^A=\sum_{1\le i \le j\le m} \tilde u_{ij}G_{ij,m}, \quad \sigma^B= \sum_{1\le p \le q\le n}\tilde v_{pq} G_{pq,n}.
\end{eqnarray*}
Then the last condition of the above maximum is $\sigma^A\otimes \mathbb{I}_n+ \mathbb{I}_m\otimes \sigma^B\le  C$.  Next observe that
\begin{eqnarray*}
\tr \sigma^A \rho^A+\tr \sigma^B \rho^B=
\bigl(\sum_{1\le i \le j\le m} a_{ij} \tilde u_{ij}\bigr) + \bigl(\sum_{1\le p \le q\le n} b_{pq}\tilde v_{pq}\bigr)
\end{eqnarray*}
Hence the dual to $\rT_{C}^Q( \rho^A, \rho^B,\R)$ is given by \eqref{dualQOT1R}.
Observe that we can choose $\sigma^A=-a\mathbb{I}_m, \sigma^B=0$, where $a$ is a positive big 
number such that $$C -\sigma^A \otimes \mathbb{I}_n-\mathbb{I}_m\otimes\sigma^B=C+a\I_{mn}>0 .$$ Hence the duality theorem \cite[Theorem 3.1]{VB96} yields that the supremum \eqref{dualQOT1R} is equal  to $\rT^Q_{C}( \rho^A, \rho^B,\R)$.  Assume that $ \rho^A, \rho^B> 0$.  Then $0< \rho^A\otimes  \rho^B\in\Gamma^Q( \rho^A, \rho^B,\R)$.  Theorem 3.1 in \cite{VB96} yields that the supremum \eqref{dualQOT1R} is achieved.

We now discuss the Hermitian case.  Let $\bi=\sqrt{-1}$. There is a standard injective map $L:\rS(\cH_m) \to \rS_{2m}$: 
\begin{eqnarray*}
L(X+\bi Y)=\begin{bmatrix} X&Y\\-Y&X\end{bmatrix}, \quad X,Y\in \R^{m\times m}, X^\top=X, Y^\top=-Y.
\end{eqnarray*}
Note that $L(X+\bi Y) \ge 0\iff X+\bi Y \ge 0$ and $L(X+\bi Y)> 0\iff X+\bi Y> 0$. 
Hence it is possible to translate an SDP problem over Hermitian matrices to an SDP problem over reals.  This yields the proof that the supremum in \eqref{dualQOT1} is equal to $\rT_{C}^Q( \rho^A, \rho^B)$.

Assume that  $\rho^{AB}\in \Gamma^Q(\rho^A,\rho^B)$ and $F=C-\sigma^A\otimes \mathbb{I}_n-\mathbb{I}_m\otimes \sigma^B\ge 0$.  As $\rho^{AB}$ and $F$ are positive semidefinite we obtain
\begin{eqnarray*}
0\le \tr F\rho^{AB}=\tr C \rho^{AB} -\tr \sigma^A\rho^A-\tr\sigma^B\rho^B
\end{eqnarray*}
The characterization \eqref{dualQOT1} yields the implication \eqref{orthcond}.  
As $F$ and $\rho^{AB}$ are positive semidefinite the condition $\tr F\rho^{AB}=0$  yields that 
$\rank F+\rank \rho^{AB}\le mn$.

Assume that $ \rho^A, \rho^B> 0$.  Then the above arguments show that the supremum in \eqref{dualQOT1} is achieved. \qed

 In Subsection \ref{subsec:nonexisF} we give an example of $\rho^A,\rho^B\in\Omega_2$, where $\rho^A$ is a pure state,  for which the supremum  \eqref{dualQOT1} is not achieved.
Note that the dual problem has an advantage over the original problem, as we are not constrained by linear conditions \eqref{margcond}.   Also the number of variables is smaller, as the supremum is restricted to $\rS(\cH_m)\times \rS(\cH_n)$.  However we have to deal with the condition $  \sigma^A\otimes \mathbb{I}_n+\mathbb{I}_m\otimes  \sigma^B\le C$.  

%
%
\section{Comparison of classical and quantum optimal transports for diagonal density matrices}\label{sec:diagdm}
The main result of this section is a technical elaboration  of the statement from \cite{CGP20} that the cost of quantum optimal transport is cheaper than that of the classical optimal transport.  This result is stated in Theorem \ref{clqotdiagdm}.  In Subsection \ref{subsec:QOTCOT}  we give a simple example in which the cost 
of the classical optimal transport is seven times
 higher than
  the cost of the quantum optimal transport.
In Subsection \ref{sec:decoh} we construct a family of 
matrices $C^Q_{\alpha}=\alpha C^Q +(1-\alpha)\diag(C^Q)$ for $\alpha\in[0,1]$, which interpolates between the quantum cost matrix $\alpha=1$ and its classical counterpart $\alpha=0$.  We show that in general the optimal transport cost $\rT^Q_{C_{\alpha}^Q}$ strictly decreases as a function of 
the coherence parameter $\alpha$.

\begin{lemma}\label{diagdecr} Assume that $\rho^A,\rho^B\in\Omega_n$ and $C_E^Q$ is defined by \eqref{defCEop}.
Then
\begin{equation}\label{diagdecr1}
 \rT_{C^Q_E}^Q(\diag(\rho^A),\diag(\rho^B))\le\rT_{C^Q_E}(\rho^A,\rho^B).
\end{equation}
\end{lemma}
\begin{proof}  Without loss of generality we can assume that the basis 
$|1\rangle,\ldots,|n\rangle$ used in \eqref{defCEop} is the standard orthonormal basis in $\cH_n=\C^n$.
Denote  by $\cD_n\subset \C^{n\times n}$ the subgroup of diagonal matrices whose diagonal entries are $\pm 1$.  Note that $|\cD_n|=2^n$ and $\cD_n$ is a subgroup of unitary matrices.   Observe next that, for $D\in\cD_n$,
\begin{align*}
(D\otimes D) |\psi_{ij}^-\rangle \langle\psi_{ij}^-| (D\otimes D)= |\psi_{ij}^-\rangle \langle\psi_{ij}^-| && \Rightarrow && (D\otimes D)C^Q_E(D\otimes D)=C^Q_E.
\end{align*}
Hence $\rT_{C^{Q}_E}^Q(\rho^A,\rho^B)=\rT_{C^{Q}_E}^Q(D\rho^A D,D\rho^B D)$ for each $D\in\cD_n$.  Clearly,
\begin{equation*}
\diag(\rho^A)=2^{-n}\sum_{D\in\cD_n} D\rho^A D, \qquad \diag(\rho^B)=2^{-n}\sum_{D\in\cD_n} D\rho^B D.
\end{equation*}
Use the convexity of $\rT_{C^{Q}_E}^Q(\rho^A,\rho^B)$ to obtain 
\begin{equation*}
\rT_{C^{Q}_E}(\diag(\rho^A),\diag(\rho^B))\le 2^{-n}\sum_{D\in\cD_n} \rT_{C^{Q}_E}^Q(D\rho^A D,D\rho^B D)=\rT_{C^{Q}_E}^Q(\rho^A ,\rho^B ).\qedhere
\end{equation*}
\end{proof}

Assume that $\bp^A\in\Pi_m,\bp^B\in \Pi_n$.
The following lemma gives the  isomorphism of $\Gamma^{cl}(\bp^A,\bp^B)$
and $\Gamma^Q_{de}(\diag(\bp^A),\diag(\bp^A))$ 
mentioned
 in the Introduction.
Furthermore, it describes a special $\rho^{AB}\in\Gamma^Q(\diag(\bp^A), \diag(\bp^B))$ induced by $\bp^{AB}\in \Gamma^{cl}(\bp^A,\bp^B)$.
\begin{lemma}\label{diaglemobs}  Let $\rho^A \in\Omega_m, \rho^B\in\Omega_n$ and assume that  $\bp^A\in\Pi^m,\bp^B\in\Pi_n$ are induced by the diagonal entries of $\rho^A,\rho^B$ respectively.  Then 
\begin{enumerate}[(a)]
\item Each matrix $X=[x_{ip}]_{i\in[m],p\in[n]}\in \Gamma^{cl}(\bp^A,\bp^B)$ induces the following two matrices 
\begin{flalign*}
&& R=[r_{(i,p)(j,q)}],\tilde R=[\tilde r_{(i,p)(j,q)}] \in\Gamma^Q(\diag(\bp^A),\diag(\bp^B)), \;\;\, i,j\in[m], p,q\in[n].
\end{flalign*} 
The matrix $R$ is diagonal with $r_{(i,p)(i,p)}=x_{ip}$ for $i\in[m], \, p\in[n]$, and $ \tilde R- R$ is a matrix whose only possible nonzero entries are the entries $((i,p)(p,i))$ for $i,p\in[\min(m,n)]$ and  $i\ne p$ which are equal to $\sqrt{x_{ip}x_{pi}}$. Furthermore, $\rank \tilde R\le mn-\min(m,n)(\min(m,n)-1)/2$.
\item Each matrix $R=[r_{(i,p)(j,q)}]\in\Gamma^Q(\rho^A,\rho^B)$ induces the following two matrices: first,
$X=[x_{ip}]\in \Gamma^{cl}(\bp^A,\bp^B)$, where $x_{ip}=r_{(i,p)(i,p)}$ for $i\in [m],p\in[n]$.
Second, $\hat R\in \Gamma^Q(\diag(\rho^A),\diag(\rho^B))$, which is obtained by replacing the entries of $R$ at places $((i,p)(j,q))$ by zero unless either $((i,p)(j,q))= ((i,p)(i,p))$ for $i\in[m], p\in[n]$ or $((i,p)(j,q))= ((i,p)(p,i))$ for $i,p\in[\min(m,n)], i\ne p$.
\end{enumerate}
\end{lemma}
\begin{proof} (a) As $X\in \Gamma^{cl}(\bp^A,\bp^B)$ we deduce
\begin{equation*}
\sum_{j=1}^n x_{ij}=p^{A}_i,\,i\in[m], \qquad \sum_{i=1}^m x_{ij}=p^B_j, j\in[n].
\end{equation*}
Assume that $R$ is a diagonal matrix with $r_{(i,p)(i,p)}=x_{ip}$.  Use 
\eqref{ptfrom} to deduce that  $R\in\Gamma^Q(\diag(\bp^A),\diag(\bp^B))$.

Consider now the matrix $\tilde R$.  In view of \eqref{ptfrom} we deduce that $\tr_B \tilde R=\diag(\bp^A)$ and  $\tr_A \tilde R=\diag(\bp^B)$.  It is left to show that $\tilde R$ is positive semidefinite.
Observe that $\tilde R$ is a direct sum of $\big(mn-\min(m,n)(\min(m,n)-1)\big)$ blocks of size one and 
$\big(\min(m,n)(\min(m,n)-1)/2\big)$ blocks of size two: $[x_{ii}]$ for $i\in[\min(m,n)]$, $[x_{ip}]$ for $i\in[m],p\in[n],\max(i,p)>\min(m,n)$, and
\begin{equation}\label{defXij}
 X_{ip}= \begin{bmatrix}x_{ip}&\sqrt{x_{ip}x_{pi}}\\\sqrt{x_{ip}x_{pi}}&x_{pi}\end{bmatrix}, \quad \textrm{ for } 1\le i<p\le \min(m,n). 
 \end{equation}
   As $X\ge 0$ each block is positive semidefinite 
    and has rank at most $1$.  Hence $\rank \tilde R\le mn- \min(n,n)(\min(m,n)-1)/2$. 

\noindent
(b) Assume that $R\in\Gamma^Q(\rho^A,\rho^B)$.  
As $R$ is positive semidefinite we deduce that $r_{(i,p)(i,p)}\ge 0$.  The above arguments yield that 
the matrix $X=[r_{(i,p)(i,p)}]\in\Gamma^{cl}(\bp^A,\bp^B)$.  Observe next that $\hat R$ is a direct sum of blocks of size one and two: $[r_{(i,i)(i,i)}]$ for $i\in[\min(m,n)]$, 
$[r_{(i,p)(i,p)}]$ for $\max(i,p) >[\min(m,n)]$, and
\begin{equation}\label{defRij}
R_{ip}=\begin{bmatrix}r_{(i,p)(i,p)}&r_{(i,p)(p,i)}\\r_{(p,i)(i,p)}&r_{(p,i)(p,i)}\end{bmatrix}, \quad  \textrm{ for } 1\le i<p\le \min(m,n).  
\end{equation}
Clearly all these blocks of size one and two are principal submatrices of $R$.  As $R$ is positive semidefinite, each such submatrix is positive semidefinite.  Hence $\hat R$ is positive semidefinite.  
Use \eqref{ptfrom} to deduce that $\tr_B \hat R=
 \diag(\bp^A)$ and 
  $\tr_A \hat R=
 \diag(\bp^B)$.
%
\end{proof}
\begin{theorem}\label{clqotdiagdm}  Assume that $\bp^A\in\Pi_m, \bp^B\in\Pi_n$ are induced by the diagonal entries of  $\rho^A\in\Omega_m,\rho^B\in\Omega_n$ respectively.   Let $C=[C_{(i,p)(j,q)}]$ for $i,j\in[m],p,q\in[n]$ be a Hermitian matrix.  Define $C^{cl}=[C^{cl}_{ip}]$ by $C^{cl}_{ip}=C_{(i,p)(i,p)}$ for $i\in[m],p\in[n]$.   Let $\Gamma_{de}^Q(\diag(\bp^A),\diag(\bp^B)) \subset \Gamma^Q(\diag(\bp^A),\diag(\bp^B))$ be the subset of diagonal matrices.  Define 
\begin{equation*}
\rT^Q_{C,de}(\diag(\bp^A),\diag(\bp^B))=\min_{R\in\Gamma^Q_{de}(\bp^A,\bp^B)}\tr CR.
\end{equation*}
Then 
\begin{enumerate}[(a)]
\item
\begin{equation}\label{OTgeQOT}
\begin{aligned}
\rT_{C^{cl}}^{cl}(\bp^A,\bp^B) & =\rT_{C,de}^Q(\diag(\bp^A),\diag(\bp^B))=
\rT_{\diag(C)}^Q(\rho^A,\rho^B) \\
& \ge \rT_{C}^Q(\diag(\bp^A),\diag(\bp^B)).
\end{aligned}
\end{equation}
\item
Assume that $m\le n$, and $C^Q=[C_{(i,p)(j,q)}^Q]\in \rS_+(\cH_n\otimes \cH_n)$.  
Denote by $C^Q_{m,n}\in\rS_+(\cH_m\otimes\cH_n)$ the submatrix of $C^Q$ whose entries are $C_{(i,p)(j,q)}^Q$ for $i,j\in[m], p,q\in[n]$.  Let $C^{cl}_{m,n}$ be the $m\times n$ nonnegative matrix induced by the diagonal entries of $C^Q_{m,n}$.
Then
\begin{align}\label{clqotdiagdm0}
& \rT_{C^{cl}_{m,n}}^{cl}(\bp^A,\bp^B) =
\frac{1}{2}\min_{X\in\Gamma^{cl}(\bp^A,\bp^B)}\Big(\sum_{1\le i<p\le m} ( x_{ip}+x_{pi} ) \;\, +\sum_{\substack{1\le i \le m,\\ m+1\le p\le n}} x_{ip}\Big) , \\
& \rT_{C^Q_{m,n}}^Q \big(\diag(\bp^A),\diag(\bp^B) \big)  \notag\\
& \qquad\;\; =\frac{1}{2}\min_{X\in\Gamma^{cl}(\bp^A,\bp^B)}\Big(\sum_{1\le i<p\le m} \big(x_{ip}+x_{pi}-2\sqrt{x_{ip}x_{pi}} \big) \; +\sum_{\substack{1\le i \le m, \\m+1\le p\le n}} x_{ip}\Big). \notag
\end{align}
\item Suppose that $m=n=2$.  Assume that $\bs=(s_1,s_2)^\top, \bt=(t_1,t_2)^\top$ are two probability vectors.  Then 
\begin{equation}\label{QOTfdiagorm2b}
\rT_{C^{Q}}^Q(\diag(\bs),\diag(\bt))=\frac{1}{2}
\begin{cases}
(\sqrt{s_1}-\sqrt{t_1})^2, \; \textrm{ if } s_2\ge t_1,\\
(\sqrt{s_2}-\sqrt{t_2})^2, \; \textrm{ if } s_2< t_1. 
\end{cases}
\end{equation}
Furthermore
\begin{equation}\label{QOTfdiagorm2a}
\rT_{C^{Q}}^Q(\diag(\bs),\diag(\bt))=\frac{1}{2}\max\big((\sqrt{s_1}-\sqrt{t_1})^2,(\sqrt{s_2}-\sqrt{t_2})^2\big).
\end{equation}
\end{enumerate}
\end{theorem}
\begin{proof} (a)  Let $X=[x_{ij}]\in\Gamma^{cl}(\bp^A,\bp^B))$ correspond to a diagonal matrix 

\noindent
$R\in \Gamma_{de}^Q(\diag(\bp^A),\diag(\bp^B))$ as in Lemma \ref{diaglemobs}.  Then $\tr C^{cl}X^\top=\tr C R$. This shows the first equality in \eqref{OTgeQOT}.
To show the second equality in \eqref{OTgeQOT} observe that for $R\in \Gamma^Q(\rho^A,\rho^B)$ we have  $\tr \diag(C) R=\tr \diag(C)\diag(R)$.
Next observe that $\diag(R)\in\Gamma_{de}^Q(\diag(\bp^A),\diag(\bp^B))$.
As $$\Gamma_{de}^Q(\diag(\bp^A),\diag(\bp^B))\,\subset\, \Gamma^Q(\diag(\bp^A),\diag(\bp^B))$$ we deduce the inequality 
\begin{equation*}
\rT_{C,de}^Q(\diag(\bp^A),\diag(\bp^B))\,\ge\, \rT_{C}^Q(\diag(\bp^A),\diag(\bp^B)).
\end{equation*}
The proof of \eqref{OTgeQOT} is complete.

\noindent
(b)  Let  $R\in \Gamma^Q(\diag(\bp^A),\diag(\bp^B))$.  Define $X\in \Gamma^{cl}(\bp^A,\bp^B)$ and $\hat R\in\Gamma^Q(\diag(\bp^A),\diag(\bp^B))$  as in part (b) of Lemma \ref{diaglemobs}.  Furthermore, let $\tilde R\in\Gamma^Q(\diag(\bp^A),\diag(\bp^B))$  be defined as in part (a) of Lemma \ref{diaglemobs}.
It is straightforward to show that  
\begin{eqnarray*}
\tr \diag(C^Q_{m,n})R=\tr C^{cl}_{m,n} X^\top,\quad
\tr C^Q_{m,n} R=\tr C^Q_{m,n}\hat R.
\end{eqnarray*}
Use the equalities in \eqref{OTgeQOT} to deduce the first equality in  \eqref{clqotdiagdm0}.

We now show the second equality in  \eqref{clqotdiagdm0}.
 As each $R_{ip}$ in \eqref{defRij} is positive semidefinite we deduce
\begin{align*}
& \tr C^Q_{m,n}\hat R \\ 
& \quad = \frac{1}{2}\Big( \sum_{1\le i<p\le m} \big( r_{(i,p)(i,p)} +r_{(p,i)(p,i)}-2\Re r_{(i,p)(p,i)}\big) \: +\!\!\!\!
\sum_{\substack{1\le i\le m,\\ m+1\le p\le n}} r_{(i,p)(i,p)}\Big) \\
& \quad \ge \frac{1}{2}\Big( \sum_{1\le i<j\le m} \big( r_{(i,p)(i,p)}+ r_{(p,i)(p,i)}-2\sqrt{r_{(i,p)(i,p)}r_{(p,i)(p,i)}}\big) \: + \!\!\!\!
\sum_{\substack{1\le i\le m,\\ m+1\le p\le n}} r_{(i,p)(i,p)}\Big) \\
& \quad \ge \frac{1}{2}\Big(\sum_{1\le i<p\le m} \big( x_{ij}+x_{ji}-2\sqrt{x_{ij}x_{ji}}\big) \: + \!\!\!\!
\sum_{\substack{1\le i\le m,\\ m+1\le p\le n}} x_{ip}\Big)=\tr C^Q_{m,n}\tilde R.
\end{align*}
This establishes the second equality in \eqref{clqotdiagdm0}.

\noindent
(c)  Assume that $s_2\ge t_1$.  Then $A=\begin{bmatrix}0&s_1\\t_1&s_2-t_1\end{bmatrix}\in\Gamma^{cl}(\bs,\bt)$.  Therefore 
\begin{equation*}
\rT_{C^{Q}}^Q(\diag(\bs),\diag(\bt))\le \frac{1}{2}(t_1+s_1-2\sqrt{t_1s_1})=\frac{1}{2}(\sqrt{s_1}-\sqrt{t_1})^2.
\end{equation*}  
If $s_1t_1=0$ then $\Gamma^{cl}(\bs,\bt)=\{A\}$, and $\rT_{C^{Q}}^Q(\bs,\bt)= \frac{1}{2}(\sqrt{s_1}-\sqrt{t_1})^2$.

Assume that $s_1t_1>0$.  Then $\Gamma^{cl}(\bs,\bt)$ is an interval $[A,B]$.
Indeed, let $C=\begin{bmatrix}1&-1\\-1&1\end{bmatrix}$.  So $A+tC\in \Gamma_{cl}(\bs,\bt)$ for $t$ small and positive, and $B=A+t_0C$ for some $t_0>0$.  
Let $g(t)=f(A+tC)$ for $t\in[0,t_0]$.  Recall that $g(t)$ is a convex function on $[0,t_0]$. 
Observe next that 
$$g'(0+)=\frac{1}{2}\big(-2+s_1^{-1/2}t_1^{1/2}+s_1^{1/2}t_1^{-1/2}\big)=\frac{1}{2}s_1^{-1/2}t_1^{-1/2}\big(\sqrt{s_1}-
\sqrt{t_1} \big)^2\ge 0.$$
Hence $g(t)\ge g(0)$ for $t\in[0,t_0]$.   This proves \eqref{QOTfdiagorm2b} for $s_2\ge t_1$.

To show the equality \eqref{QOTfdiagorm2a}  for $s_2\ge t_1$ we need 
 to show that  $\big(\sqrt{s_1}-\sqrt{t_1})^2\ge \big(\sqrt{s_2}-\sqrt{t_2})^2$.  Let $x\in[0,1/2]$.  Observe that the function $\sqrt{1/2+x}+\sqrt{1/2-x}$ is strictly  decreasing on $[0,1/2]$.   Hence
\begin{align*}
&\sqrt{s_1}+\sqrt{s_2}\le \sqrt{t_1}+\sqrt{t_2} \iff \max(s_1,s_2)\ge \max(t_1,t_2),\\
&\sqrt{s_1}+\sqrt{s_2}\ge \sqrt{t_1}+\sqrt{t_2} \iff \max(s_1,s_2)\le \max(t_1,t_2).
\end{align*}

Suppose first that $s_2\ge t_2$.  Hence $s_2\ge \max(t_1,t_2)$, and $s_1=1-s_2\le1-t_2=t_1$.  Thus 
\begin{equation*}
|\sqrt{s_1}-\sqrt{t_1}|=\sqrt{t_1}-\sqrt{s_1}\ge \sqrt{s_2}-\sqrt{t_2}=|\sqrt{s_2}-\sqrt{t_2}|.
\end{equation*}

Suppose second that $s_2<t_2$.  Hence $t_2\ge s_1>t_1$.  Thus $\max(t_1,t_2)\ge \max(s_1,t_1)$.  Hence
 \begin{equation*}
|\sqrt{s_1}-\sqrt{t_1}|=\sqrt{s_1}-\sqrt{t_1}\ge \sqrt{t_2}-\sqrt{s_2}=|\sqrt{s_2}-\sqrt{t_2}|.
\end{equation*}
This proves \eqref{QOTfdiagorm2a} in the case $s_2\geq t_1$.   
Similar arguments proves \eqref{QOTfdiagorm2b} and \eqref{QOTfdiagorm2a} in the case $s_2<t_1$.
\end{proof}

On the set of rectangular matrices $\R^{m\times n}$, where $m\le n$, define 
\begin{equation}\label{deff(X)}
\begin{aligned}
f(X)=\frac{1}{2}\Big( \sum_{1\le i<p\le m} \big( x_{ip}+x_{pi}-2\sqrt{x_{ip}x_{pi}}\big)+\sum_{\substack{1\le i\le m,\\ m+1\le p\le n}} x_{ip} \Big), \\
X=[x_{ip}]\in \R_+^{m\times n}.
\end{aligned}
\end{equation} 
Note that the second sum is zero if $m=n$.
As the function $\sqrt{xy}$ is a concave function on $\R_+^2$ it follows that $f(X)$ is a convex function on $\R^{m\times n}_+$.  Hence $\rT_{C^{Q}_{m,n}}^Q(\diag(\bp^A),\diag(\bp^B))$ is the minimum of the convex function $f(X)$ on $\Gamma^{cl}(\bp^A,\bp^B)$.  Therefore this minimum can be computed in polynomial time within precision $\varepsilon>0$.
\begin{remark}\label{remextCQE}  Note that  the second equality in \eqref{clqotdiagdm0} can be extended for a broad class of cost matrices
$C^Q_E$ introduced in \eqref{defCEop}.
\end{remark}

Lemma 11 in \cite{YZYY19} shows that 
\begin{equation}\label{YZYYdub}
\rT_{C^Q}^Q(\diag(\bs),\diag(\bt))\le\frac{1}{2}\Big( \sum_{i=1}^n (\sqrt{s_i}-\sqrt{t_i})^2-\min_{j\in[n]} (\sqrt{s_j}-\sqrt{t_j})^2\Big),
\end{equation}
 for $\bs,\bt\in\Pi_n$.
Moreover, Algorithm 1 in \cite{YZYY19} gives $X\in \Gamma^{cl}(\bs,\bt)$ such that $f(X)$ is bounded from above by the right hand side of \eqref{YZYYdub}.
Note \eqref{QOTfdiagorm2a} yields that for $n=2$ the inequality \eqref{YZYYdub} is sharp.
\subsection{Quantum optimal transport is less than classical optimal transport}\label{subsec:QOTCOT}
In this subsection we give a simple example on $\cH_2\otimes\cH_2$ that shows that
the quantum optimal transport is less than the classical optimal transport by a factor of seven.   That is, strict inequality holds in \eqref{OTgeQOT}.
A different example is given in \cite{CGP20}.
It is straightforward to show that $C^Q$ on $\cH_2\otimes\cH_2$, given by \eqref{defT}, is equal to
\begin{equation}\label{CQ2form}
\frac{1}{2}\begin{bmatrix}0&0&0&0\\0&1&-1&0\\0&-1&1&0\\0&0&0&0\end{bmatrix}.
\end{equation}

Let $\bs=(16/25,9/25)^\top, \bt=(9/25,16/25)^\top$.   The equality \eqref{QOTfdiagorm2a}
yields
\begin{equation*}
\rT_{C^Q}^Q(\diag(\bs),\diag(\bt))=\frac{1}{2}\Big(\sqrt{\frac{16}{25}}-\sqrt{\frac{9}{25}}\Big)^2=\frac{1}{50}.
\end{equation*}
Then the set of all coupling matrices is
\begin{equation*}
\Gamma^{cl}(\bs,\bt)=\bigg\{ \begin{bmatrix}x&16/25-x\\9/25-x&x\end{bmatrix}, \quad 0\le x\le 9/25 \bigg\}.
\end{equation*}
The classical cost matrix induced by 
the diagonal entries of  
 $C^Q$  is
\begin{equation}\label{ClascostC2}
C^{cl}=\frac{1}{2}\begin{bmatrix}0&1\\1&0\end{bmatrix}.
\end{equation}
Then, Formula \eqref{COTP} yields
\begin{equation*}
\rT_{C^{cl}}^{cl}(\bs,\bt)=\min_{X\in \Gamma^{cl}(\bs,\bt)} \tr C^{cl} X^\top=\frac{1}{2}\min_{x\in[0,9/25]}(1-2x)=\frac{7}{50}.
\end{equation*}

\subsection{Decoherence of the quantum cost matrix}
\label{sec:decoh}
Let us denote
\begin{equation}\label{defCqalpha}
C^{Q}_{\alpha}=\frac{1}{2}\begin{bmatrix}0&0&0&0\\0&1&-\alpha&0\\0&-\alpha&1&0\\0&0&0&0
\end{bmatrix}=\alpha C^Q+(1-\alpha)\diag(C^Q), \quad \alpha\in[0,1].
\end{equation}
Since $C^Q_1$ reduces to $C^Q$ defined in Eq. \eqref{CQ2form}
the number $\alpha$ can be called the coherence parameter.
The case $\alpha=0$ corresponds to the full decoherence, 
as the matrix $C^Q_0$ is diagonal, so that $\tr C^Q_0 \rho^{AB}=\tr C^Q_0\diag(\rho^{AB})$.
The entries of  $\diag(\rho^{AB})$ correspond to  the elements of $\Gamma^{cl}(\bx^A,\bx^B)$, where $\bx^A,\bx^B$ are the probability vectors induced by the diagonal entries of $\rho^A$ and $\rho^B$ respectively. 
 Note that the cost matrix $C^{cl}$ induced by $C^Q_0$ is given by \eqref{ClascostC2}.
Then the quantity
\begin{equation}\label{Talpha}
\rT^Q_\alpha(\rho^A,\rho^B)=\min_{\rho^{AB}\in \Gamma^{Q}(\rho^A,\rho^B)} \tr  C^Q_{\alpha}\rho^{AB}.
\end{equation}
describes a continuous interpolation between the quantum and classical optimal transports, related to the gradual decoherence of the quantum state $\vert \psi^{-} \rangle\langle \psi^-\vert$, which plays the role of the quantum cost matrix $C^Q$. 
We will show that, for two diagonal states, 
  $\rT^Q_\alpha$
  is a  decreasing  function of  the coherence parameter $\alpha$ on $[0,1]$ 
 and provide an exact expression for it.
\begin{lemma}\label{Pstalphaform} Let $\bs,\bt$ be two probability vectors in $\R^2$. 
Assume that $0\le \alpha\le 1$ and denote
\begin{align*}
& \rT^Q(\bs,\bt,\alpha) = \rT^Q_{\alpha} \big(\diag(\bs),\diag(\bt) \big),\\
& f_\alpha(X)=\frac{1}{2} \big( x_{12}+x_{21}-2\alpha\sqrt{x_{12}x_{21}} \, \big), \quad X=[x_{ij}]\in  \Gamma_{cl}(\bs,\bt).
\end{align*}
Then 
\begin{equation}\label{defPalphast}
\rT^Q(\bs,\bt,\alpha)=\min_{X\in \Gamma^{cl}(\bs,\bt)} f_\alpha(X).
\end{equation}
Let $\rT^Q(\bs,\bt,1)=\rT_{C^{Q}}^Q(\diag(\bs),\diag(\bt))$ be given by \eqref{QOTfdiagorm2a}.
Assume that $\rT^Q(\bs,\bt,1)=(\sqrt{s_i}-\sqrt{t_i})^2$.
If either $\min(s_i,t_i)=0$ or $\bs=\bt$ then 
\begin{eqnarray*}
\rT^Q(\bs,\bt,\alpha)= \rT^Q(\bs,\bt,1) \textrm{ for all }\alpha\in[0,1].  
\end{eqnarray*}
Otherwise $\rT^Q(\bs,\bt,\alpha)$ is a strictly decreasing function for $\alpha\in[0,1]$ given by the formula
\begin{equation}\label{Pstalphaform1}
\rT^Q(\bs,\bt,\alpha)=\frac{1}{2}\begin{cases}
\sqrt{1-\alpha^2}|s_i-t_i|, &  \textrm{ for } 0\le \alpha<\frac{2\sqrt{s_it_i}}{s_i+t_i},\\
2\rT^Q(\bs,\bt,1)+2(1-\alpha)\sqrt{s_it_i}, & \textrm{ for } \frac{2\sqrt{s_it_i}}{s_i+t_i}\le \alpha\le 1.
\end{cases}
\end{equation}
\end{lemma}
\begin{proof}  The equality \eqref{defPalphast} is deduced as the second equality in \eqref{clqotdiagdm0}.
Observe next that $C^{Q}_{\alpha}$ is positive semidefinite.  Hence $\rT^Q(\bs,\bt,\alpha)\ge 0$.  Therefore for $\bs=\bt$ we choose $X=\I\in\Gamma^{cl}(\bs,\bt)$ to deduce from  \eqref{defPalphast} that $\rT^Q(\bs,\bt,\alpha)=0$.  Assume that $\min(s_i,t_i)=0$.  Then $\Gamma^{cl}(\bs,\bt)=\{B\}$, where $B$ has one zero off-diagonal element, and $\rT^Q(\bs,\bt,\alpha)= \rT^Q(\bs,\bt,1)$.

Assume that $\min(s_i,t_i)>0$ and $\bs\ne \bt$.  Suppose first 
that $s_2\ge t_1$.    Then for $\alpha=1$ \eqref{QOTfdiagorm2b} yields that $\rT^Q(\bs,\bt,1)=\frac{1}{2}(\sqrt{s_1}-\sqrt{t_1})^2$, i.e., $i=1$.  Thus $\min(s_1,t_1)>0$.  The proof of part (c) of Theorem \ref{clqotdiagdm} implies that the minimum of $f_1(X)$ is achieved at the matrix $A=\begin{bmatrix}0&s_1\\t_1&s_2-t_1\end{bmatrix}$, which is an extreme point of $\Gamma^{cl}(\bs,\bt)$.  As $s_1,t_1>0$ it follows that $\Gamma^{cl}(\bp,\bt)$ is an interval, where
 the second extreme matrix is $C=\begin{bmatrix}\min(s_1,t_1)&s_1-\min(s_1,t_1)\\t_1-\min(s_1,t_1)&s_2-t_1+\min(s_1,t_1)\end{bmatrix}$.  Thus we can move from $A$ to the relative interior of $\Gamma_{cl}(\bs,\bt)$ by considering $A(x)=A+xB$,
 where $B=\begin{bmatrix}1&-1\\-1&1\end{bmatrix}$ and $x>0$. 
Denoting 
\begin{equation*}
g_{\alpha}(x)=f_{\alpha}(A(x))=\frac{1}{2} \big(s_1+t_1-2x -2\alpha\sqrt{s_1-x}\sqrt{t_1-x} \big),
\end{equation*}
one obtains
\begin{equation*}
g_\alpha'(0^+)=\frac{1}{2}\Big[-2+\alpha\Big(\frac{\sqrt{t_1}}{\sqrt{s_1}} +\frac{\sqrt{s_1}}{\sqrt{t_1}}\Big)\Big]. 
\end{equation*}
(Here $h'(x^-)$ and $h'(x^+)$ are the one-sided derivatives of a function $h$ as $x$ is approached from the left and right, respectively.)
Hence this derivative is nonnegative for 
$\alpha\ge \frac{2\sqrt{s_1t_1}}{s_1+t_1}$
and negative for $0\le \alpha<\frac{2\sqrt{s_1t_1}}{s_1+t_1}$.
As $g_{\alpha}(x)$ is convex on the interval $[0, \min(s_1,t_1)]$ we obtain that for $ \frac{2\sqrt{s_it_i}}{s_i+t_i}\le \alpha\le 1$ the minimum of $g_{\alpha}$ for $\frac{2\sqrt{s_1t_1}}{s_1+t_1}$
is achieved at $x=0$.  This proves the second part of \eqref{Pstalphaform1}.
So assume that $0\le \alpha<
\frac{2\sqrt{s_1t_1}}{s_1+t_1}$.  Clearly the minimum of $f_0(X)$ on $\Gamma^{cl}(\bs,\bt)$ is achieved at $A(\min(s_1,t_1))$. 
For $\alpha>0$ we have $g'_{\alpha}(\min(s_1,t_1)^-)=\infty$.

 Hence for $0< \alpha<\frac{2\sqrt{s_1t_1}}{s_1+t_1}$ the minimum $g_{\alpha}(x)$ is achieved at a critical point $x\in (0, \min(s_1,t_1))$.  This critical point is unique, as $g_{\alpha}(x)$ is strictly convex on $(0,\min(s_1,t_1))$ and satisfies the quadratic equation 
\begin{equation}\label{defxalpha}
4(s_1-x)(t_1-x)-\alpha^2(s_1+t_1-2x)^2=0, \quad  0\le \alpha<\frac{2\sqrt{s_1t_1}}{s_1+t_1}.
\end{equation}
We claim that the critical point is given by
\begin{equation*}
x(\alpha)=\frac{1}{2}\Big(s_1+t_1 -\frac{|s_1-t_1|}{\sqrt{1-\alpha^2}}\Big), \quad 
0\le \alpha<\frac{2\sqrt{s_1t_1}}{s_1+t_1}.
\end{equation*}
A direct computation shows that $x(\alpha)$ satisfies \eqref{defxalpha}.
Next observe that as $s_1\ne t_1$ the function $x(\alpha)$ is a strictly decreasing function on $[0,1)$.  Clearly
\begin{equation*}
x(0)=\min(s_1,t_1), \quad
x \Big(\frac{2\sqrt{s_it_i}}{s_i+t_i} \Big)=0.  
\end{equation*}
Hence $x(\alpha)\in (0,\min(s_1,t_1)]$.
Note that for $x(\alpha)$ we have equality
\begin{equation*}
2\sqrt{s_1-x(\alpha)}\sqrt{t_1-x(\alpha)}=\alpha\big(s_1+t_1 -2x(\alpha)\big).
\end{equation*}
This proves the first part of \eqref{Pstalphaform1} in the case for $i=1$.  Similar arguments show the first part of \eqref{Pstalphaform1} in the case for $i=2$.
Clearly for $s_i\ne t_i$ and $\min(s_i,t_i)>0$ the function $\rT^Q(\bs,\bt,\alpha)$ is strictly decreasing on the interval $[0,1]$.
\end{proof}
\section{A lower bound on  the quantum transport cost $\rT_{C^{Q}}^Q(\rho^A,\rho^B)$}\label{subsec:lowbdT}
The main result of this section is  Theorem \ref{lowbdT}.  Inequatity  \eqref{lowbdT0}
will show that  $\sqrt{\rT^Q_{C^Q}(\rho^A,\rho^B)}$ is a weak distance, and
\eqref{lowbdT0a} will allow us to obtain an explicit formula for $\rT_{C^Q}^Q$ for qubits.
\begin{theorem}\label{lowbdT}
Let $\rho^A,\rho^B\in\Omega_n$.
Then the following statements hold:
\begin{enumerate}[(a)]
\item 
 For any $n$ we have 
\begin{equation}\label{lowbdT0} 
\rT_{C^{Q}}^Q(\rho^A,\rho^B)\ge \frac{1}{2}\max_{U\in\rU(n),i\in[n]}\bigg(\sqrt{(U^\dagger\rho^A U)_{ii}}-\sqrt{(U^\dagger\rho^B U)_{ii}}\bigg)^2.
\end{equation}
\item For $n=2$ equality holds in \eqref{lowbdT0}: 
\begin{equation}\label{lowbdT0a} 
\rT_{C^{Q}}^Q(\rho^A,\rho^B)= \frac{1}{2}\max_{U\in\rU(2),i\in[2]}\bigg(\sqrt{(U^\dagger\rho^A U)_{ii}}-\sqrt{(U^\dagger\rho^B U)_{ii}}\bigg)^2.
\end{equation}
\end{enumerate}
\end{theorem}

We first start with the diagonal density matrices.
\subsection{A lower bound on $\rT^Q_{C^Q}\big(\diag(\bs),\diag(\bt)\big)$}\label{subsec:lbdiagdm}
\begin{lemma}\label{compcondlem}  Assume that $\bs,\bt\in\R^n$ are nonnegative probability vectors and $\rho^A = \diag(\bs)$, $\rho^B = \diag(\bt)$.  Then the dual supremum problem \eqref{dualQOT1} can be restricted to diagonal matrices $\sigma^A=-\diag(\ba),\sigma^B=-\diag(\bb)$ for $\ba, \bb\in\R^n$ which satisfy the condition that $F=C^{Q}+\diag(\ba)\otimes \mathbb{I}_n +\mathbb{I}_n\otimes \diag(\bb)$ is positive semidefinite.
 
Let $X^\star=[x_{ij}^\star]\in \Gamma^{cl}(\bs,\bt)$ be a solution to the second minimum problem in \eqref{clqotdiagdm0}, where $\bp^A=\bs,\bp^B=\bt$ and $m=n$.  
Assume that the maximum in the dual supremum problem \eqref{dualQOT1} is achieved by a matrix
of the form $F^\star=C^{Q}+\diag(\ba^\star)\otimes \mathbb{I}_n +\mathbb{I}_n\otimes \diag(\bb^\star)$, where $\rho^A=\diag(\bs),\rho^B=\diag(\bt),\sigma^A=-\diag(\ba),\sigma^B=-\diag(\bb)$.
Then the following equalities hold: 
\begin{equation}\label{compcond}
\begin{aligned}
& x^\star_{ii} (a_i^\star+b_i^\star)=0, \, \textrm{ for }i\in[n],\\
& x_{ij}^\star(a_{i}^\star +b_{j}^\star +1/2) + x_{ji}^\star(a_{j}^\star +b_{i}^\star+1/2) -\sqrt{x_{ij}^\star x_{ji}^\star}=0, \, \textrm{ for }1\le i< j\le n.
\end{aligned}
\end{equation}
Furthermore the following conditions are satisfied 
\begin{enumerate}[(a)]
\item For $i\ne j$ either $x^\star_{ij} x^\star_{ji}>0$ or $x^\star_{ij} =x^\star_{ji}=0$.
\item  Assume that $x^\star_{ii} x^\star_{jj}>0$.  Then $x^\star_{ij} =x^\star_{ji}$.
Let $X(t)$ be obtained from $X^\star$ by replacing the entries $x^\star_{ii}, x^\star_{ij},x^\star_{ji}, x^\star_{jj}$ with $x^\star_{ii}-t, x^\star_{ij}+t,x^\star_{ji}+t, x^\star_{jj}-t$.  Then $X(t)$ is also a solution to the second minimum problem in \eqref{clqotdiagdm0} for $t\in[-x_{ij}^\star, \min(x^\star_{ii},x^\star_{jj})]$.  Furthermore, $a_i^\star=a_j^\star=-b_i^\star=-b_j^\star$.
\item Suppose that $x_{ip}^\star,x_{iq}^\star,x_{jp}^\star,x_{jq}^\star$ are positive for $i\ne j, p\ne q$, where $i,j,p,q\in[n]$.  Then
\begin{equation}\label{compcond1}
\begin{aligned}
& \frac{\sqrt{x_{pi}^\star}}{\sqrt{x_{ip}^\star}}+\frac{\sqrt{x_{qj}^\star}}{\sqrt{x_{jq}^\star}}-\frac{\sqrt{x_{qi}^\star}}{\sqrt{x_{iq}^\star}}-\frac{\sqrt{x_{jp}^\star}}{\sqrt{x_{pj}^\star}}=0, && \textrm{ if } i\ne p, i\ne q, j\ne i, j\ne q,\\
& 1+\frac{\sqrt{x_{qj}^\star}}{\sqrt{x_{jq}^\star}}-\frac{\sqrt{x_{qi}^\star}}{\sqrt{x_{iq}^\star}}-\frac{\sqrt{x_{jp}^\star}}{\sqrt{x_{pj}^\star}}=0, && \textrm{ if } i=p,i\ne q, i\ne j, j\ne q.
\end{aligned}
\end{equation}
\end{enumerate}
Furthermore, there exists a minimizing matrix $X^\star$, satisfying the above conditions, such that it has at most one nonzero diagonal entry even if a maximizing $F^\star$ does not exist.
\end{lemma}
\begin{proof} Let $\ba=(a_1,\ldots,a_n)^\top,\bb=(b_1,\ldots,b_n)^\top\in\R^n$, and consider the matrix $F=C^{Q}+\diag(\ba)\otimes\I_n +\mathbb{I}_n\otimes \diag(\bb)$ .  Then 
$F$ is a direct sum of $n$ blocks of size one of the form $a_i+b_i$ corresponding to the diagonal entries $((i,i),(i,i))$ and $n(n-1)/2$ blocks of size two corresponding to 
the entries $((i,j)(i,j)), ((i,j)(j,i))$, $((j,i)(i,j)), ((j,i),(j,i))$:
\begin{equation}\label{defMij}
M_{ij}=\begin{bmatrix}a_i+b_j+1/2& -1/2\\-1/2& a_j+b_i+1/2\end{bmatrix},  \quad i\in[n]
\end{equation}
Hence $F\ge 0$ if and only if the following inequalities hold:
\begin{align}\label{posdeFcond}
& a_i+b_i\ge 0, \quad \textrm{ for } i\in[n],\\
& a_i+b_j+1/2\ge 0, \, a_j+b_i+1/2\ge 0,\, (a_i+b_j+1/2)(a_j+b_i+1/2)\ge 1/4, \quad i\ne j. \notag
\end{align}
Assume that $G=C^{Q}-\sigma^A\otimes \mathbb{I}_n - \mathbb{I}_n\otimes \sigma^B\ge 0$.  Let $\ba,\bb\in\R^n$ be the vectors obtained from the diagonal entries of $-\sigma^A,-\sigma^B$ respectively.  Observe that the $n$ $1\times 1$  and $n(n-1)/2$   diagonal blocks of $F$ and $G$ discussed above are identical.  As $G$ is positive semidefinite 
then $F$ is positive semidefinite.  Clearly 
\begin{equation*}
\tr \sigma^A\diag(\bs)=-\tr\diag(\ba)\diag(\bs), \quad \tr \sigma^B\diag(\bt)=-\tr\diag(\bb)\diag(\bt).
\end{equation*}
Hence the dual supremum problem \eqref{dualQOT1} can be restricted to diagonal matrices $\sigma^A=-\diag(\ba),\sigma^B=-\diag(\bb)$ for $\ba, \bb\in\R^n$ that satisfy the condition that $F$ is positive semidefinite.

Recall that $X^\star$ induces a solution to the original  SDP
$R^\star\in \Gamma^Q(\diag(\bs),\diag(\bt))$ of the form
described in part (a) of Lemma \ref{diaglemobs}.   That is, the diagonal entries of $R^\star$ 
are $R^\star_{(i,j)(i,j)}=x^\star_{ij}$ with additional nonnegative entries:
$R^\star_{(i,j)(j,i)}=\sqrt{x^\star_{ij}x^\star_{ji}}$ for $i\ne j$. 
Clearly, $R^\star $ is a direct sum of $n$ submatrices of order $1$ and $n(n-1)/2$ of order $2$ as above.
The implication \eqref{orthcond} yields that $\tr F^\star R^\star=0$.
 
 As $F^\star$ is positive semidefinite we deduce
the conditions \eqref{posdeFcond} for $\ba^\star, \bb^\star$.
The blocks  $[x_{ii}^\star]$ and $[a_i^\star+b_i^\star]$  contribute $1$ to the ranks of $R^\star$ and $F^\star$ if and only if $x_{ii}^\star>0$ and $a_i^\star+b_i^\star>0$.
Each $2\times 2$ block of $R^\star$ is of the form $\begin{bmatrix} x_{ij}^\star&\sqrt{x^\star_{ij}x^\star_{ji}}\\\sqrt{x^\star_{ij}x^\star_{ji}}& x_{ji}^\star\end{bmatrix}$ for $1\le i <j\le n$.   Note that the rank of this block is either zero or one.
Each corresponding $2\times 2$ submatrix  of $F^\star$ is of the form 
$M_{ij}^\star$ given by \eqref{defMij}.
Thus $M_{ij}^\star$ is positive semidefinite with rank at least one.
This matrix has rank one if and only if  the following quadratic condition holds:
\begin{equation}\label{3quadcond}
(a_i^\star+ b_j^\star+1/2)(a_j^\star+b_i^\star+1/2)-1/4=0, \textrm{ for } 1\le i < j\le n.
\end{equation}

Recall the complementary condition
\begin{align*}
0 & =\tr R^\star F^\star \\
& =\sum_{i=1}^n x_{ii}^\star(a_{i}^\star+b_{i}^\star)+\!\!\sum_{1\le i <j\le n}\!\big(x_{ij}^\star(a_{i}^\star +b_{j}^\star +1/2) + x_{ji}^\star(a_{j}^\star +b_{i}^\star+1/2) -\sqrt{x_{ij}^\star x_{ji}^\star}\big).
\end{align*}
As all three $1\times 1$ and  $2\times 2 $ corresponding blocks of $R^\star$ and $F^\star$ are positive semidefinite, it follows that we have the complementary conditions \eqref{compcond}.

We now show the second part of the lemma.

\noindent  
(a) Assume that $x_{ij}^\star=0$ for $i\ne j$.  Then the second part of \eqref{compcond} yields $x_{ji}^\star(a_{j}^\star +b_{i}^\star+1/2)=0$.  The second condition in \eqref{defMij} yield that $x^\star_{ji}=0$.

\noindent
(b) Observe that $X(t)\in \Gamma^{cl}(\bs,\bt)$ for $t\in[-\min(x_{ij}^\star,x_{ji}^\star), \min(x^\star_{ii},x^\star_{jj})]$.  Assume first that $x_{ij}^\star x_{ji}^\star>0$.
Let $f(X)$ be defined as in \eqref{deff(X)}.  
As $t=0$ is an interior point of this interval, and $X(0)=X^\star$ we have the critical condition $\left.\frac d{dt}f(X(t))\right|_{t = 0}$, with $f$ given by \eqref{deff(X)}.  This yields the equality
$2 -\frac{\sqrt{x_{ij}^\star}}{\sqrt{x_{ji}^\star}} -\frac{\sqrt{x_{ji}^\star}}{\sqrt{x_{ij}^\star}}=0$.  Hence $x_{ij}^\star=x_{ji}^\star$ and thus $f(X(t))=f(X(0))$ for 
$t\in[-x_{ij}^\star, \min(x^\star_{ii},x^\star_{jj})]$.

Assume now that $x_{ij}^\star=x_{ji}^\star=0$.   
Then $f(X(t))=f(X(0))$ for $t\in [0,\min(x^\star_{ii},x^\star_{jj})]$.

It is left to show that $a_i^{\star}=a_j^\star=-b_i^\star=-b_j^\star$.  First observe that 
the first set of conditions of \eqref{compcond} yield that $a_i^\star+b_i^\star=a_j^\star+b_j^\star=0$.  By replacing $\ba^\star, \bb^\star$ by
$\ba^\star-c\1, \bb^\star+c\1$ we do not change $F^\star$.  Hence we can assume that $a_j^\star=b_j^\star=0$.  Set $b_i^\star=-a_i^\star$.  Then the assumption that the diagonal entries of $M_{ij}^\star$ are nonnegative yields that $|a_i^\star|\le 1/2$.  Use the assumption that $\det M_{ij}^\star\ge 0$ to deduce that $0=a_i^\star=-b_i^\star$.  

\noindent
(c)  Let $X(t)$ be the matrix obtained from $X^\star$ by replacing  $x_{ip}^\star,x_{iq}^\star,x_{jp}^\star,x_{jq}^\star$ with  $x_{ip}^\star-t,x_{iq}^\star+t,x_{jp}^\star +t,x_{jq}^\star-t$. Then for $t\in [-\min(x_{iq}^\star,x_{jp}^\star), \min(x^\star_{ip},x^\star_{jq})]$ we have $X(t)\in\Gamma^{cl}(\bs,\bt)$.  As $t=0$ is an interior point of this interval we deduce that $\left.\frac d{dt}f(X(t))\right|_{t = 0}$.

Suppose first that $i\ne p, i\ne q, j\ne i, j\ne q$. Then Eq. \eqref{deff(X)} yields
\begin{multline*}
\!\!\!\!\! f(X(t))=-\Big(\sqrt{(x^\star_{ip}-t)x^\star_{pi}}+\sqrt{(x^\star_{iq}+t)x^\star_{qi}}+\sqrt{(x^\star_{jp}+t)x^\star_{pj}}+\sqrt{(x^\star_{jq}-t)x^\star_{qj}} \,\Big) \\
\qquad + C,
\end{multline*}
where $C$ is a term that does not depend on $t$.
The condition $\left.\frac d{dt}f(X(t))\right|_{t = 0}$ yields the first condition \eqref{compcond1}.

Assume now that $i=p$ and $i\ne q, j\ne i, j\ne q$.  
Then we have  
\begin{eqnarray*}
f(X(t))=t/2-\big(\sqrt{(x_{iq}^\star+t)x^\star_{qi}}+\sqrt{(x^\star_{jp}+t)x^\star_{pj}}+\sqrt{(x^\star_{jq}-t)x^\star_{qj}} \,\,\big)+ C,
\end{eqnarray*}
where $C$ does not depend on $t$.
Now, the condition $\left.\frac d{dt}f(X(t))\right|_{t = 0}$ yields the second condition in \eqref{compcond1}.

Finally, we need to prove the existence of an $X^\star$ with at most one nonzero entry that satisfies the conditions of the lemma.
Assume first that $\bs,\bt>\0$.  Then Theorem~\ref{dualQOT} yields that there exists a maximizing matrix $F^\star$ to the dual supremum problem.  As we showed above we can assume that $F^\star=C^{Q}+\diag(\ba^\star)\otimes \mathbb{I}_n +\mathbb{I}_n\otimes \diag(\bb^\star)$.  Let $X^\star$ be a minimizing matrix with 
at most $k$ zeros on the diagonal.
Assume to the contrary that  $x^\star_{ii} x^\star_{jj}>0$ for $1\le i < j\le n$.
Part (b) yields that for $t\in[-\min(x_{ij}^\star,x_{ji}^\star), \min(x^\star_{ii},x^\star_{jj})]$ the matrix  $X(t)$ minimizes $f$.  Choose $t^\star=\min(x^\star_{ii},x^\star_{jj})$.  Then $X(t^\star)$ is a minimizing matrix with at least $k+1$ 
zeros on the diagonal,
which contradicts our choice of $X^\star$.

Assume now that $\bs,\bt$ are nonnegative. Let $\bs_k,\bt_k>\0,k\in\N$ be two sequences that converge to $\bs,\bt$ respectively.  Let $X_k^\star$ be a minimizing matrix of $f(X)$ corresponding to $\bs_k,\bt_k$ that has at most one nonzero diagonal element.  Clearly, there exists a subsequence $X^\star_{k_l}$ which has either all zero diagonal elements or exactly one positive diagonal element in a fixed diagonal entry.  Choose a subsequence $[\tilde x^\star_{ij,l}], l\in\N$ of this subsequence which converges to $X^\star$.  Clearly $X^\star$ is a minimizing matrix of $f(X)$ corresponding to $\bs,\bt$.  If $x^\star_{ij}>0$ then $\tilde x^\star_{ij,l}>0$ for $l\gg 1$.  Hence $X^\star$ satisfies the conditions of the lemma.
\end{proof}
\begin{theorem}\label{lowbdTdiagn} Assume that $\bs=(s_1,\ldots,s_n)^\top$, $\bt=(t_1,\ldots,t_n)^\top\in\R^n_+$ are probability vectors and $U\in\rU(n)$.  Then
\begin{equation}\label{eqlowbdiagn}
\rT_{C^{Q}}^Q \big(U^\dagger \diag(\bs) U, U^\dagger\diag(\bt) U \big)\ge\frac{1}{2} \max_{i\in[n]} \, \big( \sqrt{s_i}-\sqrt{t_i} \big)^2
\end{equation}
Equality holds
if and only there exists $i\in[n]$ such that
\begin{equation}\label{eqlowbdiagn1}
\begin{aligned}
\textrm{ either } s_j\ge t_j \textrm{ and } t_it_j\ge s_i s_j \textrm{ for all } j\ne i, \\
\textrm{ or } t_j\ge s_j \textrm{ and } s_is_j\ge t_i t_j \textrm{ for all } j\ne i.
\end{aligned}
\end{equation}
In particular,  for $n=2$ equality in \eqref{eqlowbdiagn} holds:
\begin{equation}\label{eqlowbdiagn2}
\begin{aligned}
\rT_{C^{Q}}^Q \big(U^\dagger \diag((s_1,s_2)^\top) U, U^\dagger\diag((t_1,t_2)^\top) U \big)=\\
\frac{1}{2} \max((\sqrt{s_1}-\sqrt{t_1})^2,(\sqrt{s_2}-\sqrt{t_2})^2).
\end{aligned}
\end{equation}
\end{theorem}
\begin{proof}  Without loss of generality we can assume that  $U=\mathbb{I}_n$.  Suppose first that $\bs,\bt>\0$.  Lemma \ref{compcondlem}  yields that $\rT_{C^{Q}}^Q$ is the maximum of the dual problem where $F=C^{Q}+\diag(\ba)\otimes \mathbb{I}_n + \mathbb{I}_n\otimes \diag(\bb)$ is positive semidefinite.
Choose $i\in[n]$.  Assume that the coordinates of $\ba,\bb$ are given as follows: 
\begin{equation}\label{abichoice}
a_{i}=\frac{1}{2}\Big(\frac{\sqrt{t_i}}{\sqrt{s_i}}-1\Big), \, b_{i}=\frac{1}{2}\Big(\frac{\sqrt{s_i}}{\sqrt{t_i}}-1\Big), \quad a_{j}=b_{j}=0 \textrm{ for } j\ne i.
\end{equation}
Clearly
\begin{align*}
& a_{i}+b_{i}=\frac{(\sqrt{s_i}-\sqrt{t_i})^2}{2\sqrt{s_it_i}}\ge 0,\quad a_{j}+b_{j}=0, && \textrm{ for } j\ne i,\\
& 1/2+a_{i}>0, \,1/2+b_{i}>0, \quad 1/2+a_{j}=1/2+b_{j}=1/2, && \textrm{ for } j\ne i,\\
& (a_{i}+b_{j}+1/2)(a_{j}+b_{i}+1/2)=(a_{i}+1/2)(b_{i}+1/2)=1/4, && \textrm{ for } j\ne i,\\
& (a_{j}+b_{p}+1/2)(a_{p}+b_{j}+1/2)=1/2\times 1/2=1/4, && \textrm{ for }p\ne j\in [n]\setminus\{i\}.
\end{align*}
Thus $F\ge 0$.  Therefore 
\begin{align*}
\rT_{C^{Q}}^Q(\diag(\bs),\diag(\bt)) &\ge  -\tr\big(\diag(\ba)\diag(\bs)+\diag(\bb)\diag(\bt)\big) \\ 
& = \frac{1}{2}\Big[ \Big(1-\frac{\sqrt{t_i}}{\sqrt{s_i}} \Big)s_i + \Big(1-\frac{\sqrt{s_i}}{\sqrt{t_i}}\Big) t_i\Big]=\frac{1}{2} \big(\sqrt{s_i}-\sqrt{t_i} \big)^2. 
\end{align*}
As we let $i\in[n]$ we deduce the inequality \eqref{eqlowbdiagn}.
Since $\rT_{C^{Q}}^Q(\diag(\bs),\diag(\bt))$ is continuous on $\Pi_n\times \Pi_n$ we deduce the inequality  \eqref{eqlowbdiagn} for all $(\bs,\bt)\in \Pi_n\times \Pi_n$.

 We now discuss the equality case in \eqref{eqlowbdiagn}.
Clearly $\max_{i\in[n]}(\sqrt{s_i}-\sqrt{t_i})^2=0$ if and only if $\bs=\bt$, in which case $\rT_{C^{Q}}^Q(\diag(\bs),\diag(\bt))=0$.  Assume that  $\rT_{C^{Q}}^Q(\diag(\bs),\diag(\bt))> 0$.  Suppose first that equality holds in \eqref{eqlowbdiagn}. Then there exists an index $i\in[n]$ such that \mbox{$\rT_{C^Q}(\diag(\bs),\diag(\bt))=\frac{1}{2}(\sqrt{s_i}-\sqrt{t_i})^2>0$}.  By renaming indices and interchanging $\bs$ and $\bt$ if needed we can assume that $t_1>s_1$ and  $\rT_{C^{Q}}^Q(\diag(\bs),\diag(\bt))=\frac{1}{2}(\sqrt{s_1}-\sqrt{t_1})^2$.  
Let $X=X^\star$ be a solution to the second minimum problem in \eqref{clqotdiagdm0}.    
Recall that $f(X^\star)=\frac{1}{2}(\sqrt{t_1}-\sqrt{s_1})^2$.
Suppose first that $s_1=0$.  Then the first row of each $X\in \Gamma^{cl}(\bs,\bt)$ is zero.  Hence
\begin{eqnarray*}
2f(X)=\sum_{j=2}^n x_{j1}+\sum_{2\le j<k\le n}(\sqrt{x_{jk}}-\sqrt{x_{kj}})^2=t_1 +
\sum_{2\le j <k\le n}(\sqrt{x_{jk}}-\sqrt{x_{jk}})^2,
\end{eqnarray*}
for $ X\in\Gamma^{cl}(\bs,\bt)$.  As $f(X^\star)=t_1$ we deduce that the submatrix $Y=[x_{jk}^\star]_{j,k\ge 2}$ is a nonnegative symmetric matrix.  Thus for $j\ge 2$
\begin{eqnarray*}
s_j=\sum_{k=1}^n x_{jk}^\star=x_{j1}^\star +\sum_{k=2}^n x^\star_{jk}=x_{j1}^\star +\sum_{k=2}^n x^\star_{kj}=x_{j1}^\star+t_j.
\end{eqnarray*}
Therefore $s_j\ge t_j$ and $t_1t_j\ge 0=s_1s_j$ for $j\ge 2$.  Hence the conditions \eqref{eqlowbdiagn1} hold.

Assume now that $s_1>0$.
Let $F$ be defined as above for $i=1$.  Our assumption is that $F=F^\star$ is a solution to the maximum dual problem.   Lemma \ref{compcondlem} yields the equalities
\eqref{compcond}.  Hence $x_{11}^\star=0$.  Next consider the second part of the equalities \eqref{compcond} for $i=1$ and $j\ge 2$:
\begin{equation*}
\frac{\sqrt{t_1}}{\sqrt{s_1}}x_{1j}^\star=\frac{\sqrt{s_1}}{\sqrt{t_1}}x_{j1}^\star=c_j\ge 0 \; \textrm{ for } j\ge 2.
\end{equation*}
Observe next that
\begin{eqnarray*}
s_1=\sum_{j=2}^n x_{1j}^\star =\frac{\sqrt{s_1}}{\sqrt{t_1}}\sum_{j=2}^n c_j \quad \Rightarrow \quad \sum_{j=2}^n c_j =\sqrt{s_1t_1}.
\end{eqnarray*}
Therefore
\begin{eqnarray*}
\sum_{j=2}^n \big( x_{1j}^\star+x_{j1}^\star -2\sqrt{x_{1j}^\star x_{j1}^\star} \big) =s_1+t_1-2\sum_{j=2}^n c_j=s_1+t_1-2\sqrt{s_1 t_1}=(\sqrt{s_1}-\sqrt{t_1})^2.
\end{eqnarray*}
Hence
\begin{equation*}
2f(X^\star)=(\sqrt{s_1}-\sqrt{t_1})^2+ \sum_{2\le j<k\le n}(\sqrt{x_{jk}}-\sqrt{x_{kj}})^2=(\sqrt{s_1}-\sqrt{t_1})^2.
\end{equation*}
Therefore the submatrix $Y=[x_{jk}^\star]_{j,k\ge 2}$ is a nonnegative symmetric matrix.  Observe next that 
\begin{align*}
& s_j=x^\star_{j1}+\sum_{k=2}^n x^\star_{jk}=\frac{\sqrt{t_1}}{\sqrt{s_1}} c_j+\sum_{k=2}^n x^\star_{jk},\\
& t_j\,=x^\star_{1j}+\sum_{k=2}^n x^\star_{kj}=\frac{\sqrt{s_1}}{\sqrt{t_1}} c_j+\sum_{k=2}^n x^\star_{kj}, \; \textrm{ for }j\ge 2.
\end{align*}
As $Y$ is symmetric we obtain that 
\begin{eqnarray*}
s_j-t_j=\frac{(t_1-s_1)c_j}{\sqrt{s_1 t_1}}\ge 0 \quad \Rightarrow \quad  c_j=\frac{(s_j-t_j)\sqrt{s_1t_1}}{t_1-s_1}.
\end{eqnarray*}
As 
\begin{eqnarray*}
s_j\ge x^\star_{j1}=\frac{\sqrt{t_1}}{\sqrt{s_1}}c_j=\frac{(s_j-t_j)t_1}{t_1-s_1}
\end{eqnarray*}
we deduce that $t_1t_j\ge s_1 s_j$.  Hence  conditions \eqref{eqlowbdiagn1} hold.

Assume now that the conditions \eqref{eqlowbdiagn1} hold.  To be specific we assume that $t_1\ge s_1$ and $s_j\ge t_j$ for $j\ge 2$.  If $s_j=t_j$ for $j\ge 2$ then $\bs=\bt$ and equality holds in \eqref{eqlowbdiagn}.  Hence we assume that $t_1>s_1$.  
Define $X=[x_{ij}]$ as follows:
\begin{eqnarray*}
x_{11}=0, \, x_{1j}=\frac{s_1(s_j-t_j)}{t_1-s_1},\, x_{j1}=\frac{t_1(s_j-t_j)}{t_1-s_1},\, x_{jk}=\frac{t_1t_j-s_1s_j}{t_1-s_1}\delta_{jk} \; \textrm{ for } j,k\ge 2.
\end{eqnarray*}
Then $X\in \Gamma^{cl}(\bs,\bt)$.  Furthermore $2f(X)=s_1+t_1-2\sqrt{s_1t_1}=(\sqrt{s_1}-\sqrt{t_1})^2$.  Therefore  $2\rT_{C^{Q}}^Q(\bs,\bt)\le (\sqrt{s_1}-\sqrt{t_1})^2$. On the other hand, inequality \eqref{eqlowbdiagn} yields that $2\rT_{C^{Q}}^Q(\diag(\bs),\diag(\bt))\ge  (\sqrt{s_1}-\sqrt{t_1})^2$.  Consequently, we conclude that \mbox{$\rT_{C^{Q}}^Q(\diag(\bs),\diag(\bt))= \frac{1}{2}(\sqrt{s_1}-\sqrt{t_1})^2$.} 

Assume that $n=2$.  Then $s_1+s_2=t_1+t_2=1$.  Assume for simplicity of exposition that $t_1t_2\ge s_1s_2$.   Then $s_i=\min(s_1,s_2)\le \min(t_1,t_2)$.   Hence $s_j=1-s_i\ge t_j=1-t_i$ and the first condition of \eqref{eqlowbdiagn1} holds.  Hence,  \eqref{eqlowbdiagn2} is holds.
\end{proof}
\subsection{Proof of Theorem \ref{lowbdT}}\label{subsec:prfthm5.1}
(a)  Recall 
 inequality \eqref{diagdecr1} and the fact that $\rT_{C^{Q}}^Q$ is unitarily invariant, $\rT_{C^{Q}}^Q( \rho^A, \rho^B)=\rT_{C^{Q}}^Q(U^\dagger \rho^A U, U^\dagger\rho^B U)$, for $U\in U(n)$.  
Use the inequality \eqref{eqlowbdiagn} with $U=\mathbb{I}_n$  to deduce
\begin{align*}
\rT_{C^{Q}}^Q( \rho^A, \rho^B) &= \rT_{C^{Q}}^Q(U^\dagger \rho^A U, U^\dagger\rho^B U)\ge \rT_{C^{Q}}^Q(\diag(U^\dagger \rho^A U), \diag(U^\dagger\rho^B U))\\
& \ge\frac{1}{2}\max_{i\in[n]}\bigg(\sqrt{(U^\dagger\rho^A U)_{ii}}-\sqrt{(U^\dagger\rho^B U)_{ii}}\bigg)^2.
\end{align*}
Take the maximum over $U\in \rU(n)$ to deduce \eqref{lowbdT0}.

\noindent
(b)
First observe that $F$ that is given in \eqref{dualQOT}  is of the form:
\begin{equation}\label{sigmaFform}
\begin{aligned}
&  \sigma^A=-\begin{bmatrix}a&b\\\bar b&c\end{bmatrix}, \quad \sigma^B=-\begin{bmatrix}e&f\\\bar f&g\end{bmatrix}, \qquad a,c,e,g\in\R, \; b,f\in\C,\\
& F \;\, = \; \begin{bmatrix}a+e&f&b&0\\\bar f&a+g+1/2&-1/2&b\\\bar b&-1/2&c+e+1/2&f\\0&\bar b&\bar f&c+g
\end{bmatrix}. 
\end{aligned}
\end{equation}

We now assume that $\rho^A, \rho^B$ are positive definite and non-isospectral.  Proposition \ref{exrank1} yields that $\Gamma^Q(\rho^A,\rho^B)$ does not contain a matrix of rank one.    Let $\rho^{AB}$ and $F$ be the matrices for which  \eqref{orthcond} holds.  Our assumptions yield that $\rank \rho^{AB}\ge 2$.  Proposition \ref{dualQOT} yields that
$\tr F\rho^{AB}=0$.  Hence $\rank F\le 4-2=2$.  Note that
the second and the third columns of $F$ are nonzero.  Hence $\rank F\ge 1$.

For $U\in\rU(2)$ we have the equalities
\begin{align*}
&& \rT_{C^Q}^Q(\rho^A,\rho^B) & = \rT_{C^Q}^Q(U^\dagger\rho^A U,U^\dagger\rho^B U)=\tr \big(\sigma^A \rho^B+ \sigma^B\rho^B\big) \\
&&& = \tr \big((U^\dagger\sigma^A U)( U^\dagger\rho^B U)+(U^\dagger\sigma^B U)(U^\dagger\rho^A U)\big)
\end{align*}
\begin{align*}
\underline{ F} = (U^\dagger\otimes U^\dagger)F(U\otimes U^\dagger)=C^{Q}-(U^\dagger\sigma^A U)\otimes \mathbb{I}_2 -\mathbb{I}_2\otimes (U^\dagger\sigma^B U)\ge 0.
\end{align*}
We now choose $V\in\rU(2)$ so that $V^\dagger\sigma^A V$ is a diagonal matrix. Let
\begin{equation*}\label{sigmaFform1}
\begin{aligned}
& \underline{\rho}^A=V^\dagger\rho^A V, \quad \underline{\rho}^B=V^\dagger\rho^B V, \\
& \underline{\sigma}^A=V^\dagger\sigma^A V=-\begin{bmatrix} \underline{a}&0\\0& \underline{c}\end{bmatrix}, \quad \underline{\sigma}^B=V^\dagger\sigma^B V=-\begin{bmatrix} \underline{e}& \underline {f}\\ \bar{\underline{f}} &\underline{g}\end{bmatrix}, \quad \underline{a},\underline {c},\underline {e},\underline {g}\in\R, \underline {f}\in\C,\\
& \underline{F}=\begin{bmatrix}\underline{a}+\underline{e}&\underline{f}&0&0\\\bar {\underline{f}}&\underline{a}+\underline{g}+1/2&-1/2&0\\0&-1/2&\underline{c}+\underline{e}
+1/2&\underline{f}\\0&\bar 0&\bar {\underline{f}}&\underline{c}+\underline{g}
\end{bmatrix}.
\end{aligned}
\end{equation*}
Clearly $\rank \underline{F}=\rank F\le 2$.
We claim that $\rank \underline{F}=2$.  Assume to the contrary that $\rank \underline{F}=1$.  As the third column is nonzero we deduce that the fourth column is a multiple of the third column.  Hence the fourth column is zero.  That is, $\underline{f}=\underline{c}+\underline{g}=0$.  Similarly $\underline{a}+\underline{e}=0$.   Next observe that we can replace $\sigma^A,\sigma^B$ by $\sigma^A-\underline{a} \mathbb{I}_2, \sigma^B+\underline{a} \mathbb{I}_2$ without affecting the supremum in~\eqref{dualQOT1}.  This is equivalent to the assumption that $\underline{a}=0$.  Hence $\underline{e}=0$ and 
$\underline{g}=-\underline{c}$.
As $\underline{F}$ is Hermitian and $\rank\underline F = 1$ we have the condition
\begin{eqnarray*}
0=(-\underline{c}+1/2)(\underline{c}+1/2)-1/4=-\underline{c}^2.
\end{eqnarray*}
Hence $\underline{c}=-\underline{g}=0$.  Thus 
we can assume that $\sigma^A=\sigma^B=0$.  Equality \eqref{maxsig12} yields that $\rT_{C^Q}^Q(\rho^A,\rho^B)=0$, which implies that $\rho^A=\rho^B$.  This contradicts our assumption that $\rho^A$ and $\rho^B$ are  not similar.  Hence $\rank \underline{F}=\rank F=2$.

We claim that either $\underline{x}=\underline{a}+\underline{e}$ or $\underline{z}=\underline{c}+\underline{g}$ are zero.  Assume to the contrary that $\underline{x},\underline{z}>0$.  (Recall that $\underline{F}> 0$.)  Let $\bc_1,\bc_2,\bc_3,\bc_4$ be the four columns of $\underline{F}$.  Clearly $\bc_1,\bc_4$ are linearly independent.  Hence $\bc_2=u\bc_1+v\bc_4$.  As the fourth coordinate of $\bc_2$ is zero we deduce that $v=0$.  Hence $\bc_2=u\bc_1$. This is impossible since  the third coordinate of $\bc_1$ is $0$ and the third coordinate of $\bc_2$ is $-1/2$.    Hence  either $\underline{x}=\underline{a}+\underline{e}$ or $\underline{z}=\underline{c}+\underline{g}$ are zero.
Suppose that $\underline{x}=0$.  As $\underline{F}$ is positive semidefinite  we deduce that the first row and column of  $\underline{F}$ is zero.  Hence $f=0$.
Similarly, if $\underline{z}=0$ we deduce that $f=0$.  Thus $\underline{\sigma}^A$ and $\underline{\sigma}^B$ are diagonal matrices.  
Therefore
\begin{eqnarray*}
\rT_{C^Q}^Q(\underline{\rho}^A,\underline{\rho}^B)=\tr\big(\underline{\sigma}^A\underline{\rho}^A+
\underline{\sigma}^B\underline{\rho}^B\big)=\tr\big(\underline{\sigma}^A
\diag({\underline{\rho}}^A)+
\underline{\sigma}^B\diag({\underline{\rho}}^B)\big).
\end{eqnarray*}
As $\underline{F}\ge 0$, the maximum dual characterization yields 
\begin{align*}
\tr\big(\underline{\sigma}^A
\diag({\underline{\rho}}^A)+
\underline{\sigma}^B\diag({\underline{\rho}}^B)\big) \le \rT_{C^{Q}}^Q(\diag({\underline{\rho}}^A),\diag({\underline{\rho}}^B)).
\end{align*}
Hence $\rT_{C^Q}^Q(\underline{\rho}^A,\underline{\rho}^B)\le \rT_{C^Q}^Q(\diag({\underline{\rho}}^A),\diag({\underline{\rho}}^B))$.  Compare that with \eqref{diagdecr1} to deduce the equalities  
\begin{equation*}\label{4QOTeq}
\rT_{C^Q}^Q(\rho^A,\rho^B)=\rT_{C^Q}^Q(\underline{\rho}^A,\underline{\rho}^B)= \rT_{C^Q}^Q(\diag({\underline{\rho}}^A),\diag({\underline{\rho}}^B)).
\end{equation*}
Use the last part of Theorem \ref{lowbdTdiagn} to deduce
\begin{align*}
\rT_{C^Q}^Q( \underline{\rho}^A,\underline{\rho}^B) &= \rT_{C^Q}^Q \big(\diag(\underline{\rho})^A,\diag({\underline{\rho}}^B) \big) \\
& =\frac{1}{2}\max\Big[ \big(\sqrt{\underline{\rho}^A_{11}}-
\sqrt{\underline{\rho}^B_{11}}\, \big)^2, \, \big(\sqrt{\underline{\rho}^A_{22}}-
\sqrt{\underline{\rho}^B_{22}}\, \big)^2\Big].
\end{align*}
The inequality \eqref{lowbdT0} yields \eqref{lowbdT0a}
for $\rho^A$ and $\rho^B$ positive definite and non-isospectral.  Clearly every pair $\rho^A,\rho^B\in\Omega_2$ can be approximated by $\hat\rho^A,\hat\rho^B\in\Omega_2$ which are positive definite and non-isospectral.  Use the continuity of $T_{C^Q}^Q(\rho^A,\rho^B)$ on $\Omega_2\times\Omega_2$ (Proposition \ref{honconv})  to deduce \eqref{lowbdT0a} in the general case.  
\section{The induced quantum Wasserstein-2 distance}\label{sec:metrics}
 The main theorem of this section is:

 \begin{theorem}\label{Wa2metthm}  Assume that $C\in\rS(\cH_n\otimes\cH_n)$ is positive semidefinite 
 and vanishes exactly on $\cH_S$, 
 the symmetric subspace of $\cH_n\otimes\cH_n$. 
  Then $\sqrt{\rT^Q_{C}}$ is a weak distance, and $\windd$ defined by \eqref{defQOTmet} is the maximum distance majorized by $\sqrt{\rT^Q_{C}}$. 
  For $n=2$ the equality $\windQ(\rho^A,\rho^B)=\sqrt{\rT^Q_{C^Q}(\rho^A,\rho^B)}$ holds.
 \end{theorem} 
 
Let $X$ be a set of points.
Assume that $D:X\times X\to\R_+(=[0,\infty))$.  Then $D(\cdot,\cdot)$ is called a distance on $X$ if it satisfies the following three properties:
 \begin{enumerate}[(a)]
 \item Symmetry:  $D(x,y)=D(y,x)$;
 \item Positivity: $D(x,y)\ge 0$, and equality holds if and only if $x=y$.
 \item Triangle inequality: $D(x,y)+D(y,z)\ge D(x,z)$.
 \end{enumerate}
 
 The function
  $D(\cdot,\cdot)$ is called a semidistance if it satisfies the above first two conditions.  A semidistance is called a weak distance if there exists a distance $D'(\cdot,\cdot)$ such that 
\begin{equation}\label{DmajD'}
D'(x,y)\le D(x,y) \textrm{ for all }x,y\in X.
\end{equation}
\begin{proposition}\label{inddist} 
Assume that $D$ is a weak distance on the space $X$ satisfying \eqref{DmajD'}, where $D'$ is a distance on $X$.  For each positive integer $N$ define the following function:
\begin{equation*}
D_N(x,y)=\inf_{\substack{z_1,\ldots,z_N\in X,\\z_0=x, \; z_{N+1}=y}} \; \sum_{i=0}^N D(z_i,z_{i+1}) \textrm{ for } x,y\in X.
\end{equation*}
Then
\begin{enumerate}[(a)]
\item For each $N$ the function $D_N(\cdot,\cdot)$ is a weak distance 
that satisfies the inequality \eqref{DmajD'}.
\item For each $x,y\in X$ and $N$ we have the inequalities $0\le D_{N+1}(x,y)\le D_N(x,y)\le D(x,y)$.
\item For each $M,N\ge 1$ we have the inequality
\begin{equation*}
D_M(x,u)+D_N(u,y)\ge D_{M+N+1}(x,y) \textrm{ for } x,y,u\in X.
\end{equation*}
\item Denote by $D_\infty(x,y)=\lim_{N\to\infty}D_N(x,y)$.  Then $D_\infty(x,y)$ is a distance, called the induced distance of $D$.
Furthermore, $D_\infty$ is the maximum distance $D'$ that satisfies \eqref{DmajD'}.
\end{enumerate}
\end{proposition}
\begin{proof}
(a) Clearly $D_N(x,y)\ge 0$.  As $D(x,y)=D(y,x)$ it follows that 
\begin{equation*}
\ D(z_0,z_1)+\cdots+D(z_N,z_{N+1})=D(z_{N+1},z_N)+\cdots+D(z_1,z_0).
\end{equation*}
Hence $D_N(x,y)=D_N(y,x)$.  Assume that $y=x$.  Choose $z_1=\cdots=z_N=x$.
As $D(x,x)=0$ we deduce that $\sum_{i=0}^N D(z_i,z_{i+1})=0$.  Hence $D_N(x,x)$=0.  
As $D'$ is a distance we deduce
\begin{equation*}
\sum_{i=0}^N D'(z_i,z_{i+1})\ge D'(z_0, z_{N+1})=D'(x,y).
\end{equation*}
Use \eqref{DmajD'} to deduce that
\begin{eqnarray*}
\sum_{i=0}^N D(z_i,z_{i+1})\ge \sum_{i=0}^N D'(z_i,z_{i+1})\ge D'(x,y).
\end{eqnarray*}
Hence $D_N$ satisfies the inequality \eqref{DmajD'}.  In particular, if $x\ne y$ then
$D_N(x,y)\ge D'(x,y)>0$.  Therefore $D_N$ is a weak distance.

\noindent
(b)  Assume that $z_1=\ldots = z_{N}=x, z_{N+1}=y$.  Then 
$\sum_{i=0}^N D(z_i,z_{i+1})=D(x,y)$.  Hence $D_N(x,y)\le D(x,y)$. Now let $z_{N+1}=z_{N+2}=y$.
Then
\begin{equation*}
\sum_{i=0}^N D(z_i,z_{i+1})=\sum_{i=0}^{N+1} D(z_i,z_{i+1}).
\end{equation*}
Hence $D_{N+1}(x,y)\le D_N(x,y)$.

\noindent
(c)  Choose $z_0 = x$, $z_{M + 1} = u$, $z_{M + N + 2} = y$, and $z_1,\ldots,z_{M+N+1}$ arbitrarily.  Then $\sum_{i=0}^{M+N+1} D(z_i,z_{i+1})\ge D_{M+N+1}(x,y)$.
Compare that with the definitions of $D_{M}(x,u)$ and $D_{N}(u,y)$ to deduce the inequality $D_M(x,u)+D_N(u,y)\ge D_{M+N+1}(x,y)$.

\noindent
(d) As $\{D_N(x,y)\}$ is a non-increasing sequence such that $D_N(x,y)\ge D'(x,y)$ we deduce that the limit $D_\infty(x,y)$ exists and $D(x,y)\ge D_\infty(x,y)\ge D'(x,y)$. Since $D_N(x,y)=D_N(y,x)$ it follows that $D_\infty(x,y)=D_\infty(y,x)$. Hence $D_\infty(x,y)\ge 0$ and equality holds if and only if $x=y$.  In the inequality $D_M(x,u)+D_N(u,x)\ge D_{M+N+1}(x,y)$ let $M=N\to\infty$ to deduce that $D_\infty$ satisfies the triangle inequality.  Hence $D_\infty$ is a distance.  The inequality  $D(x,y)\ge D_\infty(x,y)\ge D'(x,y)$ yields that $D_\infty$ is a maximum distance $D'$ that satisfies \eqref{DmajD'}.
\end{proof}
\begin{theorem}\label{kapTABprop}  Let $C\in\rS(\cH_n\otimes\cH_n)$.  Then $\rT^Q_C$ is a semidistance on $\Omega_n\times \Omega_n$ if and only if $C$ is positive semidefinite and $\ker(C)=\cH_S$.
 Furthermore,
for $\rho^A,\rho^B\in\Omega_n$ the following statements hold:
\begin{enumerate}[(a)]
\item $\rT_{C}^Q(\rho^A,\rho^B)=\rT_{C}^Q(\rho^B,\rho^A)$.
\item $\rT_{C}^Q(\rho^A,\rho^B)\ge 0$.
\item $\rT_{C}^Q(\rho^A,\rho^B)=0$ if and only if $\rho^A=\rho^B$.
\item$\rT_{C^Q}^Q(\rho^A,\rho^B)\le \frac{1}{2}(1-\tr \rho^A\rho^B)$.  
Furthermore
\begin{equation}\label{QOTrankone}
\rT_{C^Q}^Q(\rho^A,\rho^B)= \frac{1}{2}(1-\tr \rho^A\rho^B) \textrm{ if either }
\rho^A \textrm{ or } \rho^B \textrm{ is a pure state}.
\end{equation} 
\item $\sqrt{\rT_{C^Q}^Q(\rho^A,\rho^B)}$ is a distance on pure states.
\end{enumerate}
\end{theorem}
\begin{proof}  We first show the second part of the theorem.   Assume that $C$ is positive semidefinite and  vanishes exactly on symmetric matrices.

\noindent
(a)   As $S$ is an involution with the eigenspaces $\rS^2\C^n$ and $\rA^2\C^n$ corresponding to the eigenvalues $1$ and $-1$ respectively, and $C\rS^2\C^n=0$, it follows that $SC=CS=-C$.  Hence $SC S^\dagger=C$.   
The second equality in \eqref{Werid} yields that $S\Gamma^Q(\rho^A,\rho^B)S^\dagger=\Gamma^Q(\rho^B,\rho^A)$.  As $\tr C \rho^{AB}=\tr C S\rho^{AB}S^\dagger$ we deduce (a).

\noindent
(b)  Since $C\ge 0$, for any $\rho^{AB}\in \Omega_{n^2}$ we get that $\tr C\rho^{AB}\ge 0$.  This proves (b).

\noindent
(c) Suppose that $\rho^A=\rho^B=\rho$.   Consider the spectral  decomposition of $\rho$ given by \eqref{specdecrho}.
Then a purification of $\rho$ is
\begin{equation}\label{purifA}
R= \Big( \sum_{i=1}^n \sqrt{\lambda_i} |\x_i\rangle|\x_i\rangle \Big) \Big( \sum_{j=1}^n \sqrt{\lambda_i} \langle\x_j|\langle\x_j| \Big)\in\Omega_{n^2}.\end{equation}
Clearly $R\in \Gamma^Q(\rho,\rho)$.  As $X=\sum_{i=1}^n \sqrt{\lambda_i} |\x_i\rangle|\x_i\rangle$ is a symmetric matrix it follows that $C X=0$.  Hence $\tr CR=0$ and  $\rT_{C}^Q(\rho,\rho)=0$.

Assume now that $\rT_{C}(\rho^A,\rho^B)=0$.  Hence $\tr C\rho^{AB}=0$ for some $\rho^{AB}\in\Gamma^Q(\rho^A,\rho^B)$.  That is, the eigenvectors of $\rho^{AB}$  are symmetric matrices.  Therefore $\rho^{AB}=\sum_{j=1}^k p_j|\psi_j\rangle \langle \psi_j|$ 
where each $ |\psi_j\rangle $ is a symmetric matrix and $p_j>0$.   We claim that each $|\psi_j\rangle \langle \psi_j|$ is of the form \eqref{purifA}.  This is equivalent to the Autonne--Takagi  factorization theorem \cite[Corollary 4.4.4, part (c)]{HJ13}
that any symmetric $X\in\C^{n\times n}$ is of the form 
\begin{eqnarray*}
X=\sum_{i=1}^n d_i|\x_i\rangle|\x_i\rangle=UDU^\top, \quad D=\diag(\bd), \quad U\in \rU(n),
\end{eqnarray*}
where the columns of $U$ represent vectors,  $\x_1,\ldots,\x_n$.
Clearly $\tr_A |\psi_j\rangle \langle \psi_j|=\tr_B |\psi_j\rangle \langle \psi_j|$.
Hence $\rho^B=\tr_A \rho^{AB}=\tr_B \rho^{AB}=\rho^A$.  

\noindent
(d)  As $\rho^A\otimes \rho^B\in \Gamma^Q(\rho^A,\rho^B)$ it follows that $\rT_{C^Q}^Q(\rho^A,\rho^B)\le \tr C^{Q} (\rho^A\otimes\rho^B)$.  Clearly $\tr \I(\rho^A\otimes\rho^B)=1$.  The first part of \eqref{Werid} yields that $\tr S(\rho^A\otimes\rho^B)=\tr(\rho^A\rho^B)$.
Hence 
$\tr C^{Q}(\rho^A\otimes\rho^B)=\frac{1}{2}\big(1-\tr \rho^A\rho^B\big)$, and
$\rT_{C^{Q}}^Q(\rho^A,\rho^B)\le \frac{1}{2}\big(1-\tr \rho^A\rho^B\big)$. Assume that either $\rho^A$ or $\rho^B$ is a pure state.  Lemma \ref{rangecont}  yields that $\Gamma^Q(\rho^A \rho^B)=\{\rho^A\otimes\rho^B\}$.  Hence \eqref{QOTrankone} holds.

\noindent
(e) It is known that if $\rho^A,\rho^B$ are pure state then \cite{Ren}
\begin{equation}\label{psiden}
\begin{aligned}
& \sqrt{1-\tr \rho^A\rho^B}=\frac{1}{2}\|\rho^A-\rho^B\|_1, \\ 
& \rho^A=|\x\rangle\langle \x|,\; \rho^B=|\y\rangle\langle \y|,\quad \langle \x|\x\rangle=\langle \y|\y\rangle=1. 
\end{aligned}
\end{equation}
Note that  if one of the states is pure
then
$\sqrt{1-\tr \rho^A\rho^B}$ reduces to  the root infidelity \cite{GLN05,MPHUZ} ---  see also Eq. \eqref{Ifidelity}. 
We give a short proof for completeness.
By changing the orthonormal basis in $\cH_n$ we can assume that $n=2$ and 
\begin{equation*}
\rho^A=\begin{bmatrix}1&0\\0&0\end{bmatrix}, \quad \rho^B=\begin{bmatrix}b&c\\c&1-b\end{bmatrix},\quad 0\le b\le 1,\,0\le  c, \,c^2=b(1-b).
\end{equation*}
As $\tr(\rho^A-\rho^B)=0$ it follows that the two eigenvalues of $\rho^A-\rho^B$
are 
\begin{equation*}
\pm \sqrt{-\det(\rho^A-\rho^B)}=\pm \sqrt{(1-b)^2 +c^2}=\pm\sqrt{1-b}=\pm\sqrt{1-\tr\rho^A\rho^B}.
\end{equation*}
This proves \eqref{psiden}.  Hence 
$\frac{1}{2}\|\rho^A-\rho^B\|_1+\frac{1}{2}\|\rho^B-\rho^C\|_1\ge\frac{1}{2}\|\rho^A-\rho^C\|_1$.  Combine that with (d) to deduce (e).

We now show the first part of the theorem.  Suppose that $C$ is positive semidefinite and  vanishes exactly on symmetric matrices. Then parts (a)-(c) of the theorem
show that  $\rT_{C}^Q$ is a semidistance.  

Assume now that $C\in\rS(\cH_n\otimes\cH_n)$ and $\rT_C^Q$ is a semidistance.
For $n=1$ it is straightforward to see that $C=0$.  Assume that $n>1$.
As $\rT_C^Q(\rho^A,\rho^B)>0$ for $\rho^A\ne \rho^B\in\Omega_n$ it follows that $C\ne 0$.   Let $R\in\rS(\cH_n\otimes \cH_n)$ be nonzero and positive semidefinite. We claim that $\tr CR\ge 0$.  It is enough to assume that $\tr R=1$.  Set $\rho^A=\tr_B R, \rho^B=\tr_A R$.  Then $R\in \Gamma^Q(\rho^A,\rho^B)$.  Thus $0\le \rT_C^Q(\rho^A,\rho^B)\le \tr CR$.  Suppose that $C=\sum_{k=1}^{n^2}\mu_k|\psi_k\rangle\langle \psi_k|$, where $|\psi_1\rangle,\ldots,|\psi_{n^2}\rangle$ is an orthonormal basis for $\cH_n\otimes\cH_n$.
Choose rank-one $R_k=|\psi_k\rangle\langle \psi_k|\ge 0$.  Thus $\mu_k=\tr CR_k\ge 0$ for $k\in[n^2]$.  Hence $C\ge 0$.
Let $\rho=|\x\rangle \langle\x|$ be a pure state.  Lemma \ref{rangecont} yields that $\Gamma^Q(\rho,\rho)=\{\rho\otimes\rho\}$.  Hence 
$0=\rT_C^Q(\rho,\rho)=\tr C(\rho\otimes\rho)$.
Noting that $\rho\otimes \rho=(|\x\rangle|\x\rangle)(\langle \x|\langle \x|)$,  as $C$ is positive semidefinite we deduce that $C(|\x\rangle|\x\rangle)=0$.  So $C$ vanishes on all rank one symmetric matrices, hence $C\cH_S=0$.    

It is left to show that $C|Y\rangle \ne 0$ if $Y$ is a nonzero skew-symmetric matrix.  
Assume to the contrary that $C|Y\rangle =0$ for some nonzero skew-symmetric matrix $Y$.
Let $Z\in \rS^2\C^n$ be the unique symmetric matrix with zero diagonal  such that $X=Z+Y$ is a nonzero lower triangular matrix with zero diagonal.  Note that $C|X\rangle =0$.  Normalize $X$ such that $\tr X X^\dagger=1$.
Let $R=|X\rangle\langle X|$, $\rho^A=\tr_B R, \rho^B =\tr_A R\in\Omega_n$.
Clearly $\tr CR=0$.  Hence $0\le \rT_C^Q(\rho^A,\rho^B)\le \tr CR=0$.  As $\rT_C^Q$ is a semidistance we deduce that $\rho^A=\rho^B$.  We now contradict this equality.
Indeed, consider the equality \eqref{Xrhosigrel}.  As $X$ is lower triangular with zero diagonal its first row is zero.  Hence $\rho^A_{11}=0$.  Hence $\rho^B_{11}=0$.  Note that $\rho_{11}^B$ is the norm squared of the first column of $X$.  Hence the first column of $X$ is zero.  Therefore the second row of $X$ is zero.
Thus $\rho^A_{22}=0$, which yields that $\rho^B_{22}=0$.  Therefore the second column of $X$ is zero.  Repeat this argument to deduce that $X=0$,
which contradicts our assumption that $\tr X X^\dagger=1$. 
\end{proof}

We now give a very general distance on positive semidefinite matrices, inspired by our lower bound \eqref{lowbdT0} on  $\rT_{C^{Q}}^Q(\rho^A,\rho^B)$ ,  which is exact on qubit density matrices.
\begin{proposition}\label{metricdenmat}  Let $\nu:\R^n \to [0,\infty)$ be a norm.  Assume that $f:[0,\infty)\to [0,\infty)$ is a continuous,  strictly increasing function.  For $\rho^A,\rho^B$ positive semidefinite define
\begin{equation}\label{defDrhosigma}
\begin{aligned}
D(\rho^A,\rho^B)=
\max_{U\in\rU(n)}
\nu\Big( & \big( f((U^\dagger\rho^A U)_{11}),\ldots,f((U^\dagger\rho^A U)_{nn}) \big)^\top
\\
& - \big(f((U^\dagger\rho^B U)_{11}),\ldots,f((U^\dagger\rho^B U)_{nn}) \big)^\top\Big).
\end{aligned}
\end{equation}
Then $D(\rho^A,\rho^B)$ is a distance on positive semidefinite matrices.  In particular,
\begin{equation}\label{defD0rhosigma}
\begin{aligned}
D_0(\rho^A,\rho^B) & =\max_{U\in\rU(n),i\in[n]} \big|f((U^\dagger\rho^A U)_{ii})-f((U^\dagger\rho^B U)_{ii}) \big|\\
& =\max_{U\in\rU(n)} \big|f((U^\dagger\rho^A U)_{11})-f((U^\dagger\rho^B U)_{11}) \big|
\end{aligned}
\end{equation}
is a distance on positive semidefinite matrices.
\end{proposition}

Let us note that the distance \eqref{defD0rhosigma} admits an operational interpretation: It quantifies the maximal distinguishability between two given state, $\rho^A, \rho^B \in \Omega_n$, achievable through local unitary rotations followed by a projective measurement. 

\begin{proof}  By definition $D(\rho^A,\rho^B)=D(\rho^B,\rho^A)\ge 0$.
Assume that $D(\rho^A,\rho^B)=0$.  Then $f\bigl((U^\dagger\rho^A U)_{ii}\bigr)=f\bigl((U^\dagger\rho^B U\bigr)_{ii})$ for each $i\in[n]$ and $U\in U(n)$.  As $f$ is strictly increasing we deduce that $(U^\dagger\rho^A U)_{ii}=(U^\dagger\rho^B U)_{ii}$ for $i\in[n]$.  That is for each $U\in \rU(n)$ the diagonal entries of $U^\dagger(\rho^A-\rho^B) U$ are 0.  Choose a unitary $V$ so that $V^\dagger(\rho^A-\rho^B)V$ is diagonal.  Then $V^\dagger(\rho^A-\rho^B)V=0$.  Hence $\rho^A=\rho^B$.  It is left to show the triangle inequality.  

Denote by $\bbf(\rho)$ the vector $\bigl(f(\rho_{11}),\ldots,f(\rho_{nn})\bigr)^\top$.
Since $f$ is continuous there exists $V\in U(n)$ such that 
$D(\rho^A,\rho^B)=\nu\big(\bbf(V^\dagger\rho^A V)-\bbf(V^\dagger\rho^B V)\big)$.  Hence
\begin{align*}
D(\rho^A,\rho^B) & = \nu\big(\bbf(V^\dagger\rho^A V)-\bbf(V^\dagger\rho^B V)\big)\\
 & \le \nu\big(\bbf(V^\dagger\rho^A V)-\bbf(V^\dagger\rho^C V)\big)+\nu\big(\bbf(V^\dagger\rho^C V)-\bbf(V^\dagger\rho^B V)\big) \\
 & \le D(\rho^A,\rho^C)+D(\rho^C,\rho^B).
\end{align*}

To show that $D_0(\cdot,\cdot)$ is a distance we observe that $D_0(\rho^A,\rho^B)=D(\rho^A,\rho^B)$ where $\nu\big((x_1,\ldots,x_n)^\top\big)=\max_{i\in[n]}|x_i|$. To show the last equality of \eqref{defD0rhosigma} let $\rP_n\subset \rU(n)$ denote the group of permutation matrices.  Then
\begin{align*}
& \max_{i\in[n]} \big|f((U^\dagger\rho^A U)_{ii}) - f((U^\dagger\rho^B U)_{ii})\big| \\
& \qquad\qquad\qquad = \max_{P\in \rP_n} \big| f \big(((UP)^\dagger\rho^A (UP))_{11} \big)-f \big(((UP)^\dagger\rho^B (UP))_{11} \big)\big|. \qquad \qedhere
\end{align*}
\end{proof}
{\textsc{Proof of Theorem} \ref{Wa2metthm}.}
 As $C$ is positive semidefinite that vanishes exactly on symmetric matrices,
Theorem \ref{kapTABprop} yields that $\rT_C^Q$ is a semidistance on $\Omega_n$.
Next observe that $C\ge aC^Q$ for some $a>0$.
  Hence $\rT_C^Q(\rho^A,\rho^B)\ge a\rT_{C^Q}^Q(\rho^A,\rho^B)$.   
  Let $D_0(\rho^A,\rho^B)$ be 
the distance defined  in \eqref{defD0rhosigma}, with 
 $f(x)=\sqrt{x/2}$ for $x\ge 0$.  
  The inequality
  \eqref{lowbdT0} yields that $\sqrt{\rT_C^Q(\rho^A,\rho^B)}\ge \sqrt{a}D_0(\rho^A,\rho^B)$.  Hence $\sqrt{\rT_C^Q(\rho^A,\rho^B)}$ is a weak distance. Proposition \ref{inddist} shows that  $\sqrt{\rT_C^Q}$ yields the induced Wasserstein-2 distance given by \eqref{defQOTmet}, which is the maximum distance majorized by $\sqrt{\rT_C^Q}$.
  
Assume that $n=2$.  Then \eqref{lowbdT0a}  yields that $\sqrt{\rT^Q_{C^Q}(\rho^A,\rho^B)}= D_0(\rho^A,\rho^B)$. Hence, $\windQ(\rho^A,\rho^B)=\sqrt{\rT^Q_{C^Q}(\rho^A,\rho^B)}$.\qed

Observe that if $\sqrt{\rT_C^Q}$ is a weak distance, then so is $\big(\rT_C^Q \big)^{1/p}$ for any $p\geq 2$, which dominates the distance  $D_0^{2/p}$. Then, one can define $\mathcal{W}^Q_{C,p}$ in a similar way as in formula \eqref{defQOTmet}.

\section{Quantum optimal transport for $d$-partite systems}\label{sec:QOTdpar}
We now explain briefly how to state the quantum optimal transport problem for a $d$-partite system, where $d\ge 3$, similarly to what was done in \cite{FrV18,Fri20}.  
The main result of this section is Theorem \ref{kapTd}.

Let $\cH_{n_j}$ be a Hilbert space of dimension $n_j$ for $j\in[d]$.
We consider the $d$-partite tensor product space $\otimes_{j=1}^d \cH_{n_j}$.  
A product state in Dirac's notation is $\otimes_{i=1}^d |\x_i\rangle$.  Then
\begin{eqnarray*}
\langle \otimes_{i=1}^d \x_i,\otimes_{j=1}^d \y_j\rangle=(\otimes_{i=1}^d \langle \x_i|)
(\otimes_{j=1}^d |\y_j\rangle)=
\prod_{j=1}^d \langle \x_j| \y_j
\rangle.
\end{eqnarray*}
Consider the space $\rB(\otimes_{j=1}^d\cH_{n_j})$ of linear operators from $\otimes_{j=1}^d\cH_{n_j}$ to itself.  A rank-one product operator is of the form
$(\otimes_{i=1}^d |\x_i\rangle)(\otimes_{j=1}^d\langle\y_j|)$ and 
acts on a product state as follows:
\begin{eqnarray*}
(\otimes_{i=1}^d |\x_i\rangle)(\otimes_{j=1}^d\langle\y_j|)(\otimes_{k=1}^d |\z_k\rangle)=(\prod_{j=1}^d\langle \y_j| \z_j\rangle) (\otimes_{i=1}^d|\x_i\rangle).
\end{eqnarray*}
Given $\rho^{A_1,\ldots ,A_d}\in \rB(\otimes_{j=1}^d\cH_{n_j})$ one can define a $k$-partial trace on $k\in[d]$:
\begin{align*}
&\tr_k: \rB \big(\otimes_{j=1}^d\cH_{n_j} \big)\to \rB \big(\otimes_{j\in[d]\setminus\{k\}}\cH_{n_j} \big),\\
&\tr_k (\otimes_{i=1}^d |\x_i\rangle)(\otimes_{j=1}^d\langle\y_j|)=\langle \y_k|\x_k\rangle 
(\otimes_{i\in[d]\setminus\{k\}} |\x_i\rangle)(\otimes_{j\in[d]\setminus\{k\}}^d\langle \y_j|).
\end{align*}
We will denote $\tr_k\rho^{A_1,\ldots,A_d}$ by $\rho^{A_1,\ldots,A_{k-1},A_{k+1},\ldots,A_d}$.  Let $\rho^{A_k}\in\rB(\cH_{n_k})$ be the operator obtained from $\rho^{A_1,\ldots,A_d}$ by tracing out all but the $k$-th component.
Thus we have the map 
\begin{align*}
& \widetilde{\mathrm{Tr}}:  \rB \big(\otimes_{j=1}^d\cH_{n_j} \big)\to \oplus_{j=1}^d \rB \big(\cH_{n_j} \big),\\
& \widetilde{\mathrm{Tr}}(\rho^{A_1,\ldots,A_d})=(\rho^{A_1},\ldots,\rho^{A_d}).
\end{align*}
Let $N=\prod_{j=1}^d n_j$ and view the set of density matrices $\Omega_N$ as a subset of selfadjoint operators on $\cH_N=\otimes_{j=1}^d \cH_{n_j}$.
For $\rho^{A_i}\in\Omega_{n_i}, i\in[d]$ denote
\begin{eqnarray*}
\Gamma^{Q}(\rho^{A_1},\ldots,\rho^{A_d})=\{\rho^{A_1,\ldots,A_d}\in \Omega_N, 
\widetilde{\mathrm{Tr}} (\rho^{A_1,\ldots,A_d})=(\rho^{A_1},\ldots,\rho^{A_d})\}.
\end{eqnarray*}

Assume that  $C$ is a selfadjoint operator on $\cH_N$.  We define the quantum optimal transport as
\begin{equation}\label{defdQT}
\rT_{C}^Q(\rho^{A_1},\ldots,\rho^{A_d})=\min_{\rho^{A_1,\ldots,A_d}\in\Gamma^Q(\rho^{A_1},\ldots,\rho^{A_d})}\tr C\rho^{A_1,\ldots,A_d}.
\end{equation}

We now give an analog of
a result in \cite{Fri20}.  Assume that $d=2\ell\ge 4$, and
$n_1=\cdots=n_d=n$.  Then $\cH_n^{\otimes d}=\otimes^d\cH_n$.
We want to give a semidistance between two ordered $\ell$-tuples of density matrices $( \rho^{A_1},\cdots,\rho^{A_\ell}),(\rho^{A_{\ell+1}},\cdots,\rho^{A_{2\ell}})\in\Omega_n^\ell$.  We view $\cH_n^{\otimes (2\ell)}$ as bipartite states $\cH_n^{\otimes \ell}\otimes \cH_n^{\otimes \ell}$.  Let $S\in\rB(\cH_n^{\otimes (2\ell)})$ be the SWAP operator:
\begin{eqnarray*}
S(\otimes_{j=1}^{2\ell} |\x_j)\rangle=(\otimes_{j=1}^{\ell} |\x_{j+\ell}\rangle)\otimes (\otimes_{j=1}^\ell |\x_j\rangle).
\end{eqnarray*}
Denote by 
$C^{Q}=\frac{1}{2}(\I-S)$. 
Then $\rT_{C^{Q}}^Q(\rho^{A_1},\ldots,\rho^{A_{2\ell}})\ge 0$. Equality holds if and only if 
$(\rho^{A_1},\ldots,\rho^{A_\ell})=(\rho^{A_{1+\ell}},\ldots,\rho^{A_{2\ell}})$.  
Also 
\begin{equation*}
\rT_{C^{Q}}^Q(\rho^{A_1},\ldots,\rho^{A_{2\ell}})=\rT_{C^{Q}}^Q(\rho^{A_{1+\ell}},\ldots,\rho^{A_{2\ell}}, \rho^{A_1},\ldots,\rho^{A_{\ell}}).
\end{equation*}
Hence $\rT_{C^{Q}}^Q(\rho^{A_1},\ldots,\rho^{A_{2\ell}})$ is a semidistance on $\Omega_n^\ell$.  As in the case of $\ell=1$ we can show that $\sqrt{\rT_{C^{Q}}^Q(\rho^{A_1},\ldots,\rho^{A_{2\ell}})}$ is a weak distance.
Denote by 

\noindent
$W^Q_{C^Q}((\rho^{A_1},\ldots,\rho^{A_\ell}),(\rho^{A_{\ell+1}},\ldots,\rho^{A_{2\ell}}))$ the Wasserstein-2 distance on $\Omega_n^\ell$ induced by the weak distance $\sqrt{\rT_{C^{Q}}^Q(\rho^{A_1},\ldots,\rho^{A_{2\ell}})}$.

Let $\Sigma_\ell$ be the group of bijections $\pi:[\ell]\to[\ell]$.
Then  
\begin{eqnarray*}
\min_{\pi\in\Sigma_{\ell}}W_{C^{Q}}^Q \big( (\rho^{A_{\pi(1)}},\ldots,\rho^{A_{\pi(\ell)}}),  (\rho^{1+\ell},\ldots,\rho^{2\ell}) \big)
\end{eqnarray*}
gives a distance on unordered $\ell$-tuples of density matrices.
We call this distance the quantum Wasserstein-2 distance on the set of unordered $\ell$-tuples $\{\rho^{A_1},\ldots,\rho^{A_\ell}\}$.

On $\cH_n^{\otimes d}$ we define for two integers $1\le p<q\le d$ the SWAP operator $S_{pq}\in \rB(\cH_{n})^{\otimes d}$, which swaps $\x_p$ with $\x_q$ in the tensor product $|\x_1\rangle\otimes\cdots\otimes|\x_d\rangle$.  Note that $S_{pq}$ is unitary and involutive.  Hence $S_{pq}$ is selfadjoint with eigenvalues $\pm 1$.
The common invariant subspace of $\cH_n^{\otimes d}$ for all $S_{pq}$ is the
the subspace of symmetric tensors\
---``bosons''---,  denoted as $\rS^d\cH_n$. 
 Let $C^{B}\in \rS_+(\cH_n^{\otimes d})$ be the projection on the orthogonal complement of $\rS^d\cH_n$.  Note that  $C^{B}=C^Q$ for $d=2$.   We now have a partial analog of Theorem \ref{kapTABprop}:
\begin{theorem}\label{kapTd}
 Let $\rho^{A_1},\ldots,\rho^{A_d}\in\Omega_n$.  Then 
 \begin{enumerate}[(a)]
 \item
 $\rT_{C^{B}}^Q(\rho^{A_1},\ldots,\rho^{A_d})\ge 0$.  
 \item 
 $\rT_{C^{B}}^Q(\rho^{A_1},\ldots,\rho^{A_d})=0$ if and only if $\rho^{A_1}=\cdots=\rho^{A_d}$.
 \item
 Assume that at least $d-1$ out of $\rho^{A_1},\ldots,\rho^{A_d}$ are pure states.
 Then 
 \begin{equation*}
 \rT_{C^{B}}^Q(\rho^{A_1},\ldots,\rho^{A_d})=\tr C^{B}(\otimes_{j=1}^d\rho^{A_j}).
 \end{equation*}
 \end{enumerate}
 \end{theorem} 
 \begin{proof} (a) This follows from the fact that $\tr C^{B}\rho^{A_1,\ldots,A_d}\ge 0$.
 
 \noindent
 (b) Assume that $\rT_{C^{B}}^Q(\rho^{A_1},\ldots,\rho^{A_d})=\tr C^{B}\rho^{A_1,\ldots,A_d}=0$.  Hence all the eigenvectors of $\rho^{A_1,\ldots,A_d}$ corresponding to positive eigenvalues are symmetric tensors.  So $S_{pq} \rho^{A_1,\ldots,A_d} S_{pq}=\rho^{A_1,\ldots,A_d}$.  Therefore $\widetilde{\mathrm{Tr}}(\rho^{A_1,\ldots,A_d})=(\rho,\ldots,\rho)$.
 Thus $ \rho^{A_1}=\cdots=\rho^{A_d}=\rho$.  We now show that $\rT_{C^{B}}^Q(\rho,\ldots,\rho)=0$.
 Suppose that $\rho$ has the spectral decomposition \eqref{specdecrho}.  Let us take 
 a $d$-purification of $\rho$
 \begin{equation*}
 \rho^{\mathrm{pur},d}=\big(\sum_{i=1}^n \sqrt{\lambda_i}\otimes^d|\x_i\rangle \big) \big( \sum_{j=1}^n \sqrt{\lambda_j}\otimes^d\langle\x_j| \big).
 \end{equation*}
 Clearly we have $\rho^{\mathrm{pur},d}\in\Gamma^Q(\rho,\ldots,\rho)$.  As $\rho^{\mathrm{pur},d}$ is a pure state whose eigenvector corresponding to its positive eigenvalue is
 a symmetric tensor  we deduce that $\tr C^{B}\rho^{\mathrm{pur},d}=0$.
 
 \noindent (c) Assume for simplicity of the exposition that $\rho^{A_2},\ldots,\rho^{A_d}$ are pure states.  Then $\rho^B=\otimes_{j=2}^d \rho^{A_j}$ is a pure state.    Lemma \ref{rangecont} yields that $\Gamma^Q(\rho^{A_1},\rho^{B})=\{\rho^{A_1}\otimes \rho^{B}\}$.  Hence $\Gamma^{Q}(\rho^{A_1},\ldots,\rho^{A_d})=\{\otimes_{j=1}^d \rho^{A_j}\}$. This proves part (c) of the theorem.
 \end{proof}

The next question concerns the optimal technique
 to compute $\tr C^{B}(\otimes_{j=1}^d\rho^{A_j})$.
This problem is related to the permanent function. Assume first that each $\rho^{A_j}$ is a pure state $|\x_j\rangle\langle\x_j|$, where $\langle\x_j|\x_j\rangle=1$.  Then $\otimes_{j=1}^d \rho^{A_j}$ is a pure product state with the positive eigenvector $\otimes_{j=1}^d |\x_j\rangle$.  A symmetrization of  $\otimes_{j=1}^d |\x_j\rangle$ is the orthogonal projection on the subspace of symmetric tensors, given by
\begin{eqnarray*}
(\I-C^{B})(\otimes_{j=1}^d |\x_j\rangle)=\frac{1}{d!}\sum_{\pi\in\Sigma_d} \otimes_{j=1}^d |\x_{\pi(j)}\rangle.
\end{eqnarray*}
Hence
\begin{eqnarray*}
\big\|(\I-C^{B})(\otimes_{j=1}^d |\x_j\rangle)\big\|^2=\frac{1}{d!} \sum_{\pi\in\Pi_d}\prod_{j=1}^d \langle \x_j|\x_{\pi(j)}\rangle.
\end{eqnarray*}
Let $X=[\x_1\cdots \x_d]\in\C^{n\times d}$ be the matrix whose columns are the vectors $[\x_1,\ldots,\x_d]$.  The $G(\x_1,\ldots,\x_d)=X^\dagger X$ is the Gramian matrix $[\langle \x_i|\x_j\rangle]\in \rH_{d,+}$.  Note that since $\|\x_1\|=\cdots=\|\x_d\|=1$ the diagonal entries of $G(\x_1,\ldots,\x_d)$ are all 1, and $G(\x_1, \ldots, \x_d)$  is called a complex covariance matrix.
It now follows that $\|(\I-C^{B})\otimes_{j=1}^d|\x_j\rangle\|^2$ is $\frac{1}{d!}$ times the permanent of $G(\x_1,\ldots,\x_d)$, denoted as per$\, G(\x_1,\ldots,\x_d)$.
Hence
\begin{align*}
\tr C^{B}(\otimes_{i=1}^d|\x_i\rangle)(\otimes_{j=1}^d\langle\x_i|)=1-\frac{1}{d!}\textrm{per}\,G(\x_1,\ldots,\x_d),&& \|\x_1\|=\cdots=\|\x_d\|=1.
\end{align*}
\begin{lemma}\label{trTprodps}
Assume that $\rho^{A_1},\ldots,\rho^{A_d}\in\Omega_n$ have the following spectral decomposition:
\begin{eqnarray*}
\rho^{A_j}=\sum_{i=1}^n \lambda_{i,j}|\x_{i,j}\rangle\langle\x_{i,j}|, \quad j\in[d].
\end{eqnarray*}
Then
\begin{equation}\label{trTprodps1}
\tr C^{B}(\otimes_{j=1}^d\rho^{A_j})=1-\frac{1}{d!}
\sum_{i_1\ldots,i_d\in[n]}\prod_{j=1}^d\lambda_{i_j,j}\; \mathrm{per}\,G(\x_{i_1,1},\ldots,\x_{i_d,d}).
\end{equation}
\end{lemma}
The proof of this lemma follows straightforwardly from the multilinearity of $\otimes_{j=1}^d\rho^{A_j}$.

We now state the analog to part (d) of Theorem \ref{kapTABprop} , which is a corollary to the above lemma:
\begin{corollary}\label{QTupbdcase}
Let $\rho^{A_1},\ldots,\rho^{A_d}$ be density matrices with the spectral decomposition given by Lemma \ref{trTprodps}.  Then
\begin{eqnarray*}
\rT_{C^{B}}^Q(\rho^{A_1},\ldots,\rho^{A_d})\le 1- \frac{1}{d!}
\sum_{i_1\ldots,i_d\in[n]}\prod_{j=1}^d\lambda_{i_j,j}\; \mathrm{per}\,G(\x_{i_1,1},\ldots,\x_{i_d,d}).
\end{eqnarray*}
If at least $d-1$ density matrices are pure states then equality holds.
\end{corollary}

\emph{Acknowledgements:}
It is a pleasure to thank Rafa{\l}~Bistro{\'n}, 
John Calsamiglia, Matt Hoogsteder,  Tomasz Miller,
Wojciech~S{\l}omczy{\'n}ski and Andreas Winter
for numerous inspiring discussions and  helpful remarks.
Financial support by Simons collaboration grant for mathematicians, Narodowe Centrum Nauki 
under the Maestro grant number DEC-2015/18/A/ST2/00274 
and by the Foundation for Polish Science 
under the Team-Net project no. POIR.04.04.00-00-17C1/18-00
is gratefully acknowledged.

\appendix
\section{Basic properties of partial traces}\label{sec:partr}

In order to understand the partial traces on $ \rB(\cH_m\otimes \cH_n)$ it is convenient to view this space as a $4$-mode tensor space \cite{FGZ19} and use Dirac notation.  Denote by $\cH_m^\vee$ the space of linear operators on $\cH_m$, i.e., the dual space.  Then $\y^\vee=\langle \y|\in \cH_m^\vee$ acts on $\z\in\cH_m$ as follows: $\y^\vee (\z)=\langle \y,\z\rangle=\langle \y|\z\rangle$.  Hence a rank-one operator in $\rB(\cH_m)$ is of the form $\x\otimes \y^\vee=|\x\rangle\langle \y|$, where $(|\x\rangle \langle\y|)(\z)=\langle\y|\z\rangle |\x\rangle$.  
So $|\x\rangle\langle\y|$ can be viewed a matrix $\rho=\x\y^\dagger\in\C^{m\times m}$.
Assume that $V_1,V_2$ are linear transformations from $\cH_m$ to itself.  Then $V_1\otimes V_2$ is a linear transformation  
from $\cH_m\otimes \cH_m^\vee$ to itself, which acts on rank one operators as follows:
\begin{equation*}
(V_1\otimes V_2)( |\x\rangle\langle \y|)=
|V_1 \x\rangle\langle V_2\y|=V_1( |\x\rangle\langle \y|)V_2^\dagger, \quad \x,\y\in\cH_m.
\end{equation*}
Assume now that $W_1,W_2$ are linear transformations  from $\cH_n$ to itself. 
Then
\begin{equation*}
(V_1\otimes W_1)|\x\rangle|\bv\rangle= |V_1\x\rangle|W_1\bv\rangle, \quad \x\in\cH_m,\y\in\cH_n.
\end{equation*} 
A tensor product of two rank-one operators is identified  a 4-tensor:
\begin{equation}\label{tenprodiden}
|\x\rangle\langle \y|\otimes  |\bu\rangle\langle \bv|=|\x\rangle|  \bu\rangle\langle \y|\langle \bv|, \quad \x,\y\in\cH_m,\bu,\bv\in\cH_n.
\end{equation}
Thus
\begin{eqnarray*}
(|\x\rangle|  \bu\rangle\langle \y|\langle \bv|)(|\z\rangle|\bw\rangle)=
\langle \y|\z\rangle \langle \bv|\bw\rangle |\x\rangle| \bu\rangle, \quad \x,\y,\z\in\cH_m, \bu,\bv,\bw\in\cH_n.
\end{eqnarray*}
Observe next that  $V_1\otimes W_1\otimes V_2\otimes W_2$ is a linear transformation of $\rB(\cH_m\otimes\cH_n)$ to itself,  which acts on a rank-one product operator as follows:
\begin{align*}
(V_1\otimes W_1\otimes V_2\otimes W_2) (|\x\rangle|\bu\rangle\langle \y|\langle \bv|) &= |V_1 \x\rangle|W_1\bu\rangle\langle V_2\y|\langle W_2\bv|\\
& = (V_1\otimes W_1)(|\x\rangle|\bu\rangle\langle \y|\langle \bv|)(V_2^\dagger\otimes W_2^\dagger).
\end{align*}
(In the last equality we view $|\x\rangle|\bu\rangle\langle \y|\langle \bv|$ as an $(mn)\times (mn)$ matrix.)
As $\tr |\x\rangle\langle \y|=\langle \y|\x\rangle$ we deduce the following lemma:
\begin{lemma}\label{partrlem}  Let 
\begin{eqnarray*}
\x,\y\in\cH_m, \bu,\bv\in\cH_n, \quad V_1,V_2\in\rB(\cH_m), \; W_1,W_2\in\rB(\cH_n).
\end{eqnarray*}
Then
\begin{eqnarray*}
&&\tr_A |\x\rangle| \bu\rangle\langle \y|\langle \bv|=\langle \y|\x\rangle |\bu\rangle\langle \bv|,\\
&&\tr_B |\x\rangle| \bu\rangle\langle \y|\langle \bv|=\langle \bv|\bu\rangle |\x\rangle\langle \by|, \\
&&\tr_A (V_1\otimes W_1\otimes V_2\otimes W_2)(|\x\rangle| \bu\rangle\langle \y|\langle \bv|)=\langle V_2\y|V_1\x\rangle |W_1\bu\rangle\langle W_2\bv|,\\
&&\tr_B (V_1\otimes W_1\otimes V_2\otimes W_2)(|\x\rangle| \bu\rangle\langle \y|\langle \bv|)=\langle W_2\bv|W_1\bu\rangle |V_1\x\rangle\langle V_2\y|.
\end{eqnarray*}
In particular, if $V_1=V_2=V$ and $W_1=W_2=W$ are unitary then
\begin{align*}
& \tr_A (V\otimes W\otimes V\otimes W)(|\x\rangle| \bu\rangle\langle \y|\langle \bv|)
=\langle\y|\x\rangle |W\bu\rangle\langle W\bv|,\\
& \tr_B (V\otimes W\otimes V\otimes W)(|\x\rangle| \bu\rangle\langle \y|\langle \bv|)
=\langle\bv|\bu\rangle |V\x\rangle\langle V\y|.
\end{align*}
\end{lemma}\label{ABUVcor}
\begin{corollary}\label{corA2}  Let $\rho^A\in\Omega_m,\rho^B\in\Omega_n$, $V\in\rB(\cH_m),W\in\rB(\cH_n)$ be unitary and $C\in\rS(\cH_m\otimes \cH_n)$.
Then
\begin{align*}
& \Gamma^Q(V\rho^AV^\dagger, W\rho^BW^\dagger)=(V\otimes W)\Gamma^Q(\rho^A,\rho^B)(V^\dagger\otimes W^\dagger),\\
& \rT_{C}^Q(\rho^A,\rho^B)=\rT_{(V\otimes W)C(V^\dagger\otimes W^\dagger)}( V\rho^AV^\dagger,W\rho^B W^\dagger).
\end{align*}
\end{corollary}
\begin{proof} View $\rho^A\in \Omega_m$ as an element  in $\cH_m\otimes \cH_m^\vee$ to deduce $V\rho^AV^\dagger=(V\otimes V) \rho^A$.
Suppose that 
\begin{eqnarray*}
\rho^{AB}=\sum_{\substack{i,j\in[m]\\ p,q\in [n]}} r_{(i,p)(j,q)} |i\rangle |p\rangle\langle j| \langle q| \in\Gamma^Q( \rho^A, \rho^B).
\end{eqnarray*}
Let 
$\tilde\rho^{AB}=(V\otimes W\otimes V\otimes W) \rho^{AB}$.
Observe that
\begin{eqnarray*}
&&\tr_A \rho^{AB}=\sum_{p,q\in[n]}\Big( \sum_{i\in[m]} r_{(i,p)(i,q)} \Big)|p\rangle\langle q|= \rho^B,\\
&&\tr_A \tilde\rho^{AB}=\sum_{p,q\in[n]} \Big( \sum_{i\in[m]} r_{(i,p)(i,q)} \Big) \big( \langle q|W^\dagger \big)\big(W|p\rangle\big) =W\rho_{B}W^\dagger.
\end{eqnarray*}
Similarly $\tr_B \tilde\rho^{AB}=V \rho^AV^\dagger$.  Hence 
\begin{equation*}
(V\otimes W\otimes V\otimes W)\Gamma^Q( \rho^A, \rho^B)\subseteq \Gamma^Q(V \rho^AV^\dagger, W \rho^BW^\dagger).  
\end{equation*}
and 
\begin{equation*}
(V^\dagger\otimes W^\dagger\otimes V^\dagger\otimes W^\dagger)\Gamma^Q(V\rho^AV^\dagger,W\rho^BW^\dagger)\subseteq \Gamma^Q(\rho^A, \rho^B).  
\end{equation*}
Hence we deduce the first part of the corollary.  The second part of the corollary follows from the identity 
\begin{equation*}
\tr C\rho^{AB}=\tr (V\otimes  W) C(V^\dagger\otimes W^\dagger)(V\otimes W) \rho^{AB}(V^\dagger\otimes W^\dagger).\qedhere
\end{equation*}
\end{proof}
The following result appeared in the literature \cite{FGZ19} 
and we state it here for completeness.
For $\rho^A\in \rB(\cH_m)$ denote by range$\, \rho^A\subseteq \cH_m$ the range of $\rho^A$.
\begin{lemma}\label{rangecont}  Let $\rho^A\in\Omega_m,\rho^B\in \Omega_n$. Then
\begin{equation*}
\Gamma^Q(\rho^A,\rho^B)\subseteq \rB(\mathrm{range}\,\rho^A)\otimes  \rB(\mathrm{range}\,\rho^B).
\end{equation*}
In particular if either $\rho^A$ or  $\rho^B$ is a pure state then $\Gamma^Q( \rho^A, \rho^B)=\{ \rho^A\otimes  \rho^B\}$.
\end{lemma}
\begin{proof}  It is enough to show that $\Gamma^Q(\rho^A,\rho^B)\subset \rB(\textrm{range}\,\rho^A)\otimes  \rB(\cH_n)$. To show this condition we can assume that range$\,\rho^A$ is a nonzero strict subspace of $\cH_m$.  By choosing a corresponding orthonormal basis consisting of eigenvectors of $\rho^A$ we can assume that $\rho^A$ is a diagonal matrix whose first $1\le \ell <m$ diagonal entries are positive, and whose last $n-\ell$ diagonal entries are zero.  Write down  $\rho^{AB}$ as a block matrix $[R_{pq}] \in\C^{(mn)\times (mn)}$, were $R_{pq}\in\C^{m\times m}, p,q\in[n]$.  Then $\tr_B \rho^{AB}=\sum_{p=1}^n R_{pp}= \rho^A$.  As $R_{pp}\ge  0$ we deduce that $ \rho^A=[a_{ij}]\ge R_{pp}\ge 0$.  As $a_{ii}=0$ for $i>\ell$ it follows that the $(i,i)$ entry of each $R_{pp}$ is zero. As $\rho^{AB}$ positive semidefinite it follows that the $((p-1)n+i)$th row and column of $\rho^{AB}$ are zero.  This proves $\Gamma^Q(\rho^A,\rho^B)\subseteq \rB(\textrm{range}\,\rho^A)\otimes  \rB(\cH_n)$.  Apply the same argument for $\rho^B$ to deduce $\Gamma^Q(\rho^A,\rho^B)\subseteq \rB(\textrm{range}\,\rho^A)\otimes  \rB(\textrm{range}\, \rho^B)$.  

Assume that $ \rho^A=|1\rangle\langle1|$ and $\rho^{AB}\in\Gamma^Q(\rho^A,\rho^B)$.   Then $\rho^{AB}=\rho^A\otimes \rho^B$.
\end{proof}

More information concerning the partial trace and its properties
can be found in a recent work \cite{FKV18}. 

The following results are used in the proof of Proposition \ref{swapfidprop}:
\begin{lemma}\label{ptcombS}  Denote by $S_N$ the SWAP operator on $\cH_{N^2}:=\cH_N\otimes \cH_N$, and by $S_{n,m}$ and $R_{n,m}$ the following SWAP operators on $\cH_{(nm)^2}:=\cH_n\otimes \cH_m\otimes \cH_n\otimes \cH_m$:
\begin{equation*}
S_{n,m}(|\bx\rangle|\bu\rangle|\by\rangle|\bv\rangle)=|\by\rangle|\bv\rangle|\bx\rangle|\bu\rangle,\,
R_{n,m}(|\bx\rangle|\bu\rangle|\by\rangle|\bv\rangle)=|\bx\rangle|\by\rangle|\bu\rangle|\bv\rangle.
\end{equation*}
\begin{enumerate}[(a)]
\item Assume that $|i\rangle$, with $i\in[N]$, is an orthonormal basis in $\cH_N$.  Suppose that  
$$\rho=\sum_{i,j,p,q\in[N]} \rho_{(i,p)(j,q)}|i\rangle|p\rangle \langle j|\langle q|\in \rB(\cH_N\otimes \cH_N).$$
Then 
$\tr S_N\rho=\sum_{i,p\in[N]}  \rho_{(p,i)(i,p)}$.
\item Assume that 
\begin{align*}
& \rho^{AB}\in\rB(\cH_n\otimes\cH_n), && \tr_B\rho^{AB}=\rho^A \in \rB(\cH_n), && \tr_A\rho^{AB}=\rho^B\in \rB(\cH_n), \\
& \sigma^{CD}\!\in\rB(\cH_m\otimes\cH_m), && \tr_D\sigma^{CD}\!=\sigma^C \in \rB(\cH_m), && \tr_C\sigma^{CD}\!=\sigma^D\in \rB(\cH_m).
\end{align*}
Then $\tau^{ACBD}:=R_{n,m}(\rho^{AB}\otimes\sigma^{CD})R_{n,m}$ is in $\rB(\cH_{(nm)^2})$.   Furthermore
\begin{equation}\label{ptcombS1} 
\begin{gathered}
\tr_{BD} \tau^{ACBD}=\rho^A\otimes\sigma^C, \quad \tr_{AC} \tau^{ACBD}=\rho^B\otimes\sigma^D,\\
\tr S_{n,m}\tau^{ACBD}=\big(\tr S_n\rho^{AB}\big)\big(\tr S_m \sigma^{CD}).
\end{gathered}
\end{equation}
\end{enumerate}
\end{lemma}

\begin{proof}
(a) View $S_N$ and $\rho$ as $N^2\times N^2$ matrices with entries indexed by the row $(i,p)$ and the column $(j,q)$.  Observe that $S_N$ is a symmetric permutation matrix.  Then  $(S_N\rho)_{(i,p),(j,q)}=\rho_{(p,i)(j,q)}$.  The trace of $S_N\rho$ is obtained by summation on the entries $p=q$ and $i=j$.

\noindent
Clearly,  $\tau^{ACBD}\in \rB(\cH_{(nm)^2})$.  Assume that 
\begin{align*}
\rho^{AB} & =\sum_{i_A,i_B,j_A,j_B \in [n]} \rho_{(i_A,i_B)(j_A,j_B)}|i_A\rangle|i_B\rangle\langle j_A|\langle j_B|,\\
\sigma^{CD} & =\sum_{p_C,p_D,q_C,q_D \in [m]} \sigma_{(p_C,p_D)(q_C,q_D)}|p_C\rangle|p_D\rangle\langle q_C|\langle q_D|.
\end{align*}
Then
\begin{align*}
\tau^{ACBD} = \!\!\!\!\! \sum_{\substack{i_A,i_B,j_A,j_B \in [n] \\ p_C,p_D,q_C,q_D \in [m]}} \!\!\!\!\!
\rho_{(i_A,i_B)(j_A,j_B)} \sigma_{(p_C,p_D)(q_C,q_D)}
|i_A\rangle |p_C\rangle |i_B\rangle |p_D\rangle
\langle j_A|\langle q_C| \langle j_B| \langle q_D|.
\end{align*}
Observe next that $\tr_{BD}\tau^{ACBD}$ is obtained when we sum on $i_B=j_B$ and $p_D=q_D$.  Hence 
\begin{align*}
\tr_{BD}\tau^{ACBD} & = \!\! \sum_{\substack{i_A,j_A \in [n] \\ p_C,q_C \in [m]}}
(\sum_{i_B=1}^n \rho_{(i_A, i_B)(j_A,i_B)})(\sum_{p_D=1}^m \sigma_{(p_C,p_D)(q_C,p_D)})|i_A\rangle |p_C\rangle \langle j_A|\langle q_C|\\
& = \!\! \sum_{\substack{i_A,j_A \in [n] \\ p_C,q_C \in [m]}} \rho^A_{i_Aj_A}\sigma^C_{p_Cq_C}|i_A\rangle |p_C\rangle \langle j_A|\langle q_C|=\rho^A\otimes\sigma^C.
\end{align*}
Similarly $\tr_{AC}\tau^{ACBD}=\rho^B\otimes\sigma^D$.
This proves the first line in \eqref{ptcombS1}. 

We now use (a) to compute $\tr S_{n,m}\tau^{ACBD}$:
\begin{align*}
\tr S_{n,m}\tau^{ACBD}= \!\! \sum_{\substack{i_A,i_B \in [n] \\ p_C,p_D \in [m]}} \rho_{(i_B,i_A)(i_A,i_B)}\sigma_{(p_D,p_C)(p_C,p_D)}=
(\tr S_n\rho^{AB})(\tr S_m\sigma^{CD}).
\end{align*}
This proves the second line in \eqref{ptcombS1}. 
\end{proof}
\section{Quantum states of a single qubit system}\label{apb:qubits}
In this Appendix  we discuss additional properties of the quantum optimal transport for qubits.
Subsection~\ref{subsec:Blochb}
provides (Theorem \ref{thmC3}) a closed formula for $\rT_{C^Q}^Q(\rho^A,\rho^B)$ in terms of solutions of the trigonometric equation \eqref{Phieq}.  Lemma \ref{6sollem} shows that this trigonometric equation is equivalent to a polynomial equation of degree
at most $6$.  
Subsection \ref{subsec:isospectral} gives a nice closed formula for the value of QOT for two isospectral qubit density matrices. In Subsection \ref{subsec:nonexisF} we present a simple example where the supremum of the dual SDP problem to QOT is not achieved. 

\subsection{A semi-analytic formula for the single-qubit optimal transport}
\label{subsec:Blochb}
We begin by introducing a convenient notation for qubits in the $y=0$ section of the Bloch ball $\Omega_2$ -- see \cite[Section 5.2]{BZ17}. Let $O$ denote the orthogonal
rotation matrix,
\begin{align*}
O(\theta)=\begin{bmatrix}\cos(\theta/2)&-\sin(\theta/2)\\\sin(\theta/2)&\cos(\theta/2)
\end{bmatrix}, \quad \text{for } \theta \in [0,2\pi),
\end{align*}
and define, for $r\in [0,1]$,
\begin{align*}
\rho(r,\theta) & = O(\theta) \begin{bmatrix}r&0\\0&1-r\end{bmatrix}
O(\theta)^\top . 
\end{align*}
Because of unitary invariance \eqref{CQinvQOT}, the quantum transport problem between two arbitrary qubits $\rho^A, \rho^B \in \Omega_2$ can be reduced to the case $\rho^A = \rho(s,0)$ and $\rho^B = \rho(r,\theta)$, with three parameters, $s,r \in [0,1]$ and $\theta \in [0,2\pi)$. The parameter $\theta$ is the angle between the Bloch vectors associated with $\rho^A$ and $\rho^B$. With such a parametrization we can further simplify the single-qubit transport problem.


Observe first that if $s \in \{0,1\}$ then $\rho^A$ is pure, and if $r \in \{0,1\}$ then $\rho^B$ is pure. In any such case an explicit solution of the qubit transport problem is given \eqref{QOTrankone}.
\begin{theorem}\label{thmC3}
Let $\rho^A = \rho(s,0), \rho^B = \rho(r,\theta)$ and assume that $0<r,s<1$. Then
\begin{equation}\label{defT0}
\rT^Q_{C^Q}(\rho^A,\rho^B) = 
\max_{\phi\in \Phi(s,r,\theta)} \frac{1}{4}\Big(\sqrt{1+ (2s-1)\cos\phi} - \sqrt{1+(2r-1)\cos(\theta+\phi)}\Big)^2,\notag
\end{equation}
where $\Phi(s,r,\theta)$ is the set of all $\phi\in[0,2\pi)$ satisfying the equation
\begin{equation}\label{Phieq}
\frac{(2s-1)^2\sin^2\phi}{1+(2s-1)\cos\phi}=\frac{(2r-1)^2\sin^2(\theta+\phi)}{1+(2r-1)\cos(\theta+\phi)}.
\end{equation}
\end{theorem}
\begin{proof}
A unitary $2 \times 2$ matrix $U$ can be parametrized, up to a global phase, with three angles $\alpha, \beta, \phi \in [0,2\pi)$,
\begin{align*}
U = \begin{bmatrix} e^{\bi \alpha} & 0\\ 0 & e^{-\bi \alpha} \end{bmatrix} O(\phi) \begin{bmatrix} e^{\bi \beta} & 0\\ 0 & e^{-\bi \beta} \end{bmatrix}.
\end{align*}
Thus, setting $f(r,\theta;\alpha,\phi) = (U^{\dagger}\rho(r,\theta) U)_{11}$, we have
\begin{align*}
f(r,\theta;\alpha,\phi) =  \frac{1}{2} \Big( 1+ (2 r-1) \big( \cos (\theta ) \cos (\phi ) + \cos (2 \alpha ) \sin (\theta ) \sin (\phi ) \big) \Big).
\end{align*}
This quantity does not depend on the parameter $\beta$, so we can set $\beta = 0$. Note also that $f(s,0;\alpha,\phi)$ does not depend on $\alpha$. With $\rho^A = \rho(s,0), \rho^B = \rho(r,\theta)$, Theorem \ref{lowbdT}  yields
\begin{align*}
\rT^Q_{C^Q}(\rho^A,\rho^B) = \frac12\max_{\alpha,\phi\in [0,2\pi)} \Big( \sqrt{f(s,0;0,\phi)} - \sqrt{f(r,\theta;\alpha,\phi)} \Big)^2.
\end{align*}
Now, note that the equation $\partial_\alpha f(r,\theta;\alpha,\phi) = 0$ yields the extreme points $\alpha_0 = k \pi /2$, with $k \in \mathbb{Z}$. Since $f(r,\theta;\alpha + \pi,\phi) = f(r,\theta;\alpha,\phi)$ we can take just $\alpha_0 \in \{0,\pi/2\}$. Consequently,
\begin{align*}
\rT^Q_{C^Q}(\rho^A,\rho^B) = 
\max_{\phi \in [0,2\pi)} \{ g_-(s,r,\theta;\phi), g_+(s,r,\theta;\phi) \},
\end{align*}
where we introduce the auxilliary functions
\begin{align}\label{g}
 g_\pm(s,r,\theta;\phi) = \frac{1}{4}\Big(\sqrt{1+ (2s-1)\cos\phi} - \sqrt{1+(2r-1)\cos(\theta\pm\phi)}\Big)^2.
\end{align}
But since $g_-(s,r,\theta;2\pi - \phi) = g_+(s,r,\theta;\phi)$ we can actually drop the $\pm$ index in the above formula. In conclusion, we have shown that it is sufficient to take $U = O(\phi)$ for $\phi \in [0,2\pi)$ in Formula \eqref{lowbdT0a}.

Finally, it is straightforward to show that the equation $\partial_\phi g(s,r,\theta;\phi) = 0$ is equivalent to \eqref{Phieq}. Hence, $\Phi(s,r,\theta)$ is the set of extreme points, and \eqref{defT0} follows.
\end{proof}


%
%
\begin{lemma}\label{6sollem}  The equation \eqref{Phieq} has at most six solutions $\phi\in[0,2\pi)$ for given $r,s\in(0,1), \theta\in[0,2\pi)$.  Moreover there is an open set of $s,r\in (0,1),\theta\in[0,2\pi)$ where there are exactly six 
 distinct solutions.
\end{lemma}
\begin{proof}
Write $z=e^{\bi\phi}, \zeta=e^{\bi\theta}$.  Then
\begin{align*}
& 2\cos\phi =z+\frac{1}{z}, &&  2\bi\sin \phi=z-\frac{1}{z}, \\
& 2\cos(\theta+\phi)=\zeta z+\frac{1}{\zeta z}, && 2\bi\sin(\theta+\phi)=\zeta z-\frac{1}{\zeta z}.
\end{align*}
Thus \eqref{Phieq} is equivalent to
\begin{multline}
(1-2 r)^2 \left[ (2 s-1) \left(z^2+1\right)+2 z\right] \left(\zeta ^2 z^2-1\right)^2 \label{z6}\\
 -\zeta  (1-2 s)^2 \left(z^2-1\right)^2 \left[ (2 r-1) \left(\zeta ^2 z^2+1\right)+2 \zeta  z\right] = 0.
\end{multline}
This a 6th order polynomial equation in the variable $z$, so it has at most 6 real solutions.  Since we must have $\vert z \vert = 1$, not every complex root of \eqref{z6} will yield a real solution to the original  \eqref{Phieq}. Nevertheless, it can be shown  that there exist open sets in the parameter space $s,r \in (0,1)$, $\theta \in [0,2\pi)$ on which \eqref{Phieq} does have 6 distinct solutions. 

Observe that if $\theta=0$ and $s,r\in(0,1)$ and $s\ne r$  then two  solutions to the equality \eqref{Phieq} are $\phi\in\{0,\pi\}$, which means that $z=\pm 1$.
In this case the equality \eqref{Phieq} is
\begin{equation*}
\sin^2\phi \: \bigg(\frac{(2s-1)^2}{1+(2s-1)\cos\phi}-\frac{(2r-1)^2}{1+(2r-1)\cos(\phi)}\bigg)=0.
\end{equation*}
As $\sin^2\phi=-(1/4)z^{-2}(z^2-1)^2$ we see that $z=\pm 1$ is a double root.

Another solution $\phi\notin\{0,\pi\}$ is given by
\begin{equation*}
\cos\phi=\frac{(2s-1)^2-(2r-1)^2}{(2r-1)^2(2s-1)-(2r-1)(2s-1)^2}=\frac{2(1-r-s)}{(2r-1)(2s-1)}.
\end{equation*}
Assume that $r+s=1$.  Then $\cos\phi=0$, so $\phi\in\{\pi/2, 3\pi/2\}$.  Thus if $r+s$ is close to $1$ we have that $\phi$ has two values close to $\pi/2$ and $3\pi/2$ respectively.  Hence in this case we have $6$ solutions counting with multiplicities.

We now take a small $|\theta|>0$.  The two simple solutions $\phi$ are close to $\pi/2$ and $3\pi/2$.  We now need to show that the double roots $\pm 1$ split to two pairs of solutions on the unit disc: one pair close to $1$ and the other pair close to $-1$.
Let us consider the pair close to $1$, i.e., $\phi$ close to zero.  Then the equation \eqref{Phieq} can be written in the form
\begin{multline*}
(2s-1)^2\big(1+(2r-1)\cos(\theta+\phi)\big)\sin^2\phi\\
- (2r-1)^2\big(1+(2s-1)\cos\phi\big)\sin^2(\theta+\phi)=0.
\end{multline*}
Replacing $\sin\phi, \sin(\theta+\phi)$ by $\phi, \theta+\phi$ respectively we see that the first term gives the equation:
$(2s-1)^2(2r)\phi^2-(2r-1)^2 2s(\theta+\phi)^2=0$. Then we obtain two possible Taylor series of $\phi$ in terms of $\theta$: 
\begin{align*}
\phi_1(\theta) & =\frac{(2r-1)\sqrt{s}\theta}{(2s-1)\sqrt{r}-(2r-1)\sqrt{s}} + \theta^2 E_1(\theta), \\
\phi_2(\theta) & =-\frac{(2r-1)\sqrt{s}\theta}{(2s-1)\sqrt{r}+(2r-1)\sqrt{s}}+\theta^2E_2(\theta).
\end{align*}
Use the implicit function theorem to show that $E_1(\theta)$ and $E_2(\theta)$ are analytic in $\theta$ in the neighborhood of $0$.
Hence in this case we have $6$ different solutions. 
\end{proof}

We have thus shown that the general solution of the quantum transport problem of a single qubit with cost matrix $C^Q = \tfrac{1}{2} \big(\I_{4} - S\big)$ is equivalent to solving a 6th degree polynomial equation with certain parameters. For some specific values of these parameters an explicit analytic solution can be given. This is discussed in the next subsection.

\subsection{Two isospectral density matrices of a single qubit}
\label{subsec:isospectral}
In view of unitary invariance \eqref{CQinvQOT} and the results of the previous section we can assume that  two isospectral qubits have the following form:
 $\rho^A = \rho(s,0)$ and $\rho^B = \rho(s,\theta)$ for some $s \in [0,1]$ and $\theta \in [0,2\pi)$.
\begin{theorem}\label{thmE1}
 For any $s \in [0,1]$ and $\theta \in [0,2\pi)$ we have
\begin{equation}\label{Tiso}
\rT^Q_{C^Q} \big(\rho(s,0),\rho(s,\theta) \big) = \Big( \tfrac{1}{2} -\sqrt{s(1-s)} \Big) \sin^2 (\theta/2).
\end{equation}
\end{theorem}
\begin{proof}
Note first that if the states $\rho^A,\rho^B$ are pure, i.e. $s = 0$ or $s=1$, formula \eqref{Tiso} gives $\rT^Q_{C^Q} \big(\rho(s,0),\rho(s,\theta) \big) = \tfrac{1}{2} \sin^2 (\theta / 2)$, which agrees with \eqref{QOTrankone}. 

From now on we assume that that $\rho^A, \rho^B$ are not pure.  When $r = s$, \eqref{z6} simplifies to the following:
\begin{multline}
 (\zeta -1) (1-2 s)^2 \left(\zeta  z^2-1\right) \times  \\
\quad \times  \left[4 s (\zeta +1) \left(\zeta  z^2+1\right) z +(2 s-1) (z-1)^2 (\zeta  z-1)^2 \right]= 0. \label{zeta_iso}
\end{multline}

Eq.\ \eqref{zeta_iso} is satisfied when $z = \pm \zeta^{-1/2}$. This corresponds to $\phi_0 = -\theta/2$ or $\phi_0' = \pi - \theta/2$. Observe, however, that we have $g(s,s,\theta;\phi_0) = g(s,s,\theta;\phi_0') = 0$, so we can safely ignore $\phi_0, \phi_0' \in \Phi(s,s,\theta)$ in the maximum in  \eqref{defT0}.

Hence, we are left with a 4th order equation
\begin{align}\label{zeta4}
4 s (\zeta +1) \left(\zeta  z^2+1\right) z +(2 s-1) (z-1)^2 (\zeta  z-1)^2 = 0,
\end{align}
which
reads
\begin{equation}\label{phi4}
(2 s-1) \big[ 2 + \cos (\theta +2 \phi )+ \cos (\theta ) \big]  
+2 \big[ \cos (\theta +\phi )+ \cos (\phi ) \big] = 0.
\end{equation}
Now, observe that if $\phi$ satisfies  \eqref{phi4}, then so does $\phi' = -\phi - \theta$. This translates to the fact that if $z$ satisfies  \eqref{zeta4}, then so does $(z \zeta)^{-1}$. Furthermore, $g(s,s,\theta;\phi) = g(s,s,\theta;\phi')$. Hence, in the isospectral case we are effectively taking the maximum over just two values of $\phi$.

Let us now seek an angle $\phi_1 \in [0,2\pi)$ such that $g(s,s,\theta;\phi_1)$ equals the righthand side of  \eqref{Tiso}. The latter equation reads 
\begin{align*}
& \Big\{ (2 s-1) \big[\cos \left(\theta +\phi _1\right)+\cos \left(\phi _1\right)\big]
 -\big(2 \sqrt{s(1-s)}-1\big) \big(\cos (\theta )-1\big)+2\Big\}^2 \\
& \qquad\quad = 4 \big[(2 s-1) \cos \left(\phi _1\right)+1\big] \big[(2 s-1) \cos
   \left(\theta +\phi _1\right)+1\big].
\end{align*}
In terms of $z$ and $\zeta$, the above is equivalent to a 4th order polynomial equation in $z$, which can be recast in the following form:
\begin{align}\label{zeq}
\Big[ \zeta  (1-2 s) z^2+(\zeta +1) \big(2 \sqrt{s(1-s)}-1\big) z-2 s+1 \big]^2 = 0.
\end{align}
Hence,  \eqref{zeq} has two double roots:
\begin{multline*}
 z_1^{\pm} = \big[ 2 \zeta  (1-2 s) \big]^{-1} \bigg\{ (\zeta +1) \big( 1-2 \sqrt{s(1-s)} \, \big) \\
 \pm \sqrt{(\zeta +1)^2 \big(1-2 \sqrt{s(1-s)} \,\big)^2-4 \zeta  (1-2s)^2} \bigg\}.
\end{multline*}
Furthermore, one can check that $z_1^{-} = (\zeta z_1^{+})^{-1}$.

Now, it turns out that $z_1^{\pm}$ are also solutions to \eqref{zeta4}, as one can quickly verify using \textsc{Mathematica}~\cite{Mathematica}. We thus conclude that $\phi_1, \phi_1' \in \Phi(s,s,\theta)$.

We now divide the polynomial in \eqref{zeta4} by $(z-z_1^{+})(z-z_1^{-})$. We are left with the following quadratic equation
\begin{align*}
\zeta  \Big[ (2 s-1) \left(\zeta  z^2+1\right)+(\zeta +1) \big(2 \sqrt{(1-s) s}+1\big) z\Big] = 0.
\end{align*}
Its solutions are
\begin{multline*}
z_2^{\pm} = \big[ 2 \zeta  (1-2 s) \big]^{-1} \bigg\{ (\zeta +1) \big( 1+2 \sqrt{s(1-s)} \, \big) \\
\pm \sqrt{(\zeta +1)^2 \big(1+2 \sqrt{s(1-s)} \,\big)^2-4 \zeta  (1-2s)^2} \bigg\}.
\end{multline*}
Again, we have $z_2^{-} = (\zeta z_2^{+})^{-1}$, in agreement with the symmetry argument. Setting $z_2^+ \cv e^{\bi \phi_2}$ and $z_2^- \cv e^{\bi \phi_2'}$ we have $\phi_2, \phi_2' \in \Phi(s,s,\theta)$. Then we deduce
that
\begin{align*}
g(s,s,\theta;\phi_2) & = g(s,s,\theta;\phi_2')  = \tfrac{1}{4} \Big[ (1-6 \sqrt{(1-s) s} - \big(1+2 \sqrt{(1-s) s} \, \big) \cos (\theta ) \Big].
\end{align*}

Finally, we observe that
\begin{align*}
g(s,s,\theta;\phi_1) - g(s,s,\theta;\phi_2) = \sqrt{(1-s) s} \, \big(1+\cos (\theta ) \big) \geq 0.
\end{align*}
This shows that, for any $s \in (0,1)$, $\theta \in [0,2\pi)$,
\begin{align*}
\rT^Q_{C^Q} \big(\rho(s,0),\rho(s,\theta) \big) = g(s,s,\theta;\phi_1),
\end{align*}
and \eqref{Tiso} follows.
\end{proof}

Note that $g(s,s,\theta;\phi_2)$ can become negative for certain values of $s$ and $\theta$. This means that for such values
the set $\Phi$ of phases defined in Theorem \ref{thmC3}
reads,
 $\Phi(s,s,\theta) = \{\phi_0,\phi_0',\phi_1,\phi_1'\}$.

\subsection{An example where the supremum  \eqref{dualQOT1} is not achieved}\label{subsec:nonexisF}
Assume that $m=n=2$, $ C=C^Q$, $ \rho^A=\vert 0 \rangle \langle 0 \vert = \left[ \begin{smallmatrix}
1 & 0 \\ 0 & 0
\end{smallmatrix} \right]$ and $ \rho^B=\mathbb{I}_2/2$. 
Recall that in such a case, $\Gamma^Q(\rho^A,\rho^B)=\{\rho^A\otimes\rho^B\}$ and
\begin{eqnarray*}
 \rho^A\otimes \rho^B=\left[\begin{array}{rrrr}\frac{1}{2}&0&0&0\\0&\frac{1}{2}&0&0\\0&0&0&0\\0&0&0&0\end{array}\right].
\end{eqnarray*}
Hence $\rT^Q(\rho^A,\rho^B)=1/4$.
We can easily see that the supremum in~\eqref{dualQOT1} is not attained in this case.  Let $F$ be of the form \eqref{sigmaFform}.  Suppose that there exists $\sigma^A,\sigma^B\in\rS(\cH_2)$ such that $F\ge 0$ and $\rT_{C^Q}^Q(\rho^A,\rho^B)=\tr (\sigma^A\rho^A+\sigma^B\rho^B)$. As in the proof of Theorem \ref{dualQOT} we deduce that 
$\tr F( \rho^A\otimes \rho^B)=0$.  Hence the $(1,1)$ and $(2,2)$ entries of $F$ are zero.  Since $F\ge 0$ it follows that the first and the second row and column of $F$ are zero.
Observe next that the $(2,3)$ and $(3,2)$ entries of $F$ are $-1/2$. Hence such $\sigma^A,\sigma^B$ do not exist. 

Since $\rho^A$ is not positive definite and $\rho^B$ is positive definite, as pointed out in the proof of Proposition \ref{redprop}, one can replace $\rho^{A}$ by $\rho^{A'}=[1]\in\Omega_1$.    Then the dual problem for $\rho^{A'},\rho^B$ boils down to
\begin{align*}
& \sigma^{A'}=-a',\quad\sigma^B=\diag(-e,-g),\quad F=\diag(a'+e,a'+g+1/2)\ge 0,\\
& \max_{a'+e\ge 0, a'+g+1/2\ge 0} \Big( -a'-\frac{e+g}{2} \Big).
\end{align*}
Then the above maximum is $1/4$, achieved for $a'=-1/2 +t, e=1/2-t, g=-t$ for each $t\in\R$.

To summarize: the supremum of the dual problem to $\rT_{C^Q}^Q(\rho^A,\rho^B)$ is achieved at \emph{infinity},  while the supremum of the dual problem to $\rT_{C^Q}^Q(\rho^{A'},\rho^B)$ is achieved on an unbounded set.
\section{Diagonal states of a qutrit}\label{subsec:diagqut}
In this section we provide a closed formula for $\rT_{C^{Q}}^Q(\diag(\bs),\diag(\bt))$ for diagonal qutrits,
for many of pairs $\bs,\bt$.  
\begin{theorem}\label{diaggen3}  Let $\bs=(s_1,s_2,s_3)^\top,\bt=(t_1,t_2,t_3)^\top\in\R^3$ be  probability vectors.
Then the quantum optimal transport problem 
for diagonal qutrits 
is determined by the given formulas in the following cases:
\begin{enumerate}[(a)]
\item 
\begin{equation*}
\rT_{C^{Q}}^Q(\diag(\bs),\diag(\bt))=\frac{1}{2}\max_{p\in [3]}(\sqrt{s_p}-\sqrt{t_p})^2
\end{equation*}
if and only if the conditions \eqref{eqlowbdiagn1} hold for $n=3$.  
\item Suppose that 
 there exists a renaming of $1,2,3$ by $p,q,r$ 
 such that
\begin{equation}\label{condminb}
\begin{aligned}
& t_r\ge s_p+s_q \textrm{ and}\\
& \textrm{either } s_p\ge t_p>0, s_q\ge t_q>0 \; \textrm{ or } \; t_p\ge s_p>0, t_q\ge s_q >0.
\end{aligned}
\end{equation}
Then
\begin{equation}\label{QOTbchar}
\rT_{C^{Q}}^Q(\diag(\bs),\diag(\bt))=\frac{1}{2}\Big((\sqrt{s_p}-\sqrt{t_p})^2 +(\sqrt{s_q}-\sqrt{t_q})^2\Big).
\end{equation}
\item  Suppose that
there exists $\{p,q,r\}=\{1,2,3\}$ such that 
\begin{equation}\label{condmin1}
s_p>t_q>0, \quad t_p>s_q>0, \quad s_q+s_r\ge t_p,
\end{equation}
and
\begin{equation}\label{posa+b2}
\begin{aligned}
& 1 +\frac{\sqrt{t_q}}{\sqrt{s_q}}-\sqrt{\frac{s_p-t_q}{t_p-s_q}}\ge  0, \qquad 
1+\frac{\sqrt{s_q}}{\sqrt{t_q}}-\sqrt{\frac{t_p-s_q}{s_p-t_q}}\ge 0,\\
& \left(1+\frac{\sqrt{t_q}}{\sqrt{s_q}}-\sqrt{\frac{s_p-t_q}{t_p-s_q}} \,\right)
\left(1+\frac{\sqrt{s_q}}{\sqrt{t_q}}-\sqrt{\frac{t_p-s_q}{s_p-t_q}} \, \right)\ge 1,
\\
& \max \Big(\frac{s_q}{t_q},\frac{t_q}{s_q}\Big)\ge \max \Big(\frac{s_p-t_q}{t_p-s_q},\frac{t_p-s_q}{s_p-t_q} \Big).
\end{aligned}
\end{equation}
Then
\begin{equation}\label{QOPTd3char1}
\rT_{C^{Q}}^Q(\diag(\bs),\diag(\bt))= \frac{1}{2}\Big((\sqrt{s_q}-\sqrt{t_q})^2 +(\sqrt{s_p-t_q}-\sqrt{t_p-s_q})^2\Big).
\end{equation}
\item Assume that $\bs=(s_1,s_2,0)^\top,\bt=(t_1,t_2,t_3)^\top$ are probability vectors.  Then
\begin{equation}\label{QOTspqutritf}
\rT^Q_{C^Q}(\diag(\bs),\diag(\bt))=
\begin{cases}
\frac{1}{2}\big((\sqrt{t_1}-\sqrt{t_2})^2+t_3\big), & \textrm{ if } s_1\ge t_2 \textrm{ and } s_2\ge t_1,\\
\frac{1}{2}\big((\sqrt{t_1}-\sqrt{s_1})^2+t_3\big), & \textrm{ if } s_1< t_2 ,\\
\frac{1}{2}\big((\sqrt{t_2}-\sqrt{s_2})^2+t_3 \big), & \textrm{ if } s_2< t_1.
\end{cases}
\end{equation}
If $\bs=(s_1,s_2,s_3)^\top,\bt=(t_1,t_2,0)^\top$, then formula \eqref{QOTspqutritf} holds after the swapping $s_i \leftrightarrow t_i$.
\end{enumerate}
\end{theorem}
\begin{proof} (a) This follows from Theorem \ref{lowbdTdiagn}.

\noindent
(b)   Suppose that the condition \eqref{condminb} holds.  By relabeling the coordinates and interchanging $\bs$ and $\bt$ if needed we can assume the conditions \eqref{condminb} are satisfied with $p=1,q=2, r=3$:
\begin{equation*}
s_1\ge t_1>0,\quad s_2\ge t_2>0, \quad t_3\ge s_1+s_2.
\end{equation*} 
Hence 
\begin{equation}\label{Bopt}
X^\star=\begin{bmatrix}0&0&s_1\\0&0&s_2\\t_1&t_2&t_3-(s_1+s_2)\end{bmatrix}
\in\Gamma^{cl}(\bs,\bt).
\end{equation}
We claim that the conditions \eqref{condminb} yield that $X^\star$ is a minimizing matrix for 

\noindent
$\rT^Q_{C^Q}(\diag(\bs),\diag(\bt))$ as given in \eqref{clqotdiagdm0}.  To show that we use the complementary conditions in Lemma \ref{compcondlem}.
Let $R^\star\in \Gamma^Q(\diag(\bs),\diag(\bt))$ be the matrix induced by $X^\star$ of the form
described in part (a) of Lemma \ref{diaglemobs}.   That is, the diagonal entries of $R^\star$ 
are $R^\star_{(i,j)(i,j)}=x^\star_{ij}$ with additional nonnegative entries:
$R^\star_{(i,j)(j,i)}=\sqrt{x^\star_{ij}x^\star_{ji}}$ for $i\ne j$. 
Clearly, $R^\star $ is a direct sums of $3$ submatrices of order $1$ and $3$ of order $2$ as above.    Let  $F^\star$ be defined as in Lemma~\ref{compcondlem} with the following parameters:
\begin{equation}\label{ab3id1}
\begin{aligned}
& a_{1}^\star=\frac{1}{2}\Big(\frac{\sqrt{t_1}}{\sqrt{s_1}}-1\Big), \qquad b_{1}^\star=\frac{1}{2}\Big(\frac{\sqrt{s_1}}{\sqrt{t_1}}-1\Big),\\
& a_{2}^\star=\frac{1}{2}\Big(\frac{\sqrt{t_2}}{\sqrt{s_2}}-1\Big), \qquad b_{2}^\star=\frac{1}{2}\Big(\frac{\sqrt{s_2}}{\sqrt{t_2}}-1\Big),\\
& a_3^\star=b_3^\star=0.
\end{aligned}
\end{equation}
We claim that  the conditions \eqref{condminb} yield that $F^\star$ is positive semidefinite.  
We verify that the three blocks of size one and the three  blocks of size two 
of $F^\star$ are positive semidefinite.  The condition $a_i^\star+b_i^\star\ge 0$ for $i\in[3]$ is straightforward.  The conditions for $M_{12}^\star$ and $M_{13}^\star$ are straightforward.  We now show that $M_{12}^\star$ is positive semidefinite.
First note that as $s_1\ge t_1$ and $s_2\ge t_2$ we get that $b_1^\star\ge 0$ and $b_2^\star\ge 0$.  Clearly $a_1^\star>-1/2$ and $a_2^\star>-1/2$.  Hence the diagonal entries of  $M_{12}^\star$ are positive.  It is left to show that $\det M_{12}^\star\ge 0$.  Set $u=\sqrt{t_1}/{\sqrt{s_1}}\le 1$ and $v=\sqrt{s_2}/{\sqrt{t_2}}\ge 1$.
Then
\begin{align*}
2(a_1^\star+b_2^\star+1/2)=u+v-1, && 2(a_2^\star+b_1^\star +1/2)=1/u+1/v-1,
\end{align*}
\begin{align*}
4\det M_{12}^\star & =(u+v-1)(1/u+1/v-1)-1 \\
& =\big(1/(uv\big) \big)\big(u+v-1)(u+v-uv)-uv\big) \\
& =\big(1/(uv)\big) (u+v)(1-u)(v-1)\ge 0.
\end{align*}
We next observe that equalities \eqref{compcond} hold.  The first three equalities hold
as $x_{11}^\star=x_{22}^\star=(a_3^\star +b_3^\star)=0$.  The equality of $i=1,j=2$ holds as $x_{12}^\star=x_{21}^\star=0$.  The equalities for $i=1, j=3$ and $i=2,j=3$ follow from the following equalities:
\begin{align*}
x_{13}^\star(a_1^\star+b_3^\star +1/2)+x_{31}^\star(a_3^\star+b_1^\star +1/2)=
\tfrac{1}{2}\big(s_1\tfrac{\sqrt{t_1}}{\sqrt{s_1}}+t_1\tfrac{\sqrt{s_1}}{\sqrt{t_1}}\big)=\sqrt{s_1t_1}=\sqrt{x_{13}^\star x_{31}^\star},\\
x_{23}^\star(a_2^\star+b_3^\star +1/2)+x_{32}^\star(a_3^\star+b_2^\star +1/2)=
\tfrac{1}{2}\big(s_2\tfrac{\sqrt{t_2}}{\sqrt{s_2}}+t_2\tfrac{\sqrt{s_2}}{\sqrt{t_2}}\big)=\sqrt{s_2t_2}=\sqrt{x_{23}^\star x_{32}^\star}.
\end{align*}
Hence $\tr R^\star F^\star=0$ and $X^\star$ is a minimizing matrix.  Therefore \eqref{QOTbchar} holds for $p=1$, $q=2$.

\bigskip

\noindent
(c)  Suppose that the condition \eqref{condmin1} holds.  By relabeling the coordinates we can assume the conditions \eqref{condmin1} are satisfied with $p=1,q=2, r=3$:
\begin{equation*}
s_1>t_2,\quad t_1> s_2, \quad s_2+s_3-t_1\ge 0.
\end{equation*} 
Hence 
\begin{equation}\label{A2opt}
X^\star=\begin{bmatrix}0&t_2&s_1-t_2\\s_2&0&0\\t_1-s_2&0&s_2+s_3-t_1\end{bmatrix}\in\Gamma^{cl}(\bs,\bt).
\end{equation}
We claim that the conditions \eqref{posa+b2} yield that $X^\star$ is a minimizing matrix for 

\noindent
$\rT^Q_{C^Q}(\diag(\bs),\diag(\bt))$ as given in \eqref{clqotdiagdm0}.  To show this we use the complementary conditions in Lemma \ref{compcondlem}.
Let $R^\star\in \Gamma^Q(\diag(\bs),\diag(\bt))$ be the matrix induced by $X^\star$ of the form
described in part (a) of Lemma \ref{diaglemobs}.   
Recall that $R^\star $ is a direct sum of $3$ submatrices of order $1$ and $3$ of order $2$ as above.
 Let  $F^\star$ correspond to
\begin{equation}\label{abstarceq}
\begin{aligned}
a^\star_1=\frac{1}{2}\Big(\frac{\sqrt{t_1-s_2}}{\sqrt{s_1-t_2}}-1\Big), && a_2^\star= 
\frac{1}{2}\Big(\frac{\sqrt{t_2}}{\sqrt{s_2}}-\sqrt{\frac{s_1-t_2}{t_2-s_1}} \, \Big),&& a_3^\star=0,\\
b^\star_1=\frac{1}{2}\Big(\frac{\sqrt{s_1-t_2}}{\sqrt{t_1-s_2}}-1\Big), &&
b_2^\star=\frac{1}{2}\Big(\frac{\sqrt{s_2}}{\sqrt{t_2}}-\sqrt{\frac{t_1-s_2}{s_1-t_2}} \,\Big), && b_3^\star=0.
\end{aligned}
\end{equation}
We claim that  \eqref{posa+b2} yield that $F^\star$ is positive semidefinite.  
We verify that the three blocks of size one and the three blocks of size two matrices of $F^\star$ are positive semidefinite.  The condition $a_1^\star+b_1^\star\ge 0$ is straightforward.
To show the condition $a_2^\star+b_2^\star\ge 0$ we argue as follows.
Let 
\begin{equation*}
u=\frac{\sqrt{t_1}}{\sqrt{s_1}}, \quad v=\sqrt{\frac{s_1-t_2}{t_2-s_1}}.
\end{equation*}
Then $2(a_2^\star+b_2^\star)=u+1/u-(v+1/v)$.  The fourth condition of \eqref{posa+b2}
is $\max(u,1/u)\ge \max(v,1/v)$.  As $w+1/w$ increases on $[1,\infty)$ we deduce that 
$a_2^\star+b_2^\star\ge 0$.  Clearly $a_3^\star+b_3^\star=0$.  
We now show that the matrices \eqref{defMij} are positive semidefinite, where the last three inequalities follow from the first three inequalities of \eqref{posa+b2}:
\begin{align*}
& 2(a_1^\star+b_2^\star +1/2)=\frac{\sqrt{s_2}}{\sqrt{t_2}}>0, \qquad\qquad 2(a_2^\star+b_1^\star +1/2)=\frac{\sqrt{t_2}}{\sqrt{s_2}}>0,\\
& (a_1^\star+b_2^\star +1/2)(a_2^\star+b_1^\star +1/2)-1/4=0,\\
& 2(a_1^\star+b_3^\star +1/2)=\frac{\sqrt{t_1-s_2}}{\sqrt{s_1-t_2}}>0, \qquad
 2(a_3^\star+b_1^\star+1/2)=\frac{\sqrt{s_1-t_2}}{\sqrt{t_1-s_2}}>0,\\
& (a_1^\star+b_3^\star +1/2)(a_3^\star+b_1^\star +1/2)-1/4=0,\\
& 2(a_2^\star+b_3^\star +1/2)=\frac{\sqrt{s_2}}{\sqrt{t_2}}-\sqrt{\frac{t_1-s_2}{s_1-t_2}}+1\ge 0,\\
& 2(a_3^\star+b_2^\star +1/2)=\frac{\sqrt{t_2}}{\sqrt{s_2}}-\sqrt{\frac{s_1-t_2}{t_1-s_2}}+1\ge 0,\\
& (a_2^\star+b_3^\star +1/2)(a_3^\star+b_2^\star +1/2)-1/4\ge 0.
\end{align*}

Moreover, the conditions \eqref{compcond} hold: as $x_{11}^\star=x_{22}^\star= a_3^\star+b^\star_3=0$ the first three conditions of \eqref{compcond} hold, and as $x_{23}^\star=x_{32}^\star=0$ the second conditions of  \eqref{compcond} for $p=2,q=3$ trivially hold.  The other two conditions follow from the following equalities:
\begin{align*}
& x_{12}^\star(a_{1}^\star +b_{2}^\star +1/2) + x_{21}^\star(a_{2}^\star +b_{1}^\star+1/2) -\sqrt{x_{12}^\star x_{21}^\star} \\
& \hspace*{3.5cm }=
t_2\frac{\sqrt{s_2}}{2\sqrt{t_2}}+s_2\frac{\sqrt{t_2}}{2\sqrt{s_2}}-\sqrt{t_2 s_2}=0,\\
& x_{13}^\star(a_{1}^\star +b_{3}^\star +1/2) + x_{31}^\star(a_{3}^\star +b_{1}^\star+1/2) -\sqrt{x_{13}^\star x_{31}^\star}\\
& \hspace*{3.5cm }=
(s_1-t_2)\frac{\sqrt{t_1-s_2}}{2\sqrt{s_1-t_2}}+s_2\frac{\sqrt{t_2}}{2\sqrt{s_2}}-\sqrt{(s_1-t_2)(t_1- s_2)}=0.
\end{align*}
$\tr F^\star R^\star=0$.   Therefore
\begin{multline*}
\rT^Q_{C^Q}\big(\diag(\bs),\diag(\bt)\big)=\tr C^Q R^\star\\
= \frac{1}{2}\big(t_2+s_2+(s_1-t_2)+(t_1-s_2)\big)-\sqrt{t_2 s_2}-\sqrt{(s_1-t_2)(t_1-s_2)}.
\end{multline*}
This proves \eqref{QOPTd3char1}.

\bigskip

\noindent
(d)  Observe that the third row of every matrix in $\Gamma^{cl}(\bs,\bt)$ is a zero row.  Let $\bs'=(s_1,s_2)^\top$.  Thus $\Gamma^{cl}(\bs',\bt)$ is obtained from $\Gamma^{cl}(\bs,\bt)$ by deleting the third row in each matrix in $\Gamma^{cl}(\bs,\bt)$.    Proposition \ref{redprop} yields that 
\begin{equation*}
\rT_{C^Q}^Q(\diag(\bs),\diag(\bt))=\rT_{C^Q_{2,3}}^Q(\diag(\bs'),\diag(\bt)).
\end{equation*}
(See Lemma \ref{clqotdiagdm} for the definition of $C_{2, 3}^Q$.)   We use now the minimum characterization of $\rT_{C^Q_{2,3}}^Q(\diag(\bs'),\diag(\bt))$ given in \eqref{clqotdiagdm0}.  Assume that the minimum is achieved for $X^\star=[x_{il}^\star]\in\Gamma^{cl}(\bs',\bt), i\in[2],l\in[3]$.  We claim that either $x_{11}^\star=0$ or $x_{22}^\star=0$.

Let $Y=[x_{il}^\star], i,l\in[2]$.  Suppose first that $Y=0$.  Then $t_1=t_2=0$ and $t_3=1$.  So $\diag(\bt)$ is a rank-one matrix and $\tr \big( \diag(\bs)\diag(\bt) \big)=0$.  The equality  \eqref{QOTrankone} yields that $\rT_{C^Q}^Q \big(\diag(\bs),\diag(\bt) \big)=1$.  Clearly, $s_1\ge t_2=0, s_2\ge t_1=0$.  Hence \eqref{QOTspqutritf} holds.

Suppose second that $Y\ne 0$. Then $t_1+t_2$, the sum of the entries of $Y$, is positive.  Using continuity arguments it is enough to consider the case $t_1,t_2,t_3>0$. 
Denote by $\Gamma'$ the set of all matrices $X=[x_{il}] \in \Gamma^{cl}(\bs',\bt)$ such that $x_{i3}=x_{i3}^\star$ for $i=1,2$.   Let $f$ be defined by \eqref{deff(X)}.
Clearly $\min_{A\in\Gamma'} f(A)=f(Y)$.  We now translate this minimum to the minimum problem we studied above.

Let $Z=\frac{1}{t_1+t_2} Y$.  The vectors corresponding to the  row sums and the column sums $Z$ are the probabilty vectors $\hat \bs = (\hat s_1,\hat s_2)^\top$
 and 
$\hat \bt=\frac{1}{t_1+t_2}(t_1,t_2)^\top$ respectively.  Consider the minimum problem
$\min_{W\in \Gamma^{cl}(\hat\bs,\hat\bt)} f(W)$.  The
 proof of Lemma \ref{Pstalphaform}
 yields that this minimum is achieved at $W^\star$ which has at least one zero diagonal element.  Hence $Y$ has at least one zero diagonal element.  

Assume first that $Y$ has two zero diagonal elements.  Then $X^\star=\begin{bmatrix}0&t_2&s_1-t_2\\t_1&0&s_2-t_1\end{bmatrix}$.
This corresponds to the first case of  \eqref{QOTspqutritf}.  It is left to show that $X^\star$ is a minimizing matrix.  Using the continuity argument we may assume that  $s_1>t_2, s_2>t_1$.  Let $B\in \R^{2\times 3}$ be a nonzero matrix such that $X^\star+cB\in \Gamma^{cl}(\bs',\bt)$ for $c\in[0,\varepsilon]$ for some small positive $\varepsilon$.  Then $B=\begin{bmatrix}a&-b&-a+b\\-a&b&a-b\end{bmatrix}$, where $a,b\ge 0$ and $a^2+b^2>0$. It is clear that
$f(X^\star)<f(X+cB)$ for each $c\in (0,\varepsilon]$.  This proves the first case of \eqref{QOTspqutritf}.

Assume second that $x_{11}^\star=0$ and $x_{22}^\star>0$.  Observe that $x_{21}^\star=t_1>0$.
We claim that $x_{13}^\star=0$.  Indeed, suppose that it is not the case.  Let $B=\begin{bmatrix}0&1&-1\\0&-1&1\end{bmatrix}$.  Then $X^\star + cB\in\Gamma^{cl}(\bs',\bt)$ for $c\in[0,\varepsilon]$ for some positive $\varepsilon$.
Clearly $f(X^\star+cB)<f(X^\star)$ for $c\in(0,\varepsilon]$.  Thus contradicts the minimality of $X^\star$.  Hence $x_{13}^\star=0$.  Therefore $X^\star=\begin{bmatrix}0&s_1&0\\t_1&t_2-s_1&t_3\end{bmatrix}$.  This corresponds to the second case of \eqref{QOTspqutritf}.

The third case is when $x_{11}^\star >0$ and $x_{22}^\star=0$. We show, as in the second case, that $x_{23}^\star=0$.  Then $X^\star=\begin{bmatrix}t_1-s_2&t_2&t_3\\s_2&0&0\end{bmatrix}$.  This corresponds to the third case of \eqref{QOTspqutritf}.

The case $\bs=(s_1,s_2,s_3)^\top,\bt=(t_1,t_2,0)^\top$ is completely analogous, hence the proof is complete.
\end{proof}

Basing on the numerical studies we conjecture that the cases (a)--(d) exhaust the parameter space $\Pi_3 \times \Pi_3$. Nevertheless, we include for completeness an analysis of the quantum optimal transport $\rT^Q_{C^Q}(\diag(\bs),\diag(\bt))$ under the assumption that this is not the case. The employed techniques might prove useful when studying more general qutrit states or diagonal ququarts.

\begin{proposition}
Let $O\subset\Pi_3\times \Pi_3$ be the set of pairs $\bs,\bt$, which do not meet neither of conditions (a)--(d) from Theorem  \ref{diaggen3}. Suppose that $O$ is nonempty. Then
each minimizing $X^\star$ in the characterization \eqref{clqotdiagdm0} of $\rT^Q_{C^Q}(\diag(\bs),\diag(\bt))$ has zero diagonal. 
Let $O'\subset O$ be an open dense subset of $O$ such that for each $(\bs,\bt)\in O'$ and each triple $\{i,j,k\}=[3]$ the  inequalities $s_p\ne t_q$ and $s_p+s_q\ne t_r$ hold.  Assume that $(\bs,\bt)\in O'$.
The set of matrices in $\Gamma^{cl}(\bs,\bt)$ with zero diagonal is an interval spanned by two distinct extreme points $E_1,E_2$, which have exactly five positive off-diagonal elements.
Let $Z(u)=uE_1+(1-u)E_2$ for $u\in[0,1]$.  Then the minimum of the function $f(Z(u)),  u\in[0,1]$, where $f$ is defined by \eqref{deff(X)},  is attained at a unique point $u^\star\in(0,1)$.  The point $u^\star$ is the unique solution in the interval $(0,1)$ to a polynomial equation of degree at most $12$.  The matrix $X^\star=Z(u^\star)$ is the minimizing matrix for the second minimum problem in \eqref{clqotdiagdm0}, and $\rT_{C^{Q}}^Q(\diag(\bs),\diag(\bt))=f(X^\star)$.
\end{proposition}
\begin{proof}
Assume first that the set $O\subset \Pi_3\times \Pi_3$ is  nonempty and satisfies the conditions \emph{(i)-(iv)}.
Combine Theorem \ref{lowbdTdiagn} with part (a) of the theorem to deduce that if 
the conditions \eqref{eqlowbdiagn1} do hold for $n=3$ then 
\begin{equation}\label{negcond(a)}
\rT^Q_{C^Q}(\diag(\bs),\diag(\bt))>\max_{p\in[3]} \frac{1}{2}(\sqrt{s_p}-\sqrt{t_p})^2.
\end{equation}
In view of our assumption the above inequality holds.
We first observe that $s_p\ne t_p$ for each $p\in[3]$.  Assume to the contrary that $s_p=t_p$.  Without loss of generality we can assume that $s_3=t_3$.  Assume that in addition $s_q=t_q$ for some $q\in[2]$.  Then $\bs=\bt$ and 
\begin{equation*}
\rT_{C^{Q}}^Q(\diag(\bs),\diag(\bt))=\frac{1}{2}\max_{p\in [3]}(\sqrt{s_p}-\sqrt{t_p})^2=0
\end{equation*}
This contradicts \eqref{negcond(a)}. Hence there exists $q\in[2]$ such that $s_q>t_q$ for $q\in[2]$.  Without loss of generality we can assume that $s_2>t_2$, therefore $s_1<t_1$, as $s_1+s_2=t_1+t_2=1-s_3=1-t_3$.   Hence for $Y=\begin{bmatrix}s_1&0\\t_1-s_1&t_2\end{bmatrix}$ we have 
$X=Y\oplus [s_3]\in \Gamma^{cl}(\bs,\bt)$.  Recall that $\bs,\bt>\0$.
We replace $Y$ by $Y^\star=Y+u^\star\begin{bmatrix}-1&1\\1&-1\end{bmatrix}$ such that $u^\star>0, Y^\star\ge 0$ and one of the diagonal elements of $Y^\star$ is zero.  
By relabeling $\{1,2\}$ if necessary we can assume that $Y^\star=\begin{bmatrix}0& s_1\\t_1&t_2-s_1\end{bmatrix}$ So $t_2\ge s_1$ and $X^\star=Y^\star\oplus [s_3]\in\Gamma^{cl}(\bs,\bt)$.  The minimal characterization  \eqref {clqotdiagdm0} of $\rT^Q_{C^Q}(\diag(\bs),\diag(\bt))$ yields
\begin{eqnarray*}
\rT^Q_{C^Q}(\diag(\bs),\diag(\bt))\le f(X^\star)=\frac{1}{2}(\sqrt{s_1}-\sqrt{t_1})^2.
\end{eqnarray*}
This contradicts \eqref{negcond(a)}.  

As $\bs,\bt>\0$ there exists a maximizing matrix $F^\star$ to the dual problem of the form given by Lemma  \ref{compcondlem}.
Let $X^\star$ be the corresponding minimizing matrix. 
We claim that $X^\star$ has zero diagonal.  Assume first that $X^\star$ has a positive diagonal.  Then the arguments in part (b) of Lemma  \ref{compcondlem} yield that $X^\star$ is a symmetric matrix. Thus $\bs=\bt$, and this contradicts~\eqref{negcond(a)}.  

Assume second that $X^\star$ has two positive diagonal entries.  By renaming the indices we can assume that $x_{11}^\star =0$, $x_{22}^\star, x_{33}^\star>0$.
Part (b) of Lemma  \ref{compcondlem} and the arguments of its proof yield that we can assume that $a_2^\star=a_3^\star=b_2^\star=0$.  Let $u^\star=a_1^\star+1/2, v^\star=b_1^\star+1/2$.  As $M_{12}^\star$ is positive semidefinite we have the inequalities: $u^\star\ge 0, v^\star\ge 0, u^\star v^\star\ge 1/4$.  Hence $x^\star>0, y^\star>0$.   Recall that $F^\star$ is a maximizing matrix for the dual problem \eqref{dualQOT1}.  Hence
\begin{align*}
\rT^Q_{C^Q}\big(\diag(\bs),\diag(\bt)\big) & = -(u^\star -1/2)s_1-(v^\star-1/2)t_1 \\
& = -u^\star s_1-v^\star t_1+(s_1+t_1)/2 \\
& \le -u^\star s_1-t_1/(4u^\star) +(s_1+t_1)/2\\
& \leq -\sqrt{s_1t_1} +(s_1+t_1)/2=(\sqrt{s_1}-\sqrt{t_1})^2/2.
\end{align*}
This contradicts \eqref{negcond(a)}.  

We now assume that $X^\star$ has one positive diagonal entry.  Be renaming the indices $1,2,3$ we can assume that $x_{11}^\star=x_{22}^\star=0, x_{33}^\star>0$.
The conditions \eqref{compcond} yield that $a_3^\star+b_3^\star=0$.  Since we can choose $b_3^\star=0$ we assume that $a_3^\star=b_3^\star=0$.

Let us assume, case (A1), that $X^\star$ has six positive off-diagonal entries.  We first claim that either $x^\star_{13}=x^\star_{31}$ or $x^\star_{23}=x^\star_{32}$.  (Those are equivalent conditions if we interchange the indices $1$ and $2$.)  We deduce these conditions and an extra condition using the second conditions of \eqref{compcond1}.  First we consider $x^\star_{12}, x_{13}^\star, x_{32}^\star,x_{33}^\star$, that is $i=p=3$, $j=1, q=2$.  By replacing these entries by $x^\star_{12}-v, x_{13}^\star+v, x_{32}^\star+v,x_{33}^\star-v$
we obtain the equalities
\begin{align*}
1 +x=y+z, \qquad x=\frac{\sqrt{x_{21}^\star}}{\sqrt{x_{12}^\star}}, \quad y=\frac{\sqrt{x_{31}^\star}}{\sqrt{x_{13}^\star}}, \quad z=\frac{\sqrt{x_{23}^\star}}{\sqrt{x_{32}^\star}}.
\end{align*} 
Second we consider $x^\star_{21}, x_{23}^\star, x_{31}^\star,x_{33}^\star$.
By replacing these entries by $x^\star_{21}-v, x_{23}^\star+v, x_{31}^\star+v,x_{33}^\star-v$ we obtain the equality: 
\begin{equation*}
1+\frac{1}{x}=\frac{1}{z} +\frac{1}{y}.
\end{equation*}
Multiply the first and the second equality 
to deduce
\begin{eqnarray*}
x+\frac{1}{x}=u+\frac{1}{u} , \quad u=\frac{y}{z}\Rightarrow \textrm{ either } x=u \textrm{ or } x=\frac{1}{u}.
\end{eqnarray*}
Assume first that $x=u=y/z$. 
 Substitute that into the first equality to deduce that $z=1$, which implies that $x_{23}^\star=x_{32}^\star$.  Similarly, if $x=1/u$ we deduce that $y=1$, which implies that $x_{13}^\star=x_{31}^\star$.  Let us assume for simplicity of exposition that $x_{23}^\star=x_{32}^\star$.  Let $X(w)$ be obtained from $X^\star$ by replacing $x_{22}^\star=0,x_{23}^\star, x_{32}^\star, x_{33}^\star$ with $x_{22}^\star+w,x_{23}^\star-w, x_{32}^\star-w, x_{33}^\star+w$ for $0<w<x_{23}^\star$.  Then $X(w)$ is a minimizing matrix and has two positive diagonal entries.  This contradicts our assumption
 that $X^\star$ has only one positive diagonal entry.

We now consider the case (A2) that $x^\star_{ij}=0$ for some $i\ne j$. Part (a) of Lemma \ref{compcondlem}  yields that $x^\star_{ji}=0$.  
We claim that all four off-diagonal entries are positive.  Assume to the contrary that $x^\star_{pq}=0$ for some $p\ne q$ and $\{p,q\}\ne \{i,j\}$.  Then $x_{qp}^\star=0$.
As $\bs,\bt>\0$ we must have that $x^\star_{12}x^\star_{21}>0$ and all four other off-diagonal entries are zero.  But then $s_1=t_2, t_1=s_2, s_3=t_3$.   This is impossible since we showed that $s_3\ne t_3$.  
Hence $X^\star$ 
has exactly four positive off-diagonal entries. 

Let us assume first that  $x_{12}^\star=x_{21}^\star=0$.  Then $X^\star$ is of the form 
given by \eqref{Bopt}, where $t_3>s_1+s_2$.
We now recall again the conditions \eqref{compcond}.  As we already showed, we can assume that $a_3^\star=b_3^\star=0$.  As $x_{11}^\star=x_{22}^\star=0$ all of the first three conditions of  \eqref{compcond} hold.  As $x_{12}^\star=x_{21}^\star=0$ the second condition of  \eqref{compcond} holds trivially for $i=1,j=2$. The conditions for $i=1,j=3$ and $i=2, j=3$ are
\begin{align*}
&s_1(a_1^\star+1/2)+t_1(b_1^\star+1/2)=\sqrt{s_1t_1},\\
&s_2(a_2^\star+1/2)+t_2(b_2^\star+1/2)=\sqrt{s_2t_2}.
\end{align*}
We claim that \eqref{ab3id1} holds.  
Using the assumption that $\det M_{13}^\star\ge 1/4$ and the inequality of arithmetic and geometric means we deduce that $\det M_{13}^\star= 1/4$.  Hence 
\begin{align*}
& a_1^\star+1/2=u, \quad b_1^\star+1/2=1/(4u), \qquad \textrm{ for some }u>0, \\
& s_1u+t_1/(4u)t_1\ge \sqrt{s_1t_1}.
\end{align*}
Equality holds if and only if $u=\sqrt{t_1}/(2\sqrt{s_1})$.  This shows the first equality in 
\eqref{ab3id1}.   The second equality in \eqref{ab3id1} is deduced similarly.  We now show that the conditions \eqref{condminb} hold for $i=1,j=2,k=3$.  As $t_3>s_1+s_2$ the first condition of \eqref{ab3id1} holds.  We use the conditions that $M_{12}^\star $ is positive semidefinite.  Let $u=\sqrt{t_1}/{\sqrt{s_1}}, v=\sqrt{s_2}{\sqrt{t_2}}$.  
Then the arguments of the proof of part (b) yield
\begin{align*}
& 2(a_1^\star +b_2^\star+1)=u+v-1 >0, \quad 2(a_2^\star +b_1^\star+1)=(1/u+1/v-1)>0,\\
& 4\det M_{12}^\star =\big(1/(uv)\big)(1-u)(v-1).
\end{align*}
So either $u\ge 1$ and $v\le 1$, or $u\le 1$ and $v\ge 1$.  Hence \eqref{condminb} holds for $i=1,j=2,k=3$.  This contradicts our assumption that \eqref{condminb} does not hold. 

Let us assume second that  $x_{12}^\star>0, x_{21}^\star>0$.  Then either $x_{13}^\star=x_{31}^\star=0$ or $x_{23}^\star=x_{32}^\star=0$. By relabeling $1,2$ we can assume that $x_{23}^\star=x_{32}^\star=0$.  Hence $X^\star$ is of the form
\eqref{A2opt}, where $s_1>t_2>0, t_1>s_2>0, s_2+s_3>t_1$.  Hence the conditions
\eqref{condmin1} are satisfied with $i=1,j=2, k=3$.  We now obtain a contradiction by showing that the conditions \eqref{posa+b2} are satisfied.  This is done using the same arguments as in the previous case as follows.  
First observe that the second  nontrivial conditions of  \eqref{compcond} are:
\begin{align*}
& t_2(a_1^\star+b_2^\star +1/2)+s_2(a_2^\star +b_1^\star+1/2)=\sqrt{s_2t_2},\\
& (s_1-t_2)(a_1^\star+1/2)+t_1(b_1^\star+1/2)=\sqrt{(s_1-t_2)(t_1-s_2)}.
\end{align*}
As in the previous case we deduce that 
\begin{align*}
& a_1^\star+b_2^\star +1/2=\sqrt{s_2}/(2\sqrt{t_2}), && b_1^\star+a_2^\star +1/2=\sqrt{t_2}/(2\sqrt{s_2}),\\
& a_1^\star +1/2=\sqrt{t_1-s_2}/(2\sqrt{s_1-t_2}), && b_1^\star +1/2=\sqrt{s_1-t_2}/(2\sqrt{t_1-s_2}).
\end{align*}
 Hence \eqref{abstarceq} holds.  We now recall the proof of part (c) of the theorem. 
 We have thus shown that the minimizing matrix $X^\star$ has zero diagonal.

We now show that $O$ is an open set.  Clearly, the set of all pairs  of probability vectors  $O_1\subset \Pi_3\times \Pi_3$ such that at least one of them has a zero coordinate is a closed set. Let  $O_2, O_3, O_4\subset \Pi_3\times \Pi_3$ be the sets which satisfy the conditions (a), (b),(c) of the theorem respectively.   It it straightforward to show: $O_2$ is a closed set, and  Closure$(O_3)\subset (O_3\cup O_1)$.  We now show that Closure$(O_4)\subset O_4\cup O_1\cup O_2$.   Indeed, assume that we have a sequence $(\bs_l,\bt_l)\in O_4, l\in\N$ that converges to $(\bs,\bt)$.  It is enough to consider the case where $\bs,\bt>\0$.  Again we can assume for simplicity that each  $(\bs_l,\bt_l)$ satisfies the conditions 
\eqref{condmin1} and \eqref{posa+b2} for $i=1, j=2,k=3$.    Then we deduce that the limit of the minimizing matrices $X^\star_l$ is of the form \eqref{A2opt}.  Hence $\lim_{l\to\infty}X^\star_l=X^\star$, where $X^\star$ is of the form \eqref{A2opt}. 
Also $X^\star$ is a minimizing matrix for $\rT^Q_{C^Q}(\diag(\bs),\diag(\bt))$.  Recall that $s_2,t_2>0$.  If $s_1-t_2>0,t_1-s_2>0$ then $(\bs,\bt)\in O_4$.  So assume that $(s_1-t_2)(t_1-s_2)=0$.  As $X^\star$ is minimizes $\rT^Q_{C^Q}(\diag(\bs),\diag(\bt))$ and $\bs,\bt>\0$, part (a) of Lemma \ref{compcondlem} yields that $s_1=t_2, t_1=s_2$.  Hence $s_3=t_3$. As $X^\star$ is minimizes $\rT^Q_{C^Q}(\diag(\bs),\diag(\bt))$ we get that $\rT^Q_{C^Q}=\frac{1}{2}(\sqrt{s_2}-\sqrt{t_2})^2$.  Hence $(\bs,\bt)\in O_2$.  This shows that $O_1\cup O_2\cup O_3\cup O_4$ is a closed set.
Therefore $O=\Pi_3\times \Pi_3\setminus(O_1\cup O_2\cup O_3\cup O_4)$ is an  open set.  If $O$ is an empty set then proof of the theorem is concluded.

Assume that $O$ is a nonempty set.
Let $O'\subset O$ be an open dense subset of $O$ such that for each $(\bs,\bt)\in O'$ and each triple $\{p,q,r\}=[3]$ the  inequality $s_p\ne t_q$
and $s_p+s_q\ne t_r$ hold.  

Assume that $(\bs,\bt)\in O'$.  Let $\Gamma^{cl}_0(\bs,\bt)$ be the convex subset of $\Gamma^{cl}(\bs,\bt)$  of  matrices with zero diagonal.
We claim that any $X\in\Gamma^{cl}_0(\bs,\bt)$ has at least $5$ nonzero entries. 
Indeed, suppose that $X\in\Gamma^{cl}_0(\bs,\bt)$ has two zero 
off-diagonal entries.
As $\bs,\bt>0$ they cannot be in the same row or column. By relabeling the rows we can assume that the two zero elements are in the first and the second row. 
Suppose first that $x_{12}^\star=x_{23}^\star=0$.  Then $X=\begin{bmatrix}0&0&s_1\\
s_2&0&0\\t_1-s_2&t_2&0\end{bmatrix}$.  Thus $s_1=t_3$ which is impossible.  Assume now that $x_{12}^\star=x_{21}^\star=0$.  Then $s_1+s_2=t_3$ which is impossible.  All other choices also are impossible.  

We claim that $ \Gamma^{cl}_0(\bs,\bt)$ is spanned by two distinct extreme points $E_1,E_2$, which have exactly five positive off-diagonal elements.  Suppose first that there exists $X\in\Gamma^{cl}_0(\bs,\bt)$ which has six postive off-diagonal elements.
Let 
\begin{equation*}
B=\begin{bmatrix}0&1&-1\\-1&0&1\\1&-1&0\end{bmatrix}.
\end{equation*}
Then all matrices in $\Gamma^{cl}_0(\bs,\bt)$ are of the form  $X^\star+uB, u\in [u_1,u_2]$ for some $u_1<u_2$.  Consider the matrix $E_1=X^\star +u_1B$.  It has at least one zero off-diagonal entry
 hence we conclude 
that $E_1$ has exactly five off-diagonal positive elements.  Similarly $E_2=X+u_2B$ has five positive off-diagonal elements.  Assume now that $E\in\Gamma^{cl}_0(\bs,\bt)$ has 
five positive off-diagonal elements.
Hence there exits a small $u>0$ such that either $E+uB$ or $E-uB$ has six positive off-diagonal elements.  Hence $ \Gamma^{cl}_0(\bs,\bt)$ contains a matrix with six positive diagonal elements.  Therefore $ \Gamma^{cl}_0(\bs,\bt)$ is an interval spanned by $E_1\ne E_2\in  \Gamma^{cl}_0(\bs,\bt)$, where $E_1$ and $E_2$ have five positive off-diagonal elements.  Part (a) of Lemma \ref{compcondlem} yields that $X^\star$ has six positive off-diagonal elements.  Consider $E_1$ and assume that the $(1,2)$ entry of $E_1$ is zero.  Then
\begin{equation*}
E_1=\begin{bmatrix}0&0&s_1\\s_1+s_2-t_3&0&t_3-s_1\\s_3-t_2&t_2&0 \end{bmatrix}.
\end{equation*}

As $f(E_1+uB)$ is strictly convex on $[0,u_3]$, there exists a unique $u^\star\in (0, u_3)$ which satisfies the equation 
\begin{multline*}
-\frac{\sqrt{s_1+s_2-t_3-u}}{\sqrt{u}}+\frac{\sqrt{u}}{\sqrt{s_1+s_2-t_3-u}} -\frac{\sqrt{s_1-u}}{\sqrt{s_3-t_2+u}} \\
 +\frac{\sqrt{s_3-t_2+u}}{\sqrt{s_1-u}}-
\frac{\sqrt{t_2-u}}{\sqrt{t_3-s_1+u}}+\frac{\sqrt{t_3-s_1+u}}{\sqrt{t_2-u}}=0.
\end{multline*}
It is not difficult to show that the above equation is equivalent to a polynomial equation of degree at most $12$ in $u$.  Indeed, group the six terms into three groups, multiply by the common denominator, and pass the last group to the other side of the equality to obtain the equality:
\begin{align*}
&\sqrt{(s_1-u)(s_3-t_2+u)(t_3-s_1+u)(t_2-u)}(2u+t_3-s_1-s_2)\\
& \, +\sqrt{u(s_1+s_2-t_3-u)(t_3-s_1+u)(t_2-u)}(2u+t_3-s_1-s_2)(2u+s_3-s_1-t_2)\\
& \hspace*{1.85cm} =\sqrt{u(s_1+s_2-t_3-u)(s_3-t_2+u)(t_3-s_1+u)}(-2u+s_1+t_2-t_3).
\end{align*}
Raise this equality to the second power.  Put all polynomial terms of degree $6$ on the left hand side, and the one term with a square radical on the other side.
Raise to the second power to obtain a polynomial equation in $u$ of degree at most 12.   Hence $X^\star=E_1+u^\star B$.  This completes the proof of (e).
\end{proof}


%

  \bigskip

\bibliographystyle{plain}

\begin{thebibliography}{MMM}

\bibitem{AF17} J. Agredo and  F. Fagnola,
On quantum versions of the classical Wasserstein distance, 
{\sl Stochastics} {\bf 89} (2017), 910. 

\bibitem{AWR17} J. Altschuler, J.  Weed and P. Rigollet, Near-linear time approximation algorithms for optimal transport via Sinkhorn iteration,
\textit{NIPS'17: Proceedings of the 31st International Conference on Neural Information Processing Systems}, December 2017, pages 1961--1971.

\bibitem{ACB17} M. Arjovsky, S. Chintala and L. Bottou, Wasserstein Generative Adversarial Networks, \textit{Proceedings of the 34th International Conference on Machine Learning}, PMLR {\bf 70} (2017), 214. 

\bibitem{BZ17} I. Bengtsson and K. {\.Z}yczkowski, 
{\sl  Geometry of Quantum States}. 2 ed., Cambridge University Press, 
Cambridge 2017.

\bibitem{BGJ19} R. Bhatia, S. Gaubert and T. Jain, 
Matrix versions of the Hellinger distance,
\textit{Lett. Math. Phys.} {\bf 109} (2019), 1777--1804.

\bibitem{QML} J. Biamonte, P. Wittek, N. Pancotti, P. Rebentrost, N. Wiebe, N., and S. Lloyd, Quantum machine learning, {\sl Nature} \textbf{549} (2017), 195.

\bibitem{BV01} P. Biane and D. Voiculescu,
 A free probability analogue of the Wasserstein distance on the trace-state space,
{\sl Geom. Funct. Anal.} {\bf 11} (2001), 1125.

\bibitem{BGKL17} J. Bigot, R. Gouet, T. Klein, and A. L\'opez, Geodesic PCA in the Wasserstein space by convex PCA. 
\emph{Ann. Inst. H. Poincaré Probab. Statist.} {\bf 53} 
(2017), 1-26.

\bibitem{BEZ22} R. Bistro\'n, M. Eckstein, K. \.Zyczkowski, Monotonicity of the quantum 2-Wasserstein distance, 
{\sl preprint arXiv}:2204.07405 (2022).

\bibitem{BPPH11} N. Bonneel, M. van de Panne, S. Paris, and W. Heidrich. Displacement interpolation using Lagrangian mass transport.
 \emph{ACM Trans. Graph.} {\bf 30} 
 (2011), 158.

\bibitem{brandao2017quantum}
F.~G. S. L Brand\~{a}o and K. Svore, Quantum speed-ups for solving semidefinite programs, {\em 2017 IEEE 58th Annual
Symposium on Foundations of Computer Science (FOCS)}, pages 415--426.

\bibitem{BC94} S. L. Braunstein  and C. M.  Caves,
 Statistical distance and the geometry of quantum states, 
 {\sl Phys. Rev. Lett.} {\bf 72} (1994), 3439.
 
\bibitem{CGP20} E. Caglioti, F. Golse, and T. Paul,
 Quantum optimal transport is cheaper,
  {\sl J. Stat. Phys.}, {\bf 181} (2020), 149.

\bibitem{CM20} E. A. Carlen and J. Maas,
Non-commutative calculus, optimal transport and
functional inequalities in dissipative quantum systems,
{\sl  J. Stat. Phys.} {\bf 178} (2020), 319.

\bibitem{CHW19}  S. Chakrabarti, Y. Huang, T. Li, S. Feizi, and X. Wu, Quantum Wasserstein Generative Adversarial Networks, \textit{33rd Conference on Neural Information Processing Systems (NeurIPS 2019)}, Vancouver, Canada, arXiv:1911.00111.

\bibitem{CSVCC21} M. Cerezo, A. Sone, T. Volkoff, L. Cincio, and P. J. Coles, Cost-function-dependent barren plateaus in shallow quantum neural networks, {\sl Nat. Comm.} \textbf{12} (2021), 1791.

\bibitem{CGGT17} Y. Chen, W. Gangbo, T. T. Georgiou, and A. Tannenbaum,
On the matrix Monge-Kantorovich problem,
{\sl  Eur. J. Appl. Math.}  {\bf 31} (2020), 574.

 
 
\bibitem{CEFZ21} S. Cole, M.  Eckstein, S. Friedland and K. {\.Z}yczkowski, Quantum Optimal Transport, arXiv:2105.06922, May, 2021.

 \bibitem{CCPS}  W.J. Cook, W.H. Cunningham, W.R. Pulleyblank, A.Schrijver,
 \emph{Combinatorial Optimization}, Wiley, 1998.
 
\bibitem{Cut13}  M. Cuturi, Sinkhorn distances: Lightspeed computation of optimal transport, In C. J. C. Burges, L. Bottou, M. Welling, Z. Ghahramani, and K. Q. Weinberger, editors, Advances in Neural
Information Processing Systems 26, pages 2292-2300. Curran Associates, Inc., 2013.

\bibitem{DdK18} P.-L. Dallaire-Demers, and N. Killoran, Quantum generative adversarial networks, {\sl Phys. Rev. A} \textbf{98} (2018), 0122324.

\bibitem{DR20} N. Datta and C. Rouz{\'e},
Relating relative entropy, optimal transport and
Fisher information: A quantum HWI inequality. 
\emph{Ann. H. Poincar{\'e}} {\bf 21}  (2020), 2115.

\bibitem{PMTL20} G.  De Palma, M. Marvian, D. Trevisan and S. Lloyd, The Quantum Wasserstein Distance of Order 1, \textit{IEEE Transactions on Information Theory} (2021), doi: 10.1109/TIT.2021.3076442.

\bibitem{PT21} G.  De Palma and D. Trevisan,  Quantum Optimal Transport with Quantum
Channels,  Ann. Henri Poincar\'e {\bf 22} (2021), 3199-3234.

\bibitem {Duv20} R. Duvenhage, Quadratic Wasserstein metrics for von Neumann algebras via transport plans, arXiv:2012.03564;
  J. Operator Theory {\sl (to appear)}

 
\bibitem{FKV18} K. Filipiak, D. Klein and  E. Vojtkov{\'a},
 The properties of partial trace and block trace operators of partitioned matrices,
\emph{E. J. Linear Algebra} {\bf 33} (2018), 3--15.
 
\bibitem{FCCR18} R. Flamary, M. Cuturi, N. Courty, A. Rakotomamonjy, 
Wasserstein Discriminant Analysis, 
\emph{Machine Learning} {\bf 107} (2018), 1923--1945.

\bibitem{Frb16} S. Friedland, \emph{Matrices: Algebra, Analysis and Applications}, World Scientific, 596 pp., 2016, Singapore.


\bibitem{FriSDP} S. Friedland, Notes on semidefinite programming, Fall 2017,  http://homepages.math.uic.edu/~friedlan/SDPNov17.pdf

 \bibitem{Fri20} S. Friedland,  Tensor optimal transport, distance between sets of measures and tensor scaling, arXiv:2005.00945.
 
 \bibitem{FECZ} S. Friedland, M. Eckstein, S. Cole and K. \.Zyczkowski, Quantum Monge-Kantorovich problem and transport distance between density matrices, 
 arXiv:2102.07787, 2021.
 
 \bibitem{FGZ19} S. Friedland, J. Ge and L. Zhi, Quantum Strassen's theorem,   \emph{Infinite Dimensional Analysis, Quantum Probability and Related Topics}, 
{\bf 23} 
 (2020), 2050020 (29 pages).
 
 \bibitem{FrV18}  G. Friesecke and D.  V\"ogler,  Breaking the curse of dimension in multi-marginal Kantorovich optimal transport on finite state spaces, \emph{SIAM J. Math. Anal.} {\bf 50} (2018), no. 4, 3996-4019.




\bibitem{GLN05} A. Gilchrist,  N.K. Langford, and M.A. Nielsen, Distance measures to compare real and ideal quantum processes, 
\emph{Phys. Rev.} {\bf A71} (2005), 062310.

\bibitem{GMP16} F.  Golse,  C. Mouhot,  and T. Paul,  On the mean field and classical limits of quantum mechanics, \emph{Comm. Math. Phys} {\bf 343} (2016), 
 165-205.

\bibitem{GP18} F. Golse and T. Paul,
 Wave packets and the quadratic Monge-Kantorovich distance in quantum mechanics,
 {\sl  Comptes Rendus Math.} {\bf 356} (2018),  
 177-197.


\bibitem{Hit41} F.L. Hitchcock,  The distribution of a product from several sources to numerous localities, \emph{J. Math. Phys. Mass. Inst. Tech.} {\bf 20} (1941), 224--230.

\bibitem{HJ13} R.A.Horn and C.R. Johnson, 
\emph{Matrix analysis},
Second edition, Cambridge University Press, Cambridge, 2013. 

\bibitem{Hor96} M. Horodecki, P. Horodecki and R. Horodecki, Separability of mixed states: necessary and sufficient conditions,  \emph{Physics Letters} {\bf A 223} 
 (1996), 1--8.
 
\bibitem{Ikeda20}
K. Ikeda, Foundation of quantum optimal transport and applications,
\emph{Quantum Inform. Process.} {\bf 19} (2020), 25.

\bibitem{Jo94} R.~Jozsa,
Fidelity for mixed quantum states,
\emph{J.\ Mod.\ Opt.\ }\textbf{41} (1994), 2315--23.

 

\bibitem{Kan42} L.V. Kantorovich. 
On the translocation of masses. 
{\sl  Dokl. Akad. Nauk. USSR} {\bf 37} (1942), 199.

\bibitem{Kan48} L.V. Kantorovich, On a problem of Monge,
 {\sl Uspekhi Mat. Nauk.} {\bf 3} (1948), 225.

\bibitem{Kan60}   L.V. Kantorovich, Mathematical methods of organizing and planning production,\emph{Management Sci.} 6 (1959/60), 366--422.

\bibitem{KdPMLL21} B.T. Kiani, G. De Palma, M. Marvian, Z.-W. Liu, and S. Lloyd,
Learning quantum data with the quantum earth mover's distance, 
{\sl Quantum Sci. Technol.} \textbf{7} (2022), 045002.
 
 \bibitem{LYLW20} J. Liu, H. Yuan, X.-M. Lu, and X. Wang,
Quantum Fisher information matrix and multiparameter estimation,
{\sl J. Phys.} {\bf  A53} (2020), 023001.
 
 \bibitem{LG15} J. R. Lloyd and Z. Ghahramani, Statistical model criticism using kernel two sample tests, In \textit{Proceedings of the 28th International Conference on Neural Information Processing Systems, NIPS'15}, pages 829--837, Cambridge, MA, USA, 2015. MIT Press.

\bibitem{LW18} S. Lloyd, and C. Weedbrook, Quantum Generative Adversarial Learning, {\sl Phys. Rev. Lett.} \textbf{121}  (2018), 040502.

\bibitem{MBSBN18} J. R. McClean, S. Boixo, V. N. Smelyanskiy, R. Babbush, and H. Neven, Barren plateaus in quantum neural network training landscapes, {\sl Nat. Comm.} \textbf{9} (2018), 4812.

\bibitem{MPHUZ} J. A. Miszczak, Z. Puchala, P. Horodecki, A. Uhlmann, 
K.~{\.Z}yczkowski, Sub-- and super--fidelity as
bounds for quantum fidelity, 
\emph{Quantum Inf. Comp.} 9 (2009), 0103--0130. 

\bibitem{Mon81} G. Monge. M\'emoire sur la th\'eorie des d\'eblais et des remblais, Histoire de l’Acad\'emie Royale des Sciences de Paris, avec les M\'emoires de Math\'ematique et de Physique pour la m\^eme ann\'ee (1781), pages 666--704.
 
 \bibitem{MJ15} J. Mueller and T. Jaakkola, Principal differences analysis: Interpretable characterization of differences between distributions, In \textit{Proceedings of the 28th International Conference on Neural Information Processing Systems, NIPS'15}, pages 1702--1710, Cambridge, MA, USA, 2015. MIT Press.

 
 \bibitem{PZ16} V. M. Panaretos and Y. Zemel, Amplitude and phase variation of point processes, \textit{Ann. Statist.} {\bf 44} (2016), 771--812.
 

\bibitem{PT11} A. N. Pechen, and D.J. Tannor, Are there traps in quantum control
 landscapes?, {\sl Phys. Rev. Lett} \textbf{106} (2011), 120402.

\bibitem{Per96} A. Peres, Separability criterion for density matrices, 
\emph{Phys. Rev. Lett.} {\bf 77} (1996), 1413--1415. 

\bibitem{Ren} R. Renner, Quantum Information Theory, Exercise Sheet 9, http://edu.itp.phys.ethz.ch/hs15/QIT/ex09.pdf

\bibitem{Rie18} M.H. Riera, A transport approach to distances
in quantum systems, Bachelor's thesis for the degree in
Physics, Universitat Aut\`onoma de Barcelona, 2018.
 
 \bibitem{RTG00} Y. Rubner, C. Tomasi, and L. J. Guibas, The earth mover’s distance as a distance for image retrieval,
\emph{Int. J. Comput. Vision}, {\bf 40} 
 (2000), 99-121.
 

\bibitem{SGPCBNDG} J. Solomon, F. de Goes, G. Peyré, M. Cuturi, A. Butscher, A. Nguyen, T. Du, and L. Guibas,
Convolutional Wasserstein distances: Efficient optimal transportation on geometric domains, 
\emph{ACM Trans. Graph.} {\bf 34} 
 (2015), 66.
 
 \bibitem{SL11} R. Sandler and M. Lindenbaum, 
Nonnegative matrix factorization with earth mover’s distance distance for image analysis, 
\emph{IEEE Trans. Pattern Anal. Mach. Intell.}
{\bf 33} 
 (2011), 1590--1602.
 
 \bibitem{Sa17} D. {\v S}afr{\'a}nek,
Discontinuities of the quantum Fisher information and the Bures distance,
{\sl Phys. Rev.} {\bf  A 95} (2017), 052320.

\bibitem{SR04} G. J. Sz\'ekely and M. L. Rizzo. Testing for equal distributions in high dimension, \emph{Inter-Stat. (London)} {\bf 11} (2004), 1--16.
 
\bibitem{Uh76} A. Uhlmann,
The `transition probability' in the state space of a *-algebra,
{\sl Rep. Math. Phys.} {\bf 9} (1976), 273.


\bibitem{VB96} L. Vandenberghe and S. Boyd, Semidefinite programming, \textit{SIAM Rev.} {\bf 38}  (1996),
 49--95.

\bibitem{Vas69}  L.N. Vasershtein, Markov processes over denumerable products of spaces describing large system of automata, \emph{Problems Inform. Transmission}
{\bf  5} (1969), 
47--52. 


\bibitem{Vil09} C. Villani, {\sl Optimal transport, Old and new}, Grundlehren der Mathematischen Wissenschaften, 338, Springer-Verlag, Berlin, 2009.

\bibitem{WFCSSCC20} S. Wang, E. Fontana, M. Cerezo, K. Sharma, A. Sone, L. Cincio, and P.J. Coles, Noise-induced barren plateaus in variational quantum algorithms, {\sl Nature Communications} \textbf{12} (2021), 6961. 

\bibitem{WWW21} Z. Wang, Y. Wang and  B. Wu,
Genuine quantum chaos and physical distance between quantum states,
{\sl Phys. Rev.} {\bf  E 103} (2021), 042209.

\bibitem{Wer89} R.F. Werner, Quantum states with Einstein-Podolsky-Rosen correlations admitting a hidden-variable mode, 
\emph{Phys. Rev.} {\bf A 40} (1989), 4277.

\bibitem{Win16}  A. Winter, Tight uniform continuity bounds for quantum entropies: conditional entropy, relative entropy distance and energy constraints, 
\emph{Comm. Math. Phys.} {\bf 347} (2016),  
 291--313.


\bibitem{Mathematica} Wolfram Research{,} Inc., Mathematica{,} Version 12.2, Champaign, IL, USA, 2020, \url{https://www.wolfram.com/mathematica}

\bibitem{YZYY19} N.Yu, L. Zhou, S. Ying, M. Ying, Quantum Earth mover's distance,
No-go Quantum Kantorovich-Rubinstein theorem, and Quantum Marginal
Problem, arXiv:1803.02673.


 \bibitem {ZS98} K. {\.Z}yczkowski and  W. S{\l}omczy{\'n}ski,
 Monge distance between quantum states,
  {\sl J.Phys.} {  A 31}  (1998),  9095--9104. 

\bibitem{ZS01} K. {\.Z}yczkowski and W. S{\l}omczy{\'n}ski,
The Monge distance on the sphere and geometry of quantum states,
{\sl  J. Phys.} {A 34} (2001), 6689.

\end{thebibliography}

\end{document}